%% file: main.tex
\pdfoutput=1
\documentclass[doctor,english,final]{postech-ucs}

\usepackage{amsmath}
\usepackage{enumitem}
\usepackage{algpseudocode}
\usepackage{comment}
\usepackage{physics}
\usepackage{anyfontsize}
\usepackage{lipsum}

\usepackage[ruled,vlined,linesnumbered]{algorithm2e}
\usepackage{multirow}
\usepackage{kotex}
\usepackage{mathtools}
\usepackage{enumitem}
\setlist[itemize]{leftmargin=3mm}
\usepackage[skip=2pt]{caption} 
\usepackage{stackengine}

\makeatletter
\newcommand*{\rom}[1]{\expandafter\@slowromancap\romannumeral #1@}
\makeatother
\usepackage{amssymb}
\usepackage{amsthm}
\usepackage{mathtools}
\usepackage{booktabs}
\newtheorem{proposition}{Proposition}
\numberwithin{proposition}{chapter}
\usepackage{rotating} 
\usepackage{float}
\usepackage{comment}
\usepackage{subfigure}
\usepackage{mwe}
\usepackage{multirow}
\usepackage{newfloat}
\usepackage{listings}
\floatstyle{ruled}
\newfloat{listing}{tb}{lst}{}
\floatname{listing}{Listing}

\theoremstyle{definition}
\newtheorem{definition}{Definition}

\DeclareMathOperator*{\argmax}{argmax}
\DeclareMathOperator*{\argsort}{argsort}
\usepackage{bbm}

\title[korean]{추천시스템을 위한 신뢰도 보정 및 활용}
\title[english]{Confidence Calibration for Recommender Systems and Its Applications}

\author[korean]{권}{원 빈}
\author[english]{Kweon}{Wonbin}

\advisor[major]{유 환 조}{Hwanjo Yu}{nosign}

\department{CITE}{}

\studentid{20192689}

\referee[1]{Hwanjo Yu}
\referee[2]{Young Myoung Ko}
\referee[3]{Gary Geunbae Lee}
\referee[4]{Dongwoo Kim}
\referee[5]{Namhoon Lee}

\approvaldate{2023}{12}{07}

\refereedate{2023}{12}{07}

\gradyear{2024}

\begin{document}


    \begin{abstract}
    \input{0.abst_eng}
    \end{abstract}

    \tableofcontents

    \listoftables

    \listoffigures



\chapter{Introduction}
\input{1.intro}

\chapter{Obtaining Calibrated Probabilities with Personalized Ranking Models}
\input{2.Cali}

\chapter{Bidirectional Distillation for Top-K Recommender System}
\input{3.BD}

\chapter{Top-Personalized-K Recommendation}
\input{4.PerK}

\chapter{Conclusion}
\input{5.Con}

\begin{summarykorean}
\input{6.abst_kor}
\end{summarykorean}

\bibliographystyle{unsrt}
\bibliography{mybib}

\acknowledgement[korean]
\input{7.ack}

\input{8.cv}

\end{document}

%% file: 0.abst_eng.tex
Personalized recommendations have a significant impact on various daily activities such as shopping, news, search, videos, and music.
However, most recommender systems only display the top-scored items to the user without providing any indication of confidence in the recommendation results.
The semantic of the same ranking position differs for each user; one user might like his third item with a probability of 30\%, whereas the other user may like her third item with 90\%.

Despite the importance of having a measure of confidence in recommendation results, it has been surprisingly overlooked in the literature compared to the accuracy of the recommendation.
In this dissertation, I propose a model calibration framework for recommender systems for estimating accurate confidence in recommendation results based on the learned ranking scores.
Moreover, I subsequently introduce two real-world applications of confidence on recommendations:
(1) Training a small student model by treating the confidence of a big teacher model as additional learning guidance,
(2) Adjusting the number of presented items based on the expected user utility estimated with calibrated probability.

\textbf{Obtaining Calibrated Confidence.}
I investigate various parametric distributions and propose two parametric calibration methods, namely Gaussian calibration and Gamma calibration.
Each proposed method can be seen as a post-processing function that maps the ranking scores of pre-trained models to well-calibrated preference probabilities, without affecting the recommendation performance.

\textbf{Bidirectional Distillation.}
I propose Bidirectional Distillation (BD) framework whereby both the teacher and the student collaboratively improve with each other.
Specifically, each model is trained with the distillation loss that makes it follow the other’s prediction confidence along with its original loss function.
Trained in a bidirectional way, it turns out that both the teacher and the student are significantly improved compared to when being trained separately.

\textbf{Top-Personalized-K Recommendation.}
I introduce Top-Personalized-$K$ Recommendation (PerK), a new recommendation task aimed at generating a personalized-sized ranking list to maximize individual user satisfaction.
PerK estimates the expected user utility by leveraging calibrated interaction probabilities, subsequently selecting the recommendation size that maximizes this expected utility.
We expect that Top-Personalized-$K$ recommendation has the potential to offer enhanced solutions for various real-world recommendation scenarios, based on its great compatibility with existing models.

%% file: 1.intro.tex
Personalized recommendations have a significant impact on various daily activities such as shopping, news, search, videos, and music.
With an enormous number of items in the system, users cannot possibly look up the entire item set, and therefore they need to rely on the recommendations provided by the system.
However, most recommender systems only display the top-scored items to the user without providing any indication of confidence in the recommendation results.
The semantic of the same ranking position differs for each user; one user might like his third item with a probability of 30\%, whereas the other user may like her third item with 90\%.
Also, recommendation accuracy differs for each user based on their activeness and preference complexity.
Therefore, I claim that real-world applications need to indicate their confidence in their recommendation results to provide reliability and interpretability to users.

Despite the importance of having a measure of confidence in recommendation results, it has been surprisingly overlooked in the literature compared to the accuracy of the recommendation.
Recommender models aim to learn the ranking score for each user-item pair, so as to produce a ranking list of items for each user.
This ranking score cannot be used as a measure of probability or confidence as it is usually an unbounded real number used for the item sorting process.
In this dissertation, I propose a model calibration framework for recommender systems for estimating accurate confidence in recommendation results based on the learned ranking scores.
Estimated confidence can be presented along with the recommendation results to increase reliability and interoperability.
Moreover, I subsequently introduce two real-world applications of confidence on recommendations:
(1) Training a small student model by treating the confidence of a big teacher model as additional learning guidance,
(2) Adjusting the number of presented items based on the expected user utility estimated with calibrated probability.

\textbf{(Section \rom{2}) Obtaining Calibrated Probabilities with Personalized Ranking Models.}
In this paper, we aim to estimate the calibrated probability of how likely a user will prefer an item.
We investigate various parametric distributions and propose two parametric calibration methods, namely Gaussian calibration and Gamma calibration.
Each proposed method can be seen as a post-processing function that maps the ranking scores of pre-trained models to well-calibrated preference probabilities, without affecting the recommendation performance.
We also design the unbiased empirical risk minimization framework that guides the calibration methods to the learning of true preference probability from the biased user-item interaction dataset.
Extensive evaluations with various personalized ranking models on real-world datasets show that both the proposed calibration methods and the unbiased empirical risk minimization significantly improve the calibration performance.

\textbf{(Section \rom{3}) Bidirectional Distillation for Top-K Recommendation.}
Recommender systems (RS) have started to employ knowledge distillation, which is a model compression technique training a compact model (student) with the knowledge transferred from a cumbersome model (teacher).
The state-of-the-art methods rely on unidirectional distillation transferring the knowledge only from the teacher to the student, with an underlying assumption that the teacher is always superior to the student.
However, we demonstrate that the student performs better than the teacher on a significant proportion of the test set, especially for RS.
Based on this observation, we propose \textit{Bidirectional Distillation} (BD) framework whereby both the teacher and the student collaboratively improve with each other.
Specifically, each model is trained with the distillation loss that makes it follow the other’s prediction confidence along with its original loss function.
For effective bidirectional distillation, we propose \textit{rank discrepancy-aware sampling} scheme to distill only the informative knowledge that can fully enhance each other.
The proposed scheme is designed to effectively cope with a large performance gap between the teacher and the student.
Trained in the bidirectional way, it turns out that both the teacher and the student are significantly improved compared to when being trained separately.
Our extensive experiments on real-world datasets show that our proposed framework consistently outperforms the state-of-the-art competitors.
We also provide analyses for an in-depth understanding of BD and ablation studies to verify the effectiveness of each proposed component.

\textbf{(Section \rom{4}) Top-Personalized-K Recommendation.}
The conventional top-\textit{K} recommendation, which presents the top-\textit{K} items with the highest ranking scores, is a common practice for generating personalized ranking lists.
However, is this fixed-size top-$K$ recommendation the optimal approach for every user’s satisfaction?
Not necessarily.
We point out that providing fixed-size recommendations without taking into account user utility can be suboptimal, as it may unavoidably include irrelevant items or limit the exposure to relevant ones.
To address this issue, we introduce Top-Personalized-$K$ Recommendation, a new recommendation task aimed at generating a personalized-sized ranking list to maximize individual user satisfaction.
As a solution to the proposed task, we develop a model-agnostic framework named PerK.
PerK estimates the expected user utility by leveraging calibrated interaction probabilities, subsequently selecting the recommendation size that maximizes this expected utility.
Through extensive experiments on real-world datasets, we demonstrate the superiority of PerK in Top-Personalized-$K$ recommendation task.
We expect that Top-Personalized-$K$ recommendation has the potential to offer enhanced solutions for various real-world recommendation scenarios, based on its great compatibility with existing models.

%% file: 2.Cali.tex
For personalized ranking models, the well-calibrated probability of an item being preferred by a user has great practical value.
While existing work shows promising results in image classification, probability calibration has not been much explored for personalized ranking.
In this paper, we aim to estimate the calibrated probability of how likely a user will prefer an item.
We investigate various parametric distributions and propose two parametric calibration methods, namely Gaussian calibration and Gamma calibration.
Each proposed method can be seen as a post-processing function that maps the ranking scores of pre-trained models to well-calibrated preference probabilities, without affecting the recommendation performance.
We also design the unbiased empirical risk minimization framework that guides the calibration methods to the learning of true preference probability from the biased user-item interaction dataset.
Extensive evaluations with various personalized ranking models on real-world datasets show that both the proposed calibration methods and the unbiased empirical risk minimization significantly improve the calibration performance.
This work was published at AAAI Conference on Artificial Intelligence (AAAI 2022) with an oral presentation \cite{kweon22}.

\section{Introduction}
Personalized ranking models aim to learn the ranking scores of items, so as to produce a ranked list of them for the recommendation \cite{bpr09}.
However, their prediction results provide an incomplete estimation of the user's potential preference for each item; the semantic of the same ranking position differs for each user.
One user might like his third item with the probability of 30\%, whereas the other user likes her third item with 90\%.
Accurately estimating the \textit{probability} of an item being preferred by a user has great practical value \cite{iso12}.
The preference probability can help the user choose the items with high potential preference and the system can raise user satisfaction by pruning the ranked list by filtering out items with low confidence \cite{pruning09}.
To ensure reliability, the predicted probabilities need to be \textit{calibrated} so that they can accurately indicate their ground truth correctness likelihood.
In this paper, our goal is to obtain the well-calibrated probability of an item matching a user's preference based on the ranking score of the pre-trained model, without affecting the ranking performance.

While recent methods \cite{cal17, kull2019beyond, rahimi2020intra} have successfully achieved model calibration for image classification, it has remained a long-standing problem for personalized ranking.
A pioneering work \cite{iso12} firstly proposed to predict calibrated probabilities from the scores of pre-trained ranking models by using isotonic regression \cite{iso72}, which is a simple non-parametric method that fits a monotonically increasing function.
Although it has shown some effectiveness, there is no subsequent study about \textit{parametric} calibration methods in the field of personalized ranking despite their richer expressiveness than non-parametric methods.

In this paper, we investigate various parametric distributions, and from which we propose two calibration methods that can best model the score distributions of the ranking models.
First, we define three desiderata that a calibration function for ranking models should meet, and show that existing calibration methods have the insufficient capability to model the diverse populations of the ranking score.
We then propose two parametric methods, namely Gaussian calibration and Gamma calibration, that satisfy all the desiderata.
We demonstrate that the proposed methods have a larger expressive power in terms of the parametric family and also effectively handles the imbalanced nature of ranking score populations compared to the existing methods \cite{platt99, cal17}.
Our methods are post-processing functions with \textit{three} learnable parameters that map the ranking scores of pre-trained models to calibrated posterior probabilities.

To optimize the parameters of the calibration functions, we can use the log-loss on the held-out validation sets \cite{cal17}.
The challenge here is that the user-item interaction datasets are implicit and missing-not-at-random \cite{ips16, saito19}.
For each user-item pair, the label is 1 if the interaction is observed, 0 otherwise.
An unobserved interaction, however, does not necessarily mean a negative preference, but the item might have not been exposed to the user yet.
Therefore, if we fit the calibration function with the log-loss computed naively on the implicit datasets, the mapped probabilities may indicate biased likelihoods of users' preference on items.
To tackle this problem, we design an unbiased empirical risk minimization framework by adopting Inverse Propensity Scoring \cite{ips94}.
We first decompose the interaction variable into two variables for observation and preference, and adopt an inverse propensity-scored log-loss that guides the calibration functions toward the true preference probability.

Extensive evaluations with various personalized ranking models on real-world datasets show that the proposed calibration methods produce more accurate probabilities than existing methods in terms of calibration measures like ECE, MCE, and NLL.
Our unbiased empirical risk minimization framework successfully estimates the ideal empirical risk, leading to performance gain over the naive log-loss.
Furthermore, reliability diagrams show that Gaussian calibration and Gamma calibration predict well-calibrated probabilities across all probability ranges.
Lastly, we provide an in-depth analysis that supports the superiority of the proposed methods over the existing methods.

\section{Preliminary \& Related Work}
\subsection{Personalized Ranking}
Let $\mathcal{U}$ and $\mathcal{I}$ denote the user space and the item space, respectively.
For each user-item pair $(u,i)$ of $u \in \mathcal{U}$ and $i \in \mathcal{I}$, a label $Y_{u,i}$ is given as 1 if their interaction is observed and 0 otherwise.
It is worth noting that unobserved interaction ($Y_{u,i}=0$) may indicate the negative preference or the unawareness, or both.
A personalized ranking model $f_\theta: \mathcal{U} \times \mathcal{I} \rightarrow \mathbb{R}$ learns the ranking scores of user-item pairs to produce a ranked list of items for each user.
$f_{\theta}$ is mostly trained with pairwise loss that makes the model put a higher score on the observed pair than the unobserved pair:
\begin{equation}
    \mathcal{L}_{\textup{pair}} = \sum_{u \in \mathcal{U}, i,j \in \mathcal{I}} \ell(f_\theta(u,i), f_\theta(u,j)) Y_{u,i} (1-Y_{u,j}),
\end{equation}
where $\ell(\cdot, \cdot)$ is some convex loss function such as BPR loss \cite{bpr09} or Margin Ranking loss \cite{margin07}.
Note that the ranking score $f_\theta(u,i) \in \mathbb{R}$ is not bounded in $[0,1]$ and therefore cannot be used as a probability.

\subsection{Calibrated Probability}
To estimate $P(Y_{u,i}=1 | f_\theta(u,i))$, which is the probability of item $i$ being interacted with user $u$ given the pre-trained ranking score, we need a post-processing calibration function $g_{\phi}(s)$ that maps the ranking score $s=f_{\theta}(u,i)$ to the calibrated probability $p$.
Here, the calibration function for the personalized ranking has to meet the following \textbf{desiderata}:
(1) the function $g_\phi: \mathbb{R} \rightarrow [0,1]$ needs to take an input from the \textit{unbounded range} of the ranking score to output a probability;
(2) the function should be \textit{monotonically increasing} so that the item with a higher ranking score gets a higher preference probability;
(3) the function needs enough expressiveness to represent diverse score distributions with \textit{asymmetricity}.

We say the probability $p$ is well-calibrated if it indicates the ground-truth correctness likelihood \cite{beta17}:
\begin{equation}
    \mathbb{E} [ Y | g_\phi(s) = p ] = p, \;\;\; \forall p \in [0,1].
\end{equation}
\begin{equation}
    g_\phi(f_\theta(u,i)) = p.
\end{equation}
For example, if we have 100 predictions with $p=0.3$, we expect 30 of them to indeed have $Y=1$ when the probabilities are calibrated.
Using this definition, we can measure the miscalibration of a model with Expected Calibration Error (ECE) \cite{bbq15}:
\begin{equation}
    \text{ECE}(g_\phi) = \mathbb{E} \big[ \abs{\mathbb{E}[Y|g_\phi(s)=p] - p} \big].
    \label{ece}
\end{equation}
However, since we only have finite samples, we cannot directly compute ECE with Eq.3.
Instead, we partition the [0,1] range of $p$ into $M$ equi-spaced bins and aggregate the value of each bin:
\begin{equation}
    \text{ECE}_M(g_\phi) = \sum_{m=1}^M \frac{\abs{B_m}}{N} \left| \frac{\sum_{k \in B_m} Y_k}{\abs{B_m}} - \frac{\sum_{k \in B_m} p_k}{\abs{B_m}} \right|,
    \label{ecem}
\end{equation}
where $B_m$ is $m$-th bin and $N$ is the number of samples.
The first term in the absolute value symbols denotes the ground-truth proportion of positive samples (accuracy) in $B_m$ and the second term denotes the average calibrated probability (confidence) of $B_m$.
Similarly, Maximum Calibration Error (MCE) is defined as follows:
\begin{equation}
    \text{MCE}_M(g_\phi) = \max_{m \in \{1,..,M\}} \left| \frac{\sum_{k\in B_m} Y_k}{\abs{B_m}} - \frac{\sum_{k \in B_m} p_k}{\abs{B_m}} \right|.
\end{equation}
MCE measures the worst-case discrepancy between the accuracy and the confidence.
Besides the above calibration measures, Negative Log-Likelihood (NLL) also can be used as a calibration measure \cite{cal17}.

\subsection{Calibration Method}
Existing methods for model calibration are categorized into two groups: non-parametric and parametric methods.
Non-parametric methods mostly adopt the binning scheme introduced by the histogram binning \cite{hist01}.
The histogram binning divides the uncalibrated model outputs into $B$ equi-spaced bins and samples in each bin take the proportion of positive samples in the bin as the calibrated probability.
Subsequently, isotonic regression \cite{iso12} adjusts the number of bins and their width, Bayesian binning into quantiles (BBQ) \cite{bbq15} takes an average of different binning models for the better generalization.
The non-parametric calibration methods (Hist, Isotonic, and BBQ) lack rich expressiveness since they rely on the binning scheme, which maps the ranking scores to the probabilities in a discrete manner.
Also, histogram binning \cite{hist01} and BBQ \cite{bbq15} do not guarantee the monotonicity and BBQ takes the input only from [0,1].
From the perspective of our desiderata, none of them meets all three conditions.

\begin{table}[ht]
    \centering\fontsize{9}{10}\selectfont
    \begin{tabular}{cccc}
        \toprule
        \multirow{2}{*}{Method} & (1) & (2) & (3) \\ 
         & \multirow{1}{*}{Input range} & \multirow{1}{*}{Monotonicity} & \multirow{1}{*}{Asymmetricity} \\ 
        \midrule
        Hist & \checkmark & & \checkmark \\
        Isotonic & \checkmark & \checkmark & \checkmark\\
        BBQ & & & \checkmark \\
        \midrule
        Platt & \checkmark & \checkmark & \\
        Temp. S & \checkmark & \checkmark &  \\
        Beta & & \checkmark & \checkmark\\
        \midrule
        Gaussian & \checkmark & \checkmark & \checkmark \\
        Gamma & \checkmark & \checkmark & \checkmark \\
        \bottomrule
    \end{tabular}
    \caption{Satisfaction of proposed desiderata for each calibration method. $\checkmark$ indicates the satisfaction and a blank denotes the unsatisfaction.}
    \label{desiderata}
\end{table}

The parametric methods try to fit calibration functions that map the output scores to the calibrated probabilities.
Temperature scaling ($\sigma(s/T)$) \cite{cal17}, a well-known technique for calibrating deep neural networks, is a simplified version of Platt scaling ($\sigma(as+b)$) \cite{platt99} that adopts Gaussian distributions with the same variance for the positive and the negative classes.
Beta calibration \cite{beta17} utilizes Beta distribution for the binary classification and Dirichlet calibration \cite{kull2019beyond} generalizes it for the multi-class classification.
While recent work \cite{rahimi2020intra, mukhoti2020calibrating} is focusing on parametric methods and shows promising results for image classification, they cannot be directly adopted for personalized ranking.
In this paper, we propose two parametric calibration methods that satisfy all the desiderata for the personalized ranking models.

\section{Proposed Calibration Method}
\subsection{Revisiting Platt Scaling}
Platt scaling \cite{platt99} is widely used parametric calibration method, which is a generalized form of the temperature scaling \cite{cal17}:
\begin{equation}
    g_\phi^{\text{Platt}}(s) = \sigma(bs + c),
\end{equation}
where $\phi = \{b,c\}$ are learnable parameters and $\sigma(x) = 1/(1+\text{exp}(-x))$ is the sigmoid function.
In this section, we show that Platt scaling can be derived from the assumption that the class-conditional scores follow Gaussian distributions with the same variance.

We first set the class-conditional score distribution for the positive and the negative classes:
\begin{equation}
\begin{split}
    & p(s|Y=0) = (\sqrt{2\pi}\sigma_0)^{-1} \text{exp} [ - (s-\mu_0)^2 / 2\sigma_0^2 ], \\
    & p(s|Y=1) = (\sqrt{2\pi}\sigma_1)^{-1} \text{exp} [ - (s-\mu_1)^2 / 2\sigma_1^2 ], \\
\end{split}
\end{equation}
where $\mu_0, \mu_1 \in \mathbb{R}$, $\sigma_0^2, \sigma_1^2 \in \mathbb{R}^+$ are the mean and the variance of each Gaussian distribution.
Then, the posterior is computed as follows:
\begin{equation}
\begin{split}
    P(Y=1|s) & = \frac{\pi_1 p(s|Y=1)}{\pi_1 p(s|Y=1) + \pi_0 p(s|Y=0)} \\
    & = \frac{1}{1 + \pi_0 p(s|Y=0) / \pi_1 p(s|Y=1)} \\
    & = \frac{1}{1 + \text{exp}\big[ (\frac{1}{2\sigma_1^2} - \frac{1}{2\sigma_0^2})s^2 + (\frac{\mu_0}{\sigma_0^2}-\frac{\mu_1}{\sigma_1^2})s - c \big]} \\
    & = \sigma(as^2 + bs + c),
\end{split}
\label{revisit}
\end{equation}
where $\pi_0$ and $\pi_1$ are the prior probability for each class, $a = (2\sigma_0^2)^{-1} - (2\sigma_1^2)^{-1}$, $b = \mu_1/\sigma_1^2 - \mu_0/\sigma_0^2$, and $c=\mu_1^2 / (2\sigma_1^2) - \mu_0^2 / (2\sigma_0^2) + \text{log}(\pi_0 \sigma_1) - \text{log}(\pi_1 \sigma_0) \in \mathbb{R}$.
We can see that Platt scaling is a special case of Eq.\ref{revisit} with the assumption $a=0$ (i.e., the same variance for both class-conditional score distributions).

\subsection{Gaussian Calibration}
For personalized ranking, however, the usage of the same variance for both class-conditional score distributions is not desirable, because a severe imbalance between the two classes exists in user-item interaction datasets.
Since users have distinct preferences for item categories, preferred items take only a small portion (up to 10\% in real-world datasets) of the entire itemset.
Therefore, the score distribution of diverse unpreferred items and that of distinct preferred items are likely to have disparate variances.

To tackle this problem, we let the variance of each class-conditional score distribution be optimized with datasets, without any naive assumption of the same variance for both classes:
\begin{equation}
    g_\phi^{\text{Gaussian}}(s) = \sigma(as^2 + bs + c),
\end{equation}
where $\phi = \{a, b,c\}$ are learnable parameters and can be any real numbers.
Since $a=(2\sigma_0^2)^{-1} - (2\sigma_1^2)^{-1}$ can capture the different deviations of two classes during the training, we can handle the distinct distribution of each class.

\subsection{Gamma Calibration}
Gamma distribution is also widely adopted to model the score distribution of ranking models \cite{gamma99}.
Unlike Gaussian distribution that is symmetric about its mean, Gamma distribution can capture the skewed population of ranking scores that might exist in the datasets.
In this section, we set the class-conditional score distribution to Gamma distribution:
\begin{equation}
\begin{split}
    & p(s|Y=0) = \Gamma(\alpha_0)^{-1} \beta_0^{\alpha_0} s^{\alpha_0-1} \text{exp}(-\beta_0s), \\
    & p(s|Y=1) = \Gamma(\alpha_1)^{-1} \beta_1^{\alpha_1} s^{\alpha_1-1} \text{exp}(-\beta_1s), \\
\end{split}
\end{equation}
where $\Gamma(\cdot)$ is the Gamma function, $\alpha_0 , \alpha_1 , \beta_0, \beta_1 \in \mathbb{R}^+$ are the shape and the rate parameters of each Gamma distribution.
Then, the posterior is computed as follows:
\begin{equation}
\begin{split}
    P(Y=1|s) & = \frac{1}{1 + \pi_0 p(s|Y=0) / \pi_1 p(s|Y=1)} \\
    & = \frac{1}{1 + \frac{\pi_0 \beta_0^{\alpha_0} \Gamma(\alpha_1)}{\pi_1 \beta_1^{\alpha_1} \Gamma(\alpha_0)} s^{\alpha_0 - \alpha_1} \text{exp}[(\beta_1-\beta_0)s]} \\
    & = \frac{1}{1 + \text{exp}\big[ (\alpha_0-\alpha_1)\text{log}s + (\beta_1-\beta_0)s - c \big]} \\
    & = \sigma(a\text{log}s + bs + c),
\end{split}
\end{equation}
where $a = \alpha_1-\alpha_0$, $b = \beta_0-\beta_1$, and $c=\text{log}(\pi_1\beta_1^{\alpha_1}\Gamma(\alpha_0) / \pi_0\beta_0^{\alpha_0}\Gamma(\alpha_1)) \in \mathbb{R}$.
Therefore, Gamma calibration can be formalized as follows:
\begin{equation}
    g_\phi^{\text{Gamma}}(s) = \sigma(a\text{log}s + bs + c),
\end{equation}
where $\phi = \{a, b,c\}$ are learnable parameters.
Since Gamma distribution is defined only for the positive real number, we need to shift the score to make all the inputs positive: $s \leftarrow s-s_{\text{min}}$, where $s_{\text{min}}$ is the minimum ranking score. 

\subsection{Other Distributions}
Besides Gaussian and Gamma distribution, Swets \cite{exp69} adopts Exponential distributions for the score distribution of both classes:
\begin{equation}
\begin{split}
    & p(s|Y=0) = \lambda_0 \text{exp}(-\lambda_0 s), \\
    & p(s|Y=1) = \lambda_1 \text{exp}(-\lambda_1 s), \\
\end{split}
\end{equation}
where $\lambda_0, \lambda_1 \in \mathbb{R}^+$ are the rate parameters for each Exponential distribution.
Note that for Gamma distribution and Exponential distribution, we need to shift the score to make all the inputs positive: $s \leftarrow s-s_{min}$, where $s_{min}$ is the minimum ranking score. 
Then, the posterior can be computed as follows:
\begin{equation}
\begin{split}
    P(Y=1|s) & = \frac{1}{1 + \pi_0 p(s|Y=0) / \pi_1 p(s|Y=1)} \\
    & = \frac{1}{1 + \text{exp}[ (\lambda_1 - \lambda_0)s + \text{log}(\pi_0 \lambda_0 / \pi_1 \lambda_1) ]} \\
    & = \sigma(bs + c),
\end{split}
\end{equation}
where $b=\lambda_0 - \lambda_1$ and $c=\text{log}(\pi_1 \lambda_1 / \pi_0 \lambda_0) \in \mathbb{R}$.
It is the exact same form as Platt scaling.

On the other hand, Manmatha \cite{normexp01} proposes Exponential distribution for the negative class and Gaussian distribution for the positive class:
\begin{equation}
\begin{split}
    & p(s|Y=0) = \lambda_0 \text{exp}(-\lambda_0 s), \\
    & p(s|Y=1) = (\sqrt{2\pi}\sigma_1)^{-1} \text{exp} [ - (s-\mu_1)^2 / 2\sigma_1^2 ],\\
\end{split}
\end{equation}
where $\lambda_0, \mu_1, \sigma_1^2 \in \mathbb{R}^+$.
Then, we have the posterior as follows:
\begin{equation}
\begin{split}
    P(Y=1|s) & = \frac{1}{1 + \pi_0 p(s|Y=0) / \pi_1 p(s|Y=1)} \\
    & = \frac{1}{1 + \frac{\pi_0 \lambda_0 \sqrt{2\pi} \sigma_1}{\pi_1} \text{exp}[(s-\mu_1)^2 / 2\sigma_1^2 - \lambda_0 s] } \\
    & = \frac{1}{1 + \text{exp}[ (2\sigma_1^2)^{-1}s^2 - (\lambda_0 + \mu_1/\sigma_1^2)s - c ]} \\ 
    & = \sigma(as^2 + bs + c),
\end{split}
\label{manmatha}
\end{equation}
where $a=(-2\sigma_1^2)^{-1} \in \mathbb{R}^-$, $b=\lambda_0 + \mu_1/\sigma_1^2 \in \mathbb{R}^+$, and $c=\text{log}(\pi_1 / \pi_0 \lambda_0 \sqrt{2\pi} \sigma_1) - \mu_1^2/(2\sigma_1^2)\in \mathbb{R}$.
With the constraints $a<0$ and $b>0$, the calibration function in Eq.\ref{manmatha} cannot be monotonically increasing on $s>0$ with any parameters.

Lastly, Kanoulas \cite{gammagaussian10} proposes gamma distribution for the negative class and Gaussian distribution for the positive class:
\begin{equation}
\begin{split}
    & p(s|Y=0) = \rho(\alpha_0)^{-1} \beta_0^{\alpha_0} s^{\alpha_0-1} \text{exp}(-\beta_0s), \\
    & p(s|Y=1) = (\sqrt{2\pi}\sigma_1)^{-1} \text{exp} [ - (s-\mu_1)^2 / 2\sigma_1^2 ],\\
\end{split}
\end{equation}
where $\alpha_0$, $\beta_0$, $\mu_1$, $\sigma_1^2 \in \mathbb{R}^{+}$.
The posterior for these likelihoods can be computed as follows:
\begin{equation}
\begin{split}
    P(Y=1|s) & = \frac{1}{1 + \pi_0 p(s|Y=0) / \pi_1 p(s|Y=1)} \\
    & = \frac{1}{1 + \text{exp}[ \frac{1}{2\sigma_1^2}s^2 - (\beta_0 + \frac{\mu_1}{\sigma_1^2})s + (\alpha_0 -1)\text{log}s - c]} \\
    & = \sigma(as^2 + bs + b'\text{log}s + c),
\end{split}
\label{kanoulas}
\end{equation}
where $a=(-2\sigma_1^2)^{-1} \in \mathbb{R}^-$, $b=\beta_0+\mu_1/\sigma_1^2 \in \mathbb{R}^+$, $b'=1-\alpha_0 \in \mathbb{R}$, and $c=\text{log}(\pi_1 \rho(\alpha_0) / \pi_0\beta_0^{\alpha_0} \sigma_1 \sqrt{2\pi}) - \mu_1^2/(2\sigma_1^2) \in \mathbb{R}$.
For this function to be monotonically increasing, we need a non-linear constraint $b'- \frac{b^2}{8a} < 0$ (derivation of this constraint can be easily done by taking the derivative of Eq.\ref{kanoulas} w.r.t. the score $s$).
Since the optimization of logistic regression with non-linear constraints is not straightforward, it is hard for this posterior to satisfy the second desideratum.

\subsection{Monotonicity for Proposed Desiderata}
The proposed calibration methods naturally satisfy the first and the third of our desiderata:
(1) the proposed methods take the unbounded ranking scores and produce calibrated probabilities;
(2) the proposed methods have richer expressiveness than Platt scaling or temperature scaling, since they have a larger capacity to express asymmetric distributions.
The last condition that our calibration methods need to meet is that they should be monotonically increasing to maintain the ranking order.
To this end, we need linear constraints on the parameters of each method: $2as+b > 0$ for Gaussian calibration and $a/s+b>0$ for Gamma calibration.
Since these constraints are linear and we have only three learnable parameters, the optimization of constrained logistic regression is easily done in at most a few minutes with the existing module of Scipy \cite{scikit-learn}.

\begin{proposition}
Gaussian calibration $g_{\phi}(s) = \sigma(as^2+bs+c)$ is monotonically increasing only and only if the parameter $a$ and $b$ satisfy the constraint $2as+b < 0$ for $s_{\textup{min}}$ and $s_{\textup{max}}$.
\end{proposition}
\begin{proof}
\begin{equation*}
\begin{split}
    & g_{\phi}'(s) = \sigma'(as^2+bs+c) \cdot (2as+b) > 0 \\
    & \Longleftrightarrow 2as+b > 0 \\
    & \because \; \sigma'(x) = \sigma(x)(1-\sigma(x)) > 0 \;\; \textup{for all} \; x \in \mathbb{R} \\
    & \Longleftrightarrow 2as_{\textup{min}}+b > 0 \textup{ and } 2as_{\textup{max}}+b > 0 \\
    & \because 2as + b \textup{ is a linear function of s in } [s_{\textup{min}}, s_{\textup{max}}]. \\
\end{split}
\end{equation*}
\end{proof}

\begin{proposition}
Gamma calibration $g_{\phi}(s) = \sigma(a\textup{log}s+bs+c)$ is monotonically increasing only and only if the parameter $a$ and $b$ satisfy the constraint $a/s+b < 0$ for $s_{\textup{min}}$ and $s_{\textup{max}}$.
\end{proposition}
\begin{proof}
\begin{equation*}
\begin{split}
    & g_{\phi}'(s) = \sigma'(a\textup{log}s+bs+c) \cdot (a/s+b) > 0 \\
    & \Longleftrightarrow a/s+b > 0 \\
    & \because \; \sigma'(x) = \sigma(x)(1-\sigma(x)) > 0 \;\; \textup{for all} \; x \in \mathbb{R} \\
    & \Longleftrightarrow a/s_{\textup{min}}+b > 0 \textup{ and } a/s_{\textup{max}}+b > 0 \\
    & \because a/s + b \textup{ is a monotonic function of s in } [s_{\textup{min}}, s_{\textup{max}}]. \\
\end{split}
\end{equation*}
\end{proof}

\section{Unbiased Parameter Fitting}
\subsection{Naive Log-loss}
After we formalize Gaussian Calibration and Gamma Calibration, we need to optimize their learnable parameters $\phi$.
A well-known way to fit them is to use log-loss on the held-out validation set, which can be the same set used for the hyperparameter tuning \cite{cal17, beta17}.
Since we only observe the interaction indicator $Y_{u,i}$, the naive negative log-likelihood is computed for a user-item pair as follows:
\begin{equation}
      \mathcal{L}_{\text{naive}} = - Y_{u,i} \, \text{log}( g_{\phi}(s_{u,i})) - (1-Y_{u,i}) \text{log}(1 - g_{\phi}(s_{u,i})).
\end{equation}
where $s_{u,i} = f_{\theta}(u,i)$ is the ranking score for the user-item pair.
Note that during the fitting of the calibration function $g_{\phi}(s)$, the parameters of the pre-trained ranking model $f_{\theta}(u,i)$ are fixed.

\subsection{Ideal Log-loss for Preference Estimation}
The observed interaction label $Y_{u,i}$, however, indicates the presence of user-item interaction, not the user's preference on the item.
Therefore, $Y_{u,i}=0$ does not necessarily mean the user's negative preference, but it can be that the user is not aware of the item.
If we fit the calibration function with $\mathcal{L}_{\textup{naive}}$, mapped probabilities could be biased towards the negative preference by treating the unobserved positive pair as the negative pair.
To handle this implicit interaction process, we borrow the idea of decomposing the interaction variable $Y_{u,i}$ into two independent binary variables \cite{ips16}:
\begin{equation}
\begin{split}
    Y_{u,i} & = O_{u,i} \cdot R_{u,i}, \\
    P(Y_{u,i}=1) & = P(O_{u,i}=1) \cdot P(R_{u,i}=1) \\
    & = \omega_{u,i} \cdot \rho_{u,i},
\end{split}
\end{equation}
where $O_{u,i}$ is a binary random variable representing whether the item $i$ is observed by user $u$, and $R_{u,i}$ is a binary random variable representing whether the item $i$ is preferred by user $u$.
The user-item pair interacts ($Y_{u,i}=1$) when the item is observed ($O_{u,i}=1$) and preferred ($R_{u,i}=1$) by the user.

The goal of this paper is to estimate the probability of an item being \textit{preferred} by a user, not the probability of an item being \textit{interacted} by a user.
Therefore, we need to train $g_{\phi}(s)$ for predicting $P(R=1|s)$ instead of $P(Y=1|s)$\footnote{We can replace $Y_{u,i}$ with $R_{u,i}$ in Eq.2.2 - 2.12.}.
To this end, we need a new ideal loss function that can guide the optimization towards the true preference probability:
\begin{equation}
      \mathcal{L}_{\text{ideal}} = - R_{u,i} \, \text{log}(g_{\phi}(s_{u,i})) - (1-R_{u,i}) \text{log}(1 - g_{\phi}(s_{u,i})).
\end{equation}
The ideal loss function enables the calibration function to learn the unbiased preference probability.
However, since we cannot observe the variable $R_{u,i}$ from the training set, the ideal log-loss cannot be computed directly.

\subsection{Unbiased Empirical Risk Minimization}
In this section, we design an unbiased empirical risk minimization (UERM) framework to obtain the ideal empirical risk minimizer.
We deploy the Inverse Propensity Scoring (IPS) estimator \cite{ips94}, which is a technique for estimating the counterfactual outcome of a subject under a particular treatment.
The IPS estimator is widely adopted for the unbiased rating prediction \cite{ips16, ips19} and the unbiased pairwise ranking \cite{joachims2017unbiased, saito19}.
For a user-item pair, the inverse propensity-scored log-loss for the unbiased empirical risk minimization is defined as follows:
\begin{equation}
      \mathcal{L}_{\text{UERM}} = - \frac{Y_{u,i}}{\omega_{u,i}} \, \text{log}( g_{\phi}(s_{u,i})) - (1-\frac{Y_{u,i}}{\omega_{u,i}}) \text{log}(1 - g_{\phi}(s_{u,i})),
\end{equation}
where $\omega_{u,i} = P(O_{u,i}=1)$ is called \textit{propensity score}. 
\begin{proposition}
$\hat{\mathcal{R}}_{\textup{UERM}}(g_{\phi}|\omega)$, which is the empirical risk of $\mathcal{L}_{\textup{UERM}}$ on validation set with true propensity score $\omega$, is equal to $\hat{\mathcal{R}}_{\textup{ideal}}(g_{\phi})$, which is the ideal empirical risk.
\end{proposition}
\begin{proof}
\begin{equation*}
\begin{split}
    \hat{\mathcal{R}}_{\text{UERM}}(g_{\phi}|\omega) = & \: \mathbb{E}_{(u,i) \in \mathcal{D}_{\text{val}}} \big[\mathcal{L}_{\text{IPS}}(u,i)\big] \\
    = & - \frac{1}{\abs{\mathcal{D}_{\text{val}}}} \sum_{(u,i) \in \mathcal{D}_{\text{val}}} \mathbb{E} \left[ \frac{Y_{u,i}}{\omega_{u,i}} \text{log}( g_{\phi}(s_{u,i})) \right. + \left. (1-\frac{Y_{u,i}}{\omega_{u,i}}) \text{log}(1 - g_{\phi}(s_{u,i})) \right]\\
    = & - \frac{1}{\abs{\mathcal{D}_{\text{val}}}} \sum_{(u,i) \in \mathcal{D}_{\text{val}}} \frac{\omega_{u,i}\rho_{u,i}}{\omega_{u,i}} \text{log}( g_{\phi}(s_{u,i})) + (1-\frac{\omega_{u,i}\rho_{u,i}}{\omega_{u,i}}) \text{log}(1 - g_{\phi}(s_{u,i})) \\
    = & - \frac{1}{\abs{\mathcal{D}_{\text{val}}}} \sum_{(u,i) \in \mathcal{D}_{\text{val}}} \rho_{u,i} \text{log}( g_{\phi}(s_{u,i})) + (1-\rho_{u,i}) \text{log}(1 - g_{\phi}(s_{u,i})) \\
    = & \: \mathbb{E}_{(u,i) \in \mathcal{D}_{\text{val}}} \big[\mathcal{L}_{\text{ideal}}(u,i)\big] \\
    = & \: \hat{\mathcal{R}}_{\text{ideal}}(g_{\phi}).
\end{split}
\end{equation*}
\end{proof}
This proposition shows that we can get the unbiased empirical risk minimizer by $\phi^{\textup{UERM}} = \textup{argmin}_{\phi}\{ \hat{\mathcal{R}}_{\text{UERM}}(g_{\phi}|\omega)\}$ when only $Y_{u,i}$ is observed.
The remaining challenge is to estimate the propensity score $\omega_{u,i}$ from the dataset.
There have been proposed several techniques for estimating the propensity score such as Naive Bayes \cite{ips16} or logistic regression \cite{rosenbaum2002overt}.
However, the Naive Bayes needs unbiased held-out data for the missing-at-random condition and the logistic regression needs additional information like user demographics and item categories.
In this paper, we adopt a simple way that utilizes the popularity of items as done in \cite{saito19}:
$\hat{\omega}_{u,i} = ( \sum_{u \in \mathcal{U}} Y_{u,i}/\text{max}_{i \in \mathcal{I}} \sum_{u \in \mathcal{U}} Y_{u,i} )^{0.5}.$
While one can be concerned that this estimate of propensity score may be inaccurate, Schnabel \cite{ips16} shows that we merely need to estimate better than the naive uniform assumption.
We provide an experimental result that demonstrates our estimate of the propensity score shows comparable performance with Naive Bayes and Logistic Regression that use additional information.

For deeper insights into the variability of the estimated empirical risk, we investigate the bias when the propensity scores are inaccurately estimated.
\begin{proposition}
The bias of $\, \hat{\mathcal{R}}_{\textup{UERM}}(g_{\phi}|\hat{\omega})$ induced by the inaccurately estimated propensity scores $\hat{\omega}$ is  $\frac{1}{\abs{\mathcal{D}_{\textup{val}}}} \sum_{(u,i) \in \mathcal{D}_{\textup{val}}} \rho_{u,i} \left(\frac{\omega_{u,i}}{\hat{\omega}_{u,i}}-1 \right) \textup{log}\left( \frac{g_{\phi}(s_{u,i})}{1-g_{\phi}(s_{u,i})} \right)$.
\end{proposition}
\begin{proof}
\begin{equation*}
\begin{split}
    \text{bias} = & \, \hat{\mathcal{R}}_{\text{UERM}}(g_{\phi}|\hat{\omega}) - \hat{\mathcal{R}}_{\text{ideal}}(g_{\phi}) \\
    = & \, \frac{1}{\abs{\mathcal{D}_{\text{val}}}} \sum_{(u,i) \in \mathcal{D}_{\text{val}}} - \frac{\omega_{u,i}\rho_{u,i}}{\hat{\omega}_{u,i}} \text{log}(g_{\phi}(s_{u,i})) \\
    & - (1-\frac{\omega_{u,i}\rho_{u,i}}{\hat{\omega}_{u,i}})\text{log}(1-g_{\phi}(s_{u,i})) + \rho_{u,i}\text{log}(g_{\phi}(s_{u,i}))  + (1-\rho_{u,i})\text{log}(1-g_{\phi}(s_{u,i})) \\
    = & \, \frac{1}{\abs{\mathcal{D}_{\text{val}}}} \sum_{(u,i) \in \mathcal{D}_{\text{val}}} (\omega_{u,i}/\hat{\omega}_{u,i}-1) [-\rho_{u,i}\text{log}(g_{\phi}(s_{u,i})) + \rho_{u,i}\text{log}(1-g_{\phi}(s_{u,i}))] \\
    = & \, \frac{1}{\abs{\mathcal{D}_{\text{val}}}} \sum_{(u,i) \in \mathcal{D}_{\text{val}}} \rho_{u,i} \left(\frac{\omega_{u,i}}{\hat{\omega}_{u,i}}-1 \right) \text{log}\left( \frac{1-g_{\phi}(s_{u,i})}{g_{\phi}(s_{u,i})} \right). \\
\end{split}
\end{equation*}
\end{proof}
\noindent Obviously, the bias is zero when the propensity score is correctly estimated.
Furthermore, we can see that the magnitude of the bias is affected by the inverse of the estimated propensity score.
This finding is consistent with the previous work \cite{clip19} that proposes to adopt a propensity clipping technique to reduce the variability of the bias.
In this work, we use a simple clipping technique $\hat{\omega}_{u,i} \leftarrow \max\{\hat{\omega}_{u,i}, 0.1\}$ that can prevent the item with extremely low popularity from amplifying the bias \cite{saito19}.

\section{Experiment}
\subsection{Experimental Setup}
Our source code is publicly available\footnote{\url{https://github.com/WonbinKweon/CalibratedRankingModels_AAAI2022}}.

\textbf{Datasets.} To evaluate the calibration quality of predicted preference probability, we need an unbiased test set where we can directly observe the preference variable $R_{u,i}$ without any bias from the observation process $O_{u,i}$.
To the best of our knowledge, there are two real-world datasets that have separate unbiased test sets where the users are asked to rate uniformly sampled items (i.e., $O_{u,i}=1$ for test sets).
Note that in the training set, we only observe the interaction $Y_{u,i}$.
Yahoo!R3\footnote{\url{http://research.yahoo.com/Academic_Relations}} has over 300K \textit{interactions} in the training set and 54K \textit{preferences} in the test set from 15.4K users and 1K songs.
Coat \cite{ips16} has over 7K interactions in the training set and 4.6K preferences in the test set from 290 users and 300 coats.
We hold out 10\% of the training set as the validation set for the hyperparameter tuning of the base models and the optimization of the calibration methods.
For the training and the test set, we convert the ratings over 3 to $Y=1$, and the ratings under 4 to $Y=0$ as done in the conventional papers \cite{ncf17, saito19}.

\textbf{Base models.}
For rigorous evaluation, we apply the calibration methods on several widely-used personalized ranking models with various model architectures and loss functions: Bayesian Personalized Ranking (BPR) \cite{bpr09}, Neural Collaborative Filtering (NCF) \cite{ncf17}, Collaborative Metric Learning (CML) \cite{cml17}, Unbiased BPR (UBPR) \cite{saito19}, and LightGCN (LGCN) \cite{lightgcn20}.
BPR \cite{bpr09} learns the user and the item embeddings with BPR loss:
\begin{equation}
    \mathcal{L}_{\textup{BPR}} = \sum_{u \in \mathcal{U}, i,j \in \mathcal{I}} - \text{log} \sigma(f_\theta(u,i) - f_\theta(u,j)) Y_{u,i} (1-Y_{u,j}),
\end{equation}
where $f_\theta(u,i) = \textbf{e}_u \cdot \textbf{e}_i$ and $\textbf{e}_u$, $\textbf{e}_i$ are user and item embeddings.
NCF \cite{ncf17} learns the ranking score of a user-item pair with the binary cross-entropy loss:
\begin{equation}
\begin{split}
      \mathcal{L}_{\textup{NCF}} & = \sum_{u \in \mathcal{U}, i \in \mathcal{I}} -Y_{u,i}\textup{log}\sigma(f_\theta(u,i)) \\ & -(1-Y_{u,i})\textup{log}(1-\sigma(f_\theta(u,i))),  
\end{split}
\end{equation}
where $f_\theta(u,i) = \textup{NeuralNet}(\textbf{e}_u, \textbf{e}_i)$.
CML \cite{cml17} learns the user and the item embeddings with the triplet loss:
\begin{equation}
    \mathcal{L}_{\textup{CML}} = \sum_{u \in \mathcal{U}, i,j \in \mathcal{I}} \textup{max}(0, m+f_\theta(u,i) - f_\theta(u,j)) Y_{u,i} (1-Y_{u,j}),
\end{equation}
where $m$ is the margin and $f_\theta(u,i) = ||\textbf{e}_u - \textbf{e}_i||^2$.
UBPR \cite{saito19} learns the user and the item embeddings with the inverse propensity-scored BPR loss:
\begin{equation}
    \mathcal{L}_{\textup{UBPR}} = \sum_{u \in \mathcal{U}, i,j \in \mathcal{I}} - \text{log} \sigma(f_\theta(u,i) - f_\theta(u,j)) \frac{Y_{u,i}}{\omega_{u,i}} (1-\frac{Y_{u,i}}{\omega_{u,i}}),
\end{equation}
where $f_\theta(u,i) = \textbf{e}_u \cdot \textbf{e}_i$ and $\omega_{u,i}$ is the propensity score estimated by the same technique used in our framework.
LGCN \cite{lightgcn20} learns the user and the item embeddings with BPR loss:
\begin{equation}
    \mathcal{L}_{\textup{LGCN}} = \sum_{u \in \mathcal{U}, i,j \in \mathcal{I}} - \text{log} \sigma(f_\theta(u,i) - f_\theta(u,j)) Y_{u,i} (1-Y_{u,j}),
\end{equation}
where $f_\theta(u,i) = \textup{LGCN}(\textbf{e}_u) \cdot \textup{LGCN}(\textbf{e}_i)$ and  $\textup{LGCN}(\cdot)$ is simplified Graph Convolutional Networks (GCN).

\begin{table}[ht!]
    \centering\fontsize{9}{10}\selectfont
    \begin{tabular}{c|ccccc|ccccc}
        \toprule
         & \multicolumn{5}{c|}{Yahoo!R3} & \multicolumn{5}{c}{Coat} \\
        \toprule
        Metric & BPR & NCF & CML & UBPR & LGCN & BPR & NCF & CML & UBPR & LGCN \\
        \midrule
        NDCG@1 & 0.5070 & 0.5181 & 0.5259 & 0.5328 & \textbf{0.5443} & 0.3924 & 0.5341 & \textbf{0.5865} & 0.3982 & 0.4008 \\
        NDCG@3 & 0.5519 & 0.5812 & 0.5716 & 0.5919 & \textbf{0.6002} & 0.3761 & \textbf{0.5129} & 0.5020 & 0.3962 & 0.3973 \\
        NDCG@5 & 0.6176 & 0.6467 & 0.6379 & 0.6555 & \textbf{0.6637} & 0.4302 & \textbf{0.5452} & 0.5142 & 0.4418 & 0.4301 \\
        Recall@1 & 0.3126 & 0.3213 & 0.3286 & 0.3345 & \textbf{0.3395} & 0.1321 & 0.2132 & \textbf{0.2334} & 0.1354 & 0.1395 \\ 
        Recall@3 & 0.5743 & 0.6098 & 0.5918 & 0.6207 & \textbf{0.6280} & 0.2852 & \textbf{0.4166} & 0.3847 & 0.3181 & 0.3113 \\ 
        Recall@5 & 0.7428 & 0.7779 & 0.7613 & 0.7820 & \textbf{0.7915} & 0.4638 & \textbf{0.5146} & 0.4847 & 0.4757 & 0.4816 \\
        \bottomrule
    \end{tabular}
    \caption{Ranking performance of each personalized ranking model. Numbers in boldface are the best results.}
    \label{baseper}
\end{table}

For the base models, we basically follow the source code of the authors.
NCF has 2-layer MLP for MLP module and 1-layer MLP for the prediction module.
LGCN has a 2-layer simplified GCN for the training and the inference.
We use 128 for the size of user and item embeddings for all base models, except NCF which adopts 64 for the embedding size.
The batch size is 512, the learning rate is 0.001, the weight decay rate is 0.001, and the negative sample rate is 1.
Each model is trained until the convergence and their ranking performance can be found in Table \ref{baseper}.
LGCN shows the best performance for Yahoo!R3 dataset and NCF or CML shows the best performance for Coat dataset.

\textbf{Calibration methods compared.}
We evaluate the proposed calibration methods with various calibration methods.
For the naive baseline, we adopt the minmax normalizer and the sigmoid function which simply re-scale the scores into [0,1] without calibration.
For non-parametric methods, we adopt Histogram binning \cite{hist01}, Isotonic regression \cite{iso12}, and BBQ \cite{bbq15}.
For Histogram binning, we set the number of bins as 50 for Yahoo!R3, 15 for Coat.
For BBQ, we set the number of binning methods as 4 and the number of bins of each binning method is \{10,20,50,100\}.
For those that do not meet the first desideratum (BBQ and Beta calibration), we adopt the sigmoid function to re-scale the input into [0,1].
For parametric methods, we adopt Platt scaling \cite{platt99} and Beta calibration \cite{beta17}.
Note that we do not compare recent work designed for multi-class classification \cite{kull2019beyond, rahimi2020intra}, since they are either the generalized version of Beta calibration or cannot be directly adopted for the personalized ranking models.

\textbf{Evaluation metrics.}
We adopt well-known calibration metrics like ECE, MCE with $M=15$, and NLL as done in recent work \cite{kull2019beyond, rahimi2020intra}.
We also plot the reliability diagram that shows the discrepancy between the accuracy and the average calibrated probability of each probability interval.
Note that evaluation metrics are computed on $R_{u,i}$ which is observed only from the test set.

\textbf{Evaluation process.}
We first train the base personalized ranking model $f_{\theta}(u,i)$ with $Y_{u,i}$ on the training set.
Second, we compute ranking score $s_{u,i} = f_{\theta}(u,i)$ for user-item pairs in the validation set.
Third, we optimize the calibration method $g_{\phi}(s)$ on the validation set with the computed $s_{u,i}$ and the estimated $\hat{\omega}_{u,i}$, with $f_{\theta}(u,i)$ fixed.
Lastly, we evaluate the calibrated probability $p=g_{\phi}(s_{u,i})$ with $R_{u,i}$ from the unbiased test set by using the above evaluation metrics.

\textbf{Computing  infrastructures}
We adopt a Titan X GPU and an Intel(R) Core(TM) i7-7820X 3.60GHz CPU.
Optimization of all calibration methods is done in at most a few minutes.

\begin{sidewaystable}[ph!]
    \caption{Expected Calibration Error of each calibration method applied on five personalized ranking models. Numbers in boldface are the best results and \textit{Improv} denotes the improvement of the best proposed method over the best competitor (Platt or Beta with $\mathcal{L}_{\text{UERM}}$).}
    \begin{tabular}{cc|ccccc|ccccc}
        \toprule
         & & \multicolumn{5}{c|}{Yahoo!R3} & \multicolumn{5}{c}{Coat} \\
        \toprule
        Type & Methods & BPR & NCF & CML & UBPR & LGCN & BPR & NCF & CML & UBPR & LGCN \\
        \midrule
        \multirow{2}{*}{uncalibrated} & MinMax & 0.4929 & 0.4190 & 0.3152 & 0.3004 & 0.2258 & 0.1790 & 0.4624 & 0.1834 & 0.1920 & 0.2350\\
         & Sigmoid & 0.3065 & 0.0729 & 0.0526 & 0.2516 & 0.3024 & 0.2196 & 0.1422 & 0.0647 & 0.1415 & 0.0508\\
        \midrule
        \multirow{3}{*}{non-parametric} & Hist & 0.0161 & 0.0133 & 0.0641 & 0.0130 & 0.0194 & 0.0552 & 0.0230 & 0.0161 & 0.0514 & 0.0470\\
         & Isotonic & 0.0146 & 0.0130 & 0.0635 & 0.0127 & 0.0154 & 0.0474 & 0.0159 & 0.0160 & 0.0490 & 0.0453\\
         & BBQ & 0.0136 & 0.0137 & 0.0634 & 0.0140 & 0.0165 & 0.0552 & 0.0178 & 0.0198 & 0.0459 & 0.0494\\
        \midrule
        \multirow{4}{*}{\shortstack{parametric \\ w/ $\mathcal{L}_{\textup{naive}}$}} & Platt & 0.0126 & 0.0146 & 0.0515 & 0.0107 & 0.0099 & 0.0441 & 0.0245 & 0.0203 & 0.0423 & 0.0407\\
         & Beta & 0.0127 & 0.0144 & 0.0504 & 0.0105 & 0.0150 & 0.0451 & 0.0258 & 0.0270 & 0.0416 & 0.0407\\
         & Gaussian & 0.0129 & 0.0104 & 0.0486 & 0.0105 & 0.0073 & 0.0436 & 0.0264 & 0.0245 & 0.0410 & 0.0404\\
         & Gamma & 0.0108 & 0.0145 & 0.0512 & 0.0107 & 0.0098 & 0.0424 & 0.0239 & 0.0208 & 0.0405 & 0.0406\\
        \midrule
          & Platt & 0.0106 & 0.0129 & 0.0303 & 0.0100 & 0.0070 & 0.0411 & 0.0120 & 0.0155 & 0.0354 & 0.0224\\
        parametric & Beta & 0.0109 & 0.0132 & 0.0305 & 0.0094 & 0.0076 & 0.0414 & 0.0075 & 0.0183 & 0.0375 & 0.0266\\
        w/ $\mathcal{L}_{\textup{UERM}}$ & Gaussian & 0.0106 & \textbf{0.0096} & \textbf{0.0285} & \textbf{0.0070} & \textbf{0.0061} & 0.0393 & 0.0062 & \textbf{0.0147} & \textbf{0.0323} & \textbf{0.0208}\\
         & Gamma & \textbf{0.0100} & 0.0117 & 0.0287 &0.0085 &0.0065 & \textbf{0.0390} & \textbf{0.0061} &0.0148 & 0.0326 & 0.0215\\
        \midrule
        & \textit{Improv} & 5.85\% & 25.35\% & 5.94\% & 25.85\% & 12.86\% & 5.21\% & 18.67\% & 5.41\% & 8.81\% & 7.14\% \\
        \bottomrule
    \end{tabular}
    \label{calimain}
\end{sidewaystable}

\begin{sidewaystable}[ph!]
    \caption{Maximum Calibration Error with $M=15$ of each calibration method applied on five personalized ranking models. Numbers in boldface are the best results.}
    \begin{tabular}{cc|ccccc|ccccc}
        \toprule
         & & \multicolumn{5}{c|}{Yahoo!R3} & \multicolumn{5}{c}{Coat} \\
        \toprule
        Type & Methods & BPR & NCF & CML & UBPR & LGCN & BPR & NCF & CML & UBPR & LGCN \\
        \midrule
        \multirow{2}{*}{uncalibrated} & MinMax & 0.4929 & 0.4190 & 0.3152 & 0.3004 & 0.2258 & 0.1790 & 0.4624 & 0.1834 & 0.1920 & 0.2350\\
         & Sigmoid & 0.5726 & 0.5248 & 0.2571 & 0.7103 & 0.5082 & 0.3868 & 0.6223 & 0.6707 & 0.2699 & 0.4642\\
        \midrule
        \multirow{3}{*}{non-parametric} & Hist & 0.2620 & 0.2369 & 0.6184 & 0.1349 & 0.1366 & 0.2492 & 0.4083 & 0.4924 & 0.2910 & 0.3940\\
         & Isotonic & 0.2136 & 0.2173 & 0.5252 & 0.1909 & 0.1270 & 0.2495 & 0.4246 & 0.4812 & 0.2547 & 0.3500\\
         & BBQ & 0.2116 & 0.2489 & 0.7076 & 0.1701 & 0.1157 & 0.2545 & 0.3749 & 0.3423 & 0.3358 & 0.2589\\
        \midrule
        \multirow{4}{*}{\shortstack{parametric \\ w/ $\mathcal{L}_{\textup{naive}}$}} & Platt & 0.2032 & 0.2876 & 0.3254 & 0.3117 & 0.3068 & 0.3664 & 0.1844 & 0.5796 & 0.4079 & 0.4633\\
         & Beta & 0.2563 & 0.2575 & 0.5049 & 0.2705 & 0.3033 & 0.4071 & 0.2175 & 0.5782 & 0.4352 & 0.3367\\
         & Gaussian & 0.2387 & 0.2366 & 0.3737 & 0.3282 & 0.2200 & 0.3177 & 0.1678 & 0.5878 & 0.4807 & 0.3854\\
         & Gamma & 0.2123 & 0.2260 & 0.3149 & 0.3126 & 0.2962 & 0.3616 & 0.1702 & 0.5968 & 0.4185 & 0.3480\\
        \midrule
          & Platt & 0.1951 & 0.2504 & 0.2293 & 0.1257 & 0.1143 & 0.2625 & 0.1882 & 0.4423 & 0.2708 & 0.2225\\
        parametric & Beta & 0.2004 & 0.2453 & 0.2734 & 0.1317 & 0.1905 & 0.3278 & 0.2085 & 0.4150 & 0.2574 & 0.2501\\
        w/ $\mathcal{L}_{\textup{UERM}}$ & Gaussian & 0.1816 & 0.2259 & 0.2451 & \textbf{0.1231} & \textbf{0.1064} & \textbf{0.2380} & \textbf{0.1236} & \textbf{0.3390} & \textbf{0.2419} & 0.2117\\
         & Gamma & \textbf{0.1592} & \textbf{0.2074} & \textbf{0.2268} & \textbf{0.1231} & 0.1118 & 0.2543 & \textbf{0.1236} &0.4663 &0.2478 & \textbf{0.2023}\\
        \bottomrule
    \end{tabular}
    \label{tabmce}
\end{sidewaystable}

\begin{sidewaystable}[ph!]
\caption{Negative Log-Likelihood of each calibration method applied on five personalized ranking models. Numbers in boldface are the best results.}
    \begin{tabular}{cc|ccccc|ccccc}
        \toprule
         & & \multicolumn{5}{c|}{Yahoo!R3} & \multicolumn{5}{c}{Coat} \\
        \toprule
        Type & Methods & BPR & NCF & CML & UBPR & LightGCN & BPR & NCF & CML & UBPR & LightGCN \\
        \midrule
        \multirow{2}{*}{uncalibrated} & MinMax & 0.4929 & 0.4190 & 0.3152 & 0.3004 & 0.2258 & 0.1790 & 0.4624 & 0.1834 & 0.1920 & 0.2350\\
         & Sigmoid & 0.7030 & 0.4641 & 0.3072 & 0.7083 & 0.6890 & 0.6727 & 1.0978 & 0.4819 & 0.6780 & 0.6904\\
        \midrule
        \multirow{3}{*}{non-parametric} & Hist & 0.2753 & 0.2717 & 0.3361 & 0.2728 & 0.2769 & 0.6264 & 0.5116 & 0.5412 & 0.6177 & 0.5077\\
         & Isotonic & 0.2726 & 0.2723 & 0.3240 & 0.2728 & 0.2683 & 0.5110 & 0.4953 & 0.4711 & 0.5224 & 0.5386\\
         & BBQ & 0.2725 & 0.2693 & 0.3517 & 0.2749 & 0.2691 & 0.4867 & 0.4790 & 0.4895 & 0.5092 & 0.4813\\
        \midrule
        \multirow{4}{*}{\shortstack{parametric \\ w/ $\mathcal{L}_{\textup{naive}}$}} & Platt & 0.2756 & 0.2700 & 0.3184 & 0.2735 & 0.2671 & 0.4748 & 0.4771 & 0.4747 & 0.4735 & 0.4758\\
         & Beta & 0.2755 & 0.2697 & 0.3250 & 0.2738 & 0.2673 & 0.4766 & 0.4796 & 0.4747 & 0.4741 & 0.4776\\
         & Gaussian & 0.2748 & 0.2676 & 0.3196 & 0.2725 & 0.2671 & 0.4766 & 0.4768 & 0.4755 & 0.4730 & 0.4761\\
         & Gamma & 0.2735 & 0.2699 & 0.3181 & 0.2735 & 0.2672 & 0.4754 & 0.4764 & 0.4739 & 0.4747 & 0.4749\\
        \midrule
          & Platt & 0.2749 & 0.2684 & 0.3034 & 0.2723 & 0.2675 & 0.4746 & 0.4655 & 0.4654 & 0.4716 & 0.4741\\
        parametric & Beta & 0.2747 & 0.2685 & 0.3022 & 0.2721 & 0.2672 & 0.4744 & 0.4684 & 0.4665 & 0.4713 & 0.4751\\
        w/ $\mathcal{L}_{\textup{UERM}}$ & Gaussian & 0.2743 & \textbf{0.2666} & \textbf{0.3021} & \textbf{0.2715} & \textbf{0.2671} & \textbf{0.4735} & \textbf{0.4619} & 0.4653 & \textbf{0.4711} & \textbf{0.4727}\\
         & Gamma & \textbf{0.2723} & 0.2681 & 0.3035 & 0.2720 &0.2674 & 0.4742& 0.4649 &\textbf{0.4651} &\textbf{0.4711} & 0.4733\\
        \bottomrule
    \end{tabular}
    \label{tabnll}
\end{sidewaystable}

\subsection{Comparing Calibration Performance}
Table \ref{calimain}, \ref{tabmce}, and \ref{tabnll} show ECE, MCE, and NLL of each calibration method applied on the various personalized ranking models, respectively.
ECE indicates how well the calibrated probabilities and ground-truth likelihoods match on the test set across all probability ranges.
First, the minmax normalizer and the sigmoid function produce poorly calibrated preference probabilities.
It is obvious because the ranking scores do not have any probabilistic meaning and naively re-scaling them cannot reflect the score distribution.

Second, the parametric methods better calibrate the preference probabilities than the non-parametric methods in most cases.
This is consistent with recent work \cite{cal17, kull2019beyond} for image classification.
The non-parametric calibration methods lack rich expressiveness since they rely on the binning scheme, which maps the ranking scores to the probabilities in a discrete manner.
On the other hand, the parametric calibration methods fit the continuous functions based on the parametric distributions.
Therefore, they have a more granular mapping from the ranking scores to the preference probabilities.

Third, every parametric calibration method benefits from adopting $\mathcal{L}_{\textup{UERM}}$ instead of $\mathcal{L}_{\textup{naive}}$ for the parameter fitting.
The naive log-loss treats all the unobserved pairs as negative pairs and makes the calibration methods produce biased preference probabilities.
On the contrary, inverse propensity-scored log-loss handles such problem and enables us to compute the ideal empirical risk indirectly.
As a result, ECE decreases by 7.40\%-76.52\% for all parametric methods compared to when the naive log-loss is used for the optimization.

Lastly, Gaussian calibration and Gamma calibration with $\mathcal{L}_{\textup{UERM}}$ show the best calibration performance across all base models and datasets.
Platt scaling can be seen as a special case of the proposed methods with $a=0$, so it has less expressiveness in terms of the capacity of parametric family.
Beta distribution is only defined in [0,1], so it cannot represent the unbounded ranking scores.
To adopt Beta calibration, we need to re-scale the ranking score, however, it is not verified for the optimality \cite{iso12}.
As a result, our calibration methods improve ECE by 5.21\%-25.85\% over the best competitor.
Also, since our proposed models have a larger capacity of expressiveness, they show larger improvement on Yahoo!R3, which has more samples to fit the parameters than Coat. 

\begin{figure*}[t!]
\centering 
\subfigure{\includegraphics[width=0.18\linewidth]{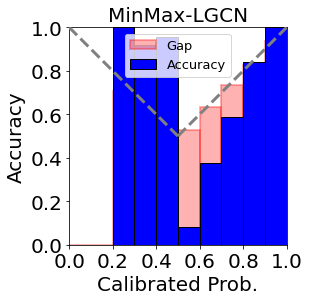}}
\subfigure{\includegraphics[width=0.18\linewidth]{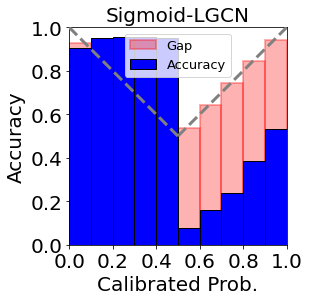}}
\subfigure{\includegraphics[width=0.18\linewidth]{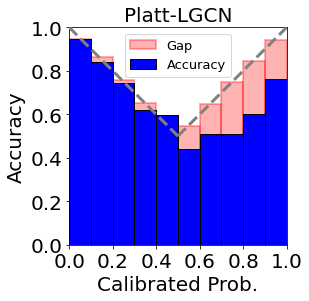}}
\subfigure{\includegraphics[width=0.18\linewidth]{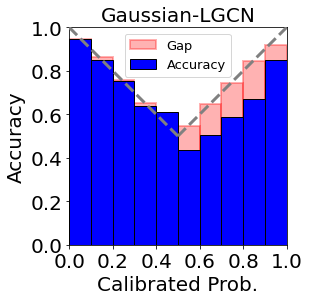}}
\subfigure{\includegraphics[width=0.18\linewidth]{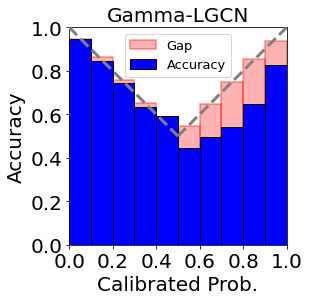}}
\subfigure{\includegraphics[width=0.18\linewidth]{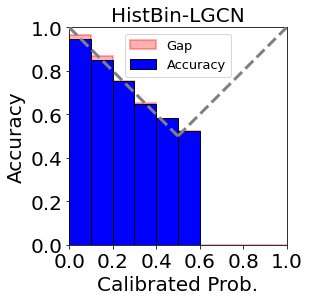}}
\subfigure{\includegraphics[width=0.18\linewidth]{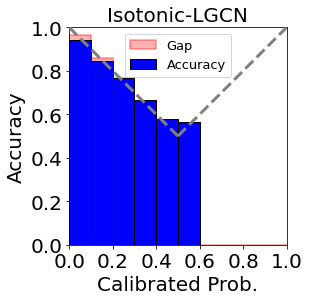}}
\subfigure{\includegraphics[width=0.18\linewidth]{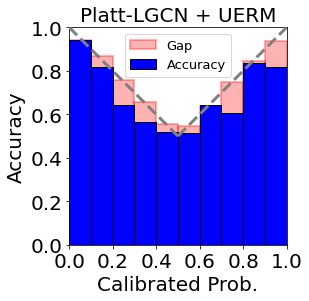}}
\subfigure{\includegraphics[width=0.18\linewidth]{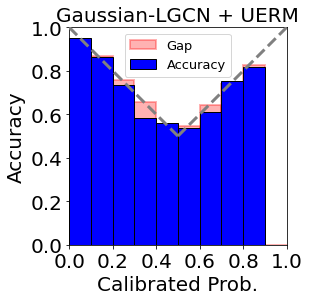}}
\subfigure{\includegraphics[width=0.18\linewidth]{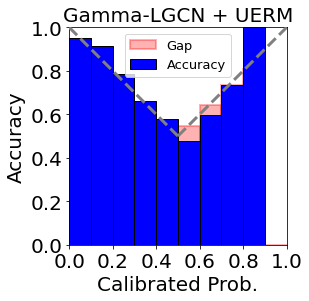}}
\caption{Reliability diagram of each calibration method. Gap denotes the discrepancy between the accuracy and the average calibrated probability for each bin. The grey dashed line is a diagonal function that indicates the ideal reliability line where the blue accuracy bar should meet.}
\label{rd}
\end{figure*}

\subsection{Reliability Diagram}
Figure \ref{rd} shows the reliability diagram \cite{cal17} for each calibration method applied on LGCN for Yahoo!R3.
We partition the calibrated probabilities $g_{\phi}(s)$ into 10 equi-spaced bins and compute the accuracy and the average calibrated probability for each bin (i.e., the first and the second term in Eq.\ref{ecem}, respectively).
The accuracy is the same with the ground-truth proportion of positive samples for the positive bins (i.e., probability over 0.5) and the ground-truth proportion of negative samples for the negative bins (i.e., probability under 0.5).
Note that the bar does not exist if the bin does not have any prediction in it.

First, the non-parametric calibration methods do not produce the probability over 0.6.
It is because they can easily be overfitted to the unbalanced user-item interaction datasets since they do not have any prior distribution. 
On the other hand, the parametric calibration methods produce probabilities across all ranges by avoiding such overfitting problem with the prior parametric distributions.

Second, the parametric calibration methods with UERM produce well-calibrated probabilities especially for the positive preference ($p>0.5$).
The naive log-loss makes the calibration methods biased towards the negative preference, by treating all the unobserved pairs as the negative pairs.
As a result, the parametric methods with the naive log-loss (upper-right three diagrams of Figure 1) show large gaps in the positive probability range ($p>0.5$).
On the contrary, UERM framework successfully alleviates this problem and produces much smaller gaps for the positive preference (lower-right three diagrams of Figure \ref{rd}).
Lastly, it is quite a natural result that parametric methods with UERM do not produce the probability over 0.9, considering that the users prefer only a few items among a large number of items.

\subsection{Score Distribution \& Fitted Function}
\begin{figure}[t]
\centering 
\subfigure{\includegraphics[width=0.4\linewidth]{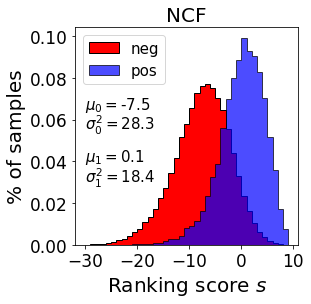}}
\subfigure{\includegraphics[width=0.4\linewidth]{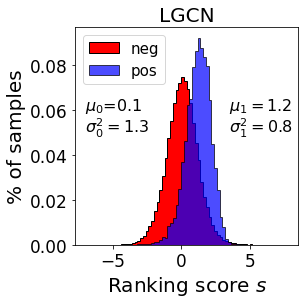}}
\caption{Ranking score distributions of negative and positive pairs.}
\label{sd}
\end{figure}

\begin{figure}[t]
\centering 
\subfigure{\includegraphics[width=0.3\linewidth]{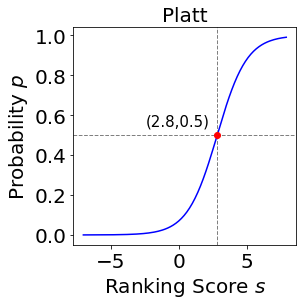}}
\subfigure{\includegraphics[width=0.3\linewidth]{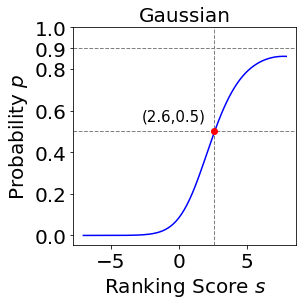}}
\subfigure{\includegraphics[width=0.3\linewidth]{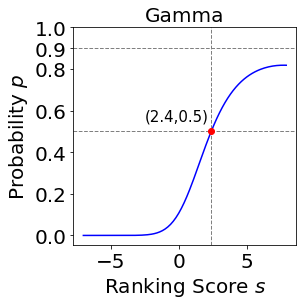}}
\caption{Fitted function of each calibration method.}
\label{fcf}
\end{figure}
Figure \ref{sd} shows the distribution of ranking scores trained by NCF and LGCN on Yahoo!R3.
We can see that the class-conditional score distributions have different deviations ($\sigma_0 > \sigma_1$) and skewed shapes (left tails are longer than the right tails).
This indicates that Platt scaling (or temperature scaling) assuming the same variance for both classes cannot effectively handle these score distributions.
Figure \ref{fcf} shows the fitted calibration function of each parametric method adopted on LGCN and optimized with UERM.
Since most of the user-item pairs are negative in the interaction datasets, all three functions are fitted to produce the low probability under 0.1 for a wide bottom range to reflect the dominant negative preferences.
Platt scaling is forced to have the symmetric shape due to its parametric family, so it produces the high probability over 0.9 which is symmetrical to that of under 0.1.
On the other hand, Gaussian calibration and Gamma calibration, which have a larger expressive power, learn asymmetric shapes tailored to the score distributions having different deviations and skewness.
This result shows that they effectively handle the imbalance of user-item interaction datasets and supports the experimental superiority of the proposed methods.

\subsection{Propensity Estimation}
\begin{figure}[t]
\centering 
\includegraphics[width=0.9\linewidth]{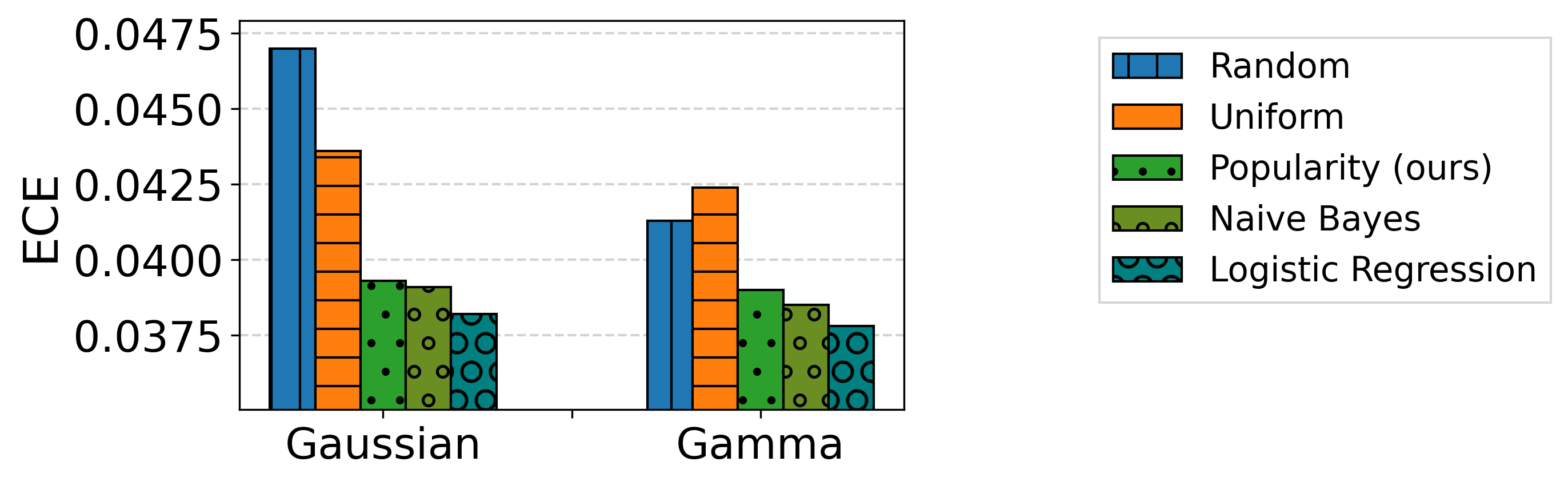}
\caption{ECE with various propensity estimation techniques. In this work, we utilize item popularities to estimate the propensity scores.}
\label{pe}
\end{figure}
Figure \ref{pe} shows the ECE of the proposed methods adopted on BPR for Coat dataset with various propensity estimation techniques.
Random estimation produces random propensity scores from [0,1].
Uniform denotes $\omega=1$, which is equivalent to the naive log-loss.
In this work, we utilize item popularities to estimate the propensity scores.
Naive Bayes exploits the test set to compute the conditional probabilities.
Lastly, Logistic Regression uses user demographics and item categories to estimate the propensity scores.
Note that Naive Bayes and Logistic Regression utilize the additional information that is not available in our setting, and their propensity scores are provided by the author \cite{ips16}.
Surprisingly, merely utilizing popularity is enough for the estimation and it shows a comparable performance with Naive Bayes and Logistic Regression which use additional information.

\subsection{User Degree vs Calibrated Confidence}
\begin{figure}[t]
\centering 
\includegraphics[width=0.45\linewidth]{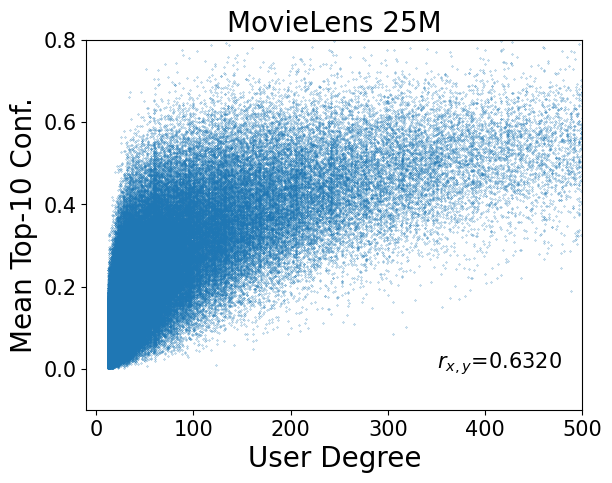}
\includegraphics[width=0.45\linewidth]{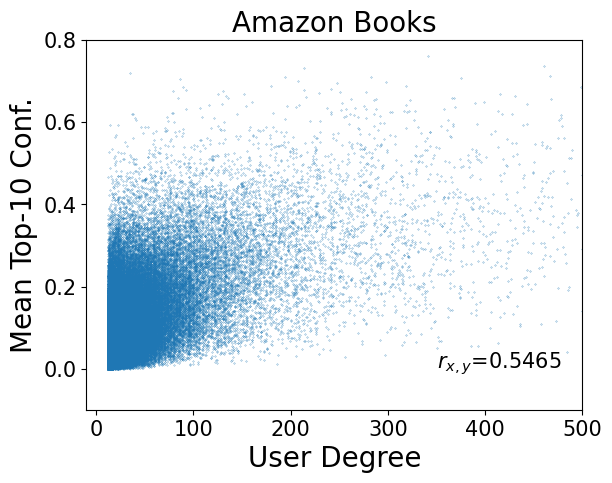}
\caption{User degree versus mean top-10 confidence.}
\label{udcc}
\end{figure}
Figure \ref{udcc} shows the relation between user degree (number of interactions) and mean calibrated confidence on top-10 items.
The base model is BPR.
$r_{x,y}$ denotes the Pearson correlation coefficient.
Since both correlation coefficients are positive and larger than 0.5, there is a positive correlation between user degree and mean top-10 calibrated confidence.
If a user has sufficient interactions, the recommender model can capture stable and certain preferences for that user.
Therefore, generally speaking, the mean confidence on top-10 items increases.
On the other hand, users with fewer interactions also may get confident recommendations (Upper left of Figure \ref{udcc}).
In this case, if a user has a distinct preference for small interactions, the recommender model also may capture certain preferences for that user.
For example, if a user has only three interactions for The Avengers, Avengers: Age of Ultron, and Avengers: Infinity War, the recommender system can provide Avengers: Endgame with very high confidence.

\subsection{Case Study}
\begin{figure}[t]
\centering 
\includegraphics[width=0.65\linewidth]{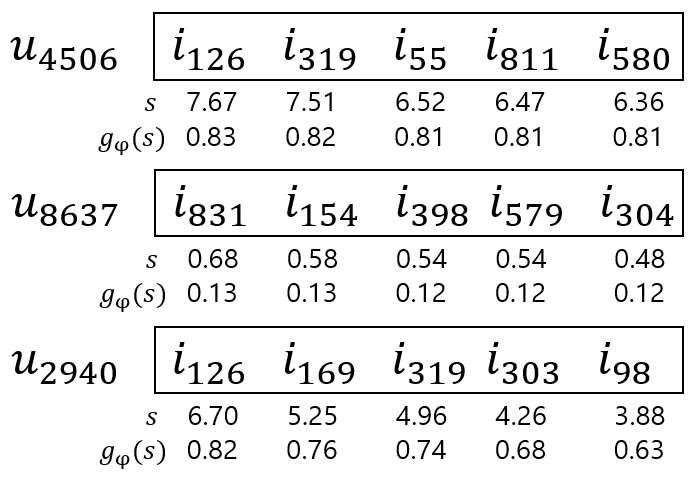}
\caption{Case study. Top-5 items for each user with ranking score $s$ and calibrated probability $g_{\phi}(s)$.}
\label{cs}
\end{figure}
Figure \ref{cs} shows the case study on Yahoo!R3 with Gaussian calibration adopted on LGCN.
The personalized ranking model first learns the ranking scores and produces a top-5 ranking list for each user.
Then, Gaussian calibration transforms the ranking scores to the well-calibrated preference probabilities.
For the first user $u_{4506}$, the method produces high preference probabilities for all top-5 items.
In this case, we can recommend them to him with confidence.
On the other hand, for the second user $u_{8637}$, all top-5 items have low preference probabilities, and the last user $u_{2940}$ has a wide range of preference probabilities.
For these users, merely recommending all the top-ranked items without consideration of potential preference degrade their satisfaction.
It is also known that the unsatisfactory recommendations even make the users leave the platform \cite{leave06}.
Therefore, instead of recommending items with low confidence, the system should take other strategies, such as requesting additional user feedback \cite{dre}.

\section{Conclusion}
In this paper, we aim to obtain calibrated probabilities with personalized ranking models.
We investigate various parametric distributions and propose two parametric calibration methods, namely Gaussian calibration and Gamma calibration.
We also design the unbiased empirical risk minimization framework that helps the calibration methods to be optimized towards true preference probability with the biased user-item interaction dataset.
Our extensive evaluation demonstrates that the proposed methods and framework significantly improve calibration metrics and have a richer expressiveness than existing methods.
Lastly, our case study shows that the calibrated probability provides an objective criterion for the reliability of recommendations, allowing the system to take various strategies to increase user satisfaction.

%% file: 3.BD.tex
Recommender systems (RS) have started to employ knowledge distillation, which is a model compression technique training a compact model (student) with the knowledge transferred from a cumbersome model (teacher).
The state-of-the-art methods rely on unidirectional distillation transferring the knowledge only from the teacher to the student, with an underlying assumption that the teacher is always superior to the student.
However, we demonstrate that the student performs better than the teacher on a significant proportion of the test set, especially for RS.
Based on this observation, we propose \textit{Bidirectional Distillation} (BD) framework whereby both the teacher and the student collaboratively improve with each other.
Specifically, each model is trained with the distillation loss that makes to follow the other’s prediction along with its original loss function.
For effective bidirectional distillation, we propose \textit{rank discrepancy-aware sampling} scheme to distill only the informative knowledge that can fully enhance each other.
The proposed scheme is designed to effectively cope with a large performance gap between the teacher and the student.
Trained in the bidirectional way, it turns out that both the teacher and the student are significantly improved compared to when being trained separately.
Our extensive experiments on real-world datasets show that our proposed framework consistently outperforms the state-of-the-art competitors.
We also provide analyses for an in-depth understanding of BD and ablation studies to verify the effectiveness of each proposed component.
This work was published at The Web Conference (WWW 2021) with an oral presentation \cite{kweon2021bidirectional}.

\section{Introduction}
Nowadays, the size of recommender systems (RS) is continuously increasing, as they have adopted deep and sophisticated model architectures to better understand the complex relationships between users and items \cite{rd18, cd19}.
A large recommender with many learning parameters usually has better performance due to its high capacity, but it also has high computational costs and long inference time.  
This problem is exacerbated for web-scale applications having numerous users and items, since the number of the learning parameters increases proportionally to the number of users and items.
Therefore, it is challenging to adopt such a large recommender for real-time and web-scale platforms. 

To tackle this problem, a few recent work \cite{rd18, cd19, DERRD} has adopted \textit{Knowledge Distillation} (KD) to RS.
KD is a model-agnostic strategy that trains a compact model (student) with the guidance of a pre-trained cumbersome model (teacher).
The distillation is conducted in two stages;
First, the large teacher recommender is trained with user-item interactions in the training set with binary labels (i.e., $0$ for unobserved interaction and $1$ for observed interaction.).
Second, the compact student recommender is trained with the recommendation list predicted by the teacher along with the binary training set. 
Specifically, in \cite{rd18}, the student is trained to give high scores on the top-ranked items of the teacher's recommendation list. 
Similarly, in \cite{cd19}, the student is trained to imitate the teacher's prediction scores with particular emphasis on the high-ranked items in the teacher's recommendation list.
The teacher's predictions provide additional supervision, which is not explicitly revealed from the binary training set, to the student.
By distilling the teacher's knowledge, the student can achieve comparable performance to the teacher with fewer learning parameters and lower inference latency.

Despite their effectiveness, they still have some limitations.
First, they rely on \textit{unidirectional} knowledge transfer, which distills knowledge only from the teacher to the student, with an underlying assumption that the teacher is always superior to the student.
However, based on the in-depth analyses provided in Section 2, we argue that the knowledge of the student also could be useful for the teacher.
Indeed, we demonstrate that the teacher is not always superior to the student and further the student performs better than the teacher on a significant proportion of the test set.
In specific, the student recommender better predicts 36$\sim$42\% of the test interactions than the teacher recommender.
These proportions are remarkably large compared to 4$\sim$9\% from our experiment on the computer vision task.
In this regard, we claim that both the student and the teacher can take advantage of each other's complementary knowledge, and be improved further.
Second, the existing methods have focused on distilling knowledge of items ranked highly by the teacher.
However, we observe that most of the items ranked highly by the teacher are already ranked highly by the student (See Section 2 for details).
Therefore, merely high-ranked items may not be informative to fully enhance the other model, which leads to limiting the effectiveness of the distillation.

We propose a novel \textbf{\underline{B}}idirectional \textbf{\underline{D}}istillation (BD) framework for RS.
Unlike the existing methods that transfer the knowledge only from the pre-trained teacher recommender, both the teacher and the student transfer their knowledge to each other within our proposed framework.
Specifically, they are trained simultaneously with the distillation loss that makes to follow the other’s predictions along with the original loss function.
Trained in the bidirectional way, it turns out that both the teacher and the student are significantly improved compared to when being trained separately.
In addition, the student recommender trained with BD achieves superior performance compared to the student trained with conventional distillation from a pre-trained teacher recommender.

For effective bidirectional distillation, the remaining challenge is to design an \textit{informative} distillation strategy that can fully enhance each other, considering the different capacities of the teacher and the student.
As mentioned earlier, items merely ranked highly by the teacher cannot give much information to the student.
Also, it is obvious that the student's knowledge is not always helpful to improve the teacher, as the teacher has a much better overall performance than the student.
To tackle this challenge, we propose \textit{rank discrepancy-aware sampling} scheme differently tailored for the student and the teacher.
In the scheme, the probability of sampling an item is defined based on its rank discrepancy between the two recommenders.
Specifically, each recommender focuses on learning the knowledge of the items ranked highly by the other recommender but ranked lowly by itself.
Taking into account the performance gap between the student and the teacher, we enable the teacher to focus on the items that the student has very high confidence in, whereas making the student learn the teacher's broad knowledge on more diverse items.

The proposed BD framework can be applicable in many application scenarios.
Specifically, it can be used to maximize the performance of the existing recommender in the scenario where there is no constraint on the model size (in terms of the teacher), or to train a small but powerful recommender as targeted in the conventional KD methods (in terms of the student).
The key contributions of our work are as follows:
\begin{itemize}[leftmargin=*]
    \item Through our exhaustive analyses on real-world datasets, we demonstrate that the knowledge of the student recommender also could be useful for the teacher recommender.
    Based on the results, we also point out the limitation of the existing distillation methods for RS under the assumption that the teacher is always superior to the student.
    \item We propose a novel bidirectional KD framework for RS, named BD, enabling that the teacher and the student can collaboratively improve with each other during the training.
    BD also adopts the rank discrepancy-aware sampling scheme differently designed for the teacher and the student considering their capacity gap.
    \item We validate the superiority of the proposed framework by extensive experiments on real-world datasets. 
    BD considerably outperforms the state-of-the-art KD competitors.
    We also provide both qualitative and quantitative analyses to verify the effectiveness of each proposed component\footnote{We provide the source code of BD at \url{https://github.com/WonbinKweon/BD_WWW2021}}.
\end{itemize}

\section{Analysis: Teacher vs Student}
In this section, we provide detailed analyses of comparison between the teacher and the student on two real-world datasets: CiteULike \cite{CUL13} and Foursquare \cite{FS14}.
Following \cite{cd19}, we hold out the last interaction of each user as test interaction, and the rest of the user-item interactions are used for training data
(see Section 3.5.1 for details of the experiment setup).
We use NeuMF \cite{ncf17} as a base model, and adopt NeuMF-50 as the teacher and NeuMF-5 as the student. The number indicates the dimension of the user and item embedding; 
the number of learning parameters of the teacher is 10 times bigger than that of the student, and accordingly, the teacher shows a much better overall performance than the student (reported in Table ???).
Note that the teacher and the student are trained separately without any distillation technique in the analyses.

\begin{figure*}[ht]
    \includegraphics[width=0.5\linewidth]{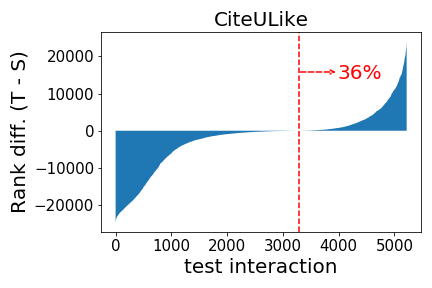}
    \includegraphics[width=0.5\linewidth]{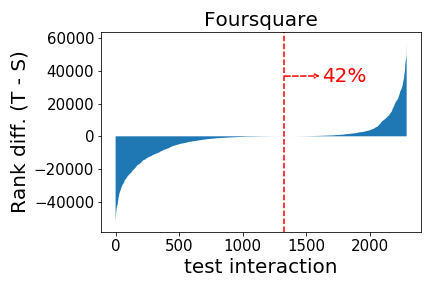}
    \caption{Rank difference between teacher and student. Rank differences are computed on the test interactions and the test interactions (x-axis) are sorted in increasing order by the rank difference.}
\label{TvsS1}
\end{figure*}
\textbf{The teacher is not always superior to the student.}
Figure \ref{TvsS1} shows bar graphs of the rank difference between the teacher and the student on the test interactions (the bars in the graphs are sorted in increasing order along with the x-axis).
The rank difference on a test interaction $(u, i)$ is defined as follows:
\begin{equation}
\text{Rank diff.}^{u}(i) = rank_T^{u}(i) - rank_S^{u}(i),
\end{equation}
where $rank_T^{u}(i)$ and $rank_S^{u}(i)$ denote the ranks assigned by the teacher and the student on the item $i$ for the user $u$ respectively, and $rank_{*}^{u}(i)=1$ is the highest ranking.
If the rank difference is bigger than zero, it means that the student assigns a higher rank than the teacher on the test interaction; the student better predicts the test interaction than the teacher.
We observe that the student model assigns higher ranks than the teacher model on the significant proportion of the entire test interactions: 36\% for CiteULike and 42\% for Foursquare.
These proportions are remarkably large compared to 4\% for CIFAR-10 and 9\% for CIFAR-100 from our experiments on computer vision task\footnote{We compare the predictions of the teacher and the student on the test set. We adopt ResNet-80 as the teacher and ResNet-8 as the student \cite{resnet16} and train them separately for the image classification task on CIFAR-10 and CIFAR-100 dataset \cite{cifar10}.}.

The results raise a few questions.
Why the student can perform better than the teacher on the test data and why this phenomenon is intensified for RS?
We provide some possible reasons as follows:
Firstly, not all user-item interactions require sophisticated calculations in high-dimensional space for correct prediction.
As shown in the computer vision \cite{overthinking19}, a simple image without complex background or patterns can be better predicted at the lower layer than the final layer in a deep neural network.
Likewise, some user-item relationships can be better captured based on simple operations without expensive computations.
Moreover, unlike the image classification task, RS has very high ambiguity in nature;
in many applications, the user’s feedback on an item is given in the binary form: 1 for observed interaction, 0 for unobserved interaction.
However, the binary labels do not explicitly show the user's preference.
For example, “0” does not necessarily mean the user’s negative preference for the item.
It can be that the user may not be aware of the item.
Due to the ambiguity in the ground-truth labels, the lower training loss does not necessarily guarantee a better ranking performance.
Specifically, in the case of NeuMF, which is trained with the binary cross-entropy loss,
the teacher can better minimize the training loss and thus better discriminate the observed/unobserved interactions in the training set. 
However, this may not necessarily result in better predicting user's preference on all the unobserved items due to the ambiguity of supervision.

\begin{figure*}[h!]
    \includegraphics[width=0.5\linewidth]{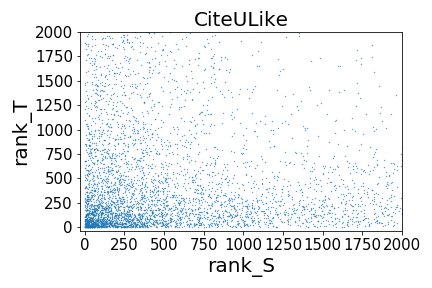}
    \includegraphics[width=0.5\linewidth]{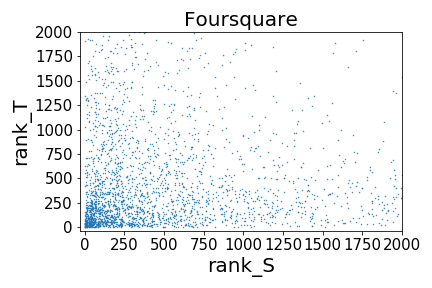}
    \caption{Rank comparison between teacher and student. Each dot represents a sampled unobserved interaction.}
\label{TvsS2}
\end{figure*}
\textbf{The high-ranked items are not that informative as expected.}
Figure \ref{TvsS2} shows the rank comparison between the teacher and the student.
We plot the high-ranked items in the recommendation list from the teacher and that from the student for all users\footnote{We sample the items by using the rank-aware sampling scheme suggested in \cite{cd19}. Note that in \cite{cd19}, the sampled items are used for the distillation.}.
Each point corresponds to $(rank^u_S(i), rank^u_T(i))$.
We observe that most of the points are located near the lower-left corner, which means that most of the items ranked highly by the teacher are already ranked highly by the student.
In this regard, unlike the motivation of the previous work \cite{rd18, cd19}, items merely ranked highly by the teacher may be not informative enough to fully enhance the student, and thus may limit the effectiveness of the distillation.
We argue that each model can be further improved by focusing on learning the knowledge of items ranked highly by the other recommender but ranked lowly by itself (e.g., points located along the x-axis and y-axis).

In summary, the teacher is not always superior to the student, so the teacher can also learn from the student.
Their complementarity should be importantly considered for effective distillation in RS.
Also, the distillation should be conducted with consideration of the rank discrepancy between the two recommenders.
Based on the observations, we are motivated to design a KD framework that the teacher and the student can collaboratively improve with each other based on the rank discrepancy.

\section{Problem Formulation}
Let the set of users and items be denoted as $\mathcal{U}=\{u_1, u_2,...,u_n\}$ and $\mathcal{I}=\{ i_1, i_2,...,i_m \}$, where $n$ is the number of users and $m$ is the number of items.
Let $\textbf{R} \in \{0,1\}^{n \times m}$ be the user-item interaction matrix, 
where $r_{ui}=1$ if the user $u$ has an interaction with the item $i$, otherwise $r_{ui}=0$.
Also, for a user $u$, $\mathcal{I}_{u}^{-} = \{ i \in \mathcal{I} | r_{ui} = 0\}$ denotes the set of unobserved items and $\mathcal{I}_{u}^{+} = \{ i \in \mathcal{I} | r_{ui} = 1\}$ denotes the set of interacted items.
A top-$K$ recommender system aims to find a recommendation list of unobserved items for each user.
To make the recommendation list, the system predicts the score $\hat{r}_{ui} = P(r_{ui}=1|u,i)$ for each item $i$ in $\mathcal{I}_{u}^{-}$ for each user $u$, then ranks the unobserved items according to their scores.

\section{Method}
\begin{figure*}[h!]
\includegraphics[width=1\linewidth]{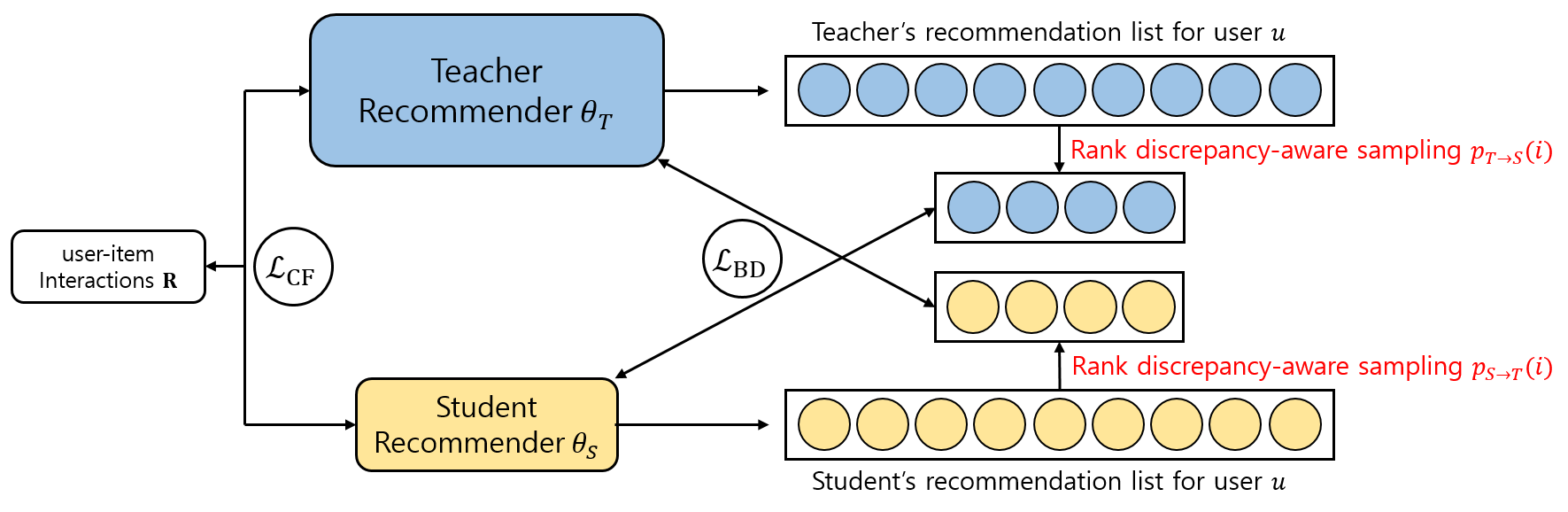}
\caption{Illustration of Bidirectional Distillation (BD) for top-$K$ recommender systems.}
\label{bd}
\end{figure*}

We propose a novel \textbf{\underline{B}}idirectional \textbf{\underline{D}}istillation (BD) framework for top-$K$ recommender systems.
Within our proposed framework, both the teacher and the student transfer their knowledge to each other.
We first provide an overview of BD (Section 3.4.1).
Then, we formalize the distillation loss to transfer the knowledge between the two recommenders (Section 3.4.2).
We also propose \textit{rank discrepancy-aware sampling} scheme differently tailored for the teacher and the student to fully enhance each other (Section 3.4.3).
Lastly, we provide the details of the end-to-end training process within the proposed framework (Section 3.4.4).

\subsection{Overview}
Figure \ref{bd} shows an overview of BD.
Unlike the existing KD methods for RS, both the teacher and the student are trained simultaneously by using each other's knowledge (i.e., recommendation list) along with the binary training set (i.e., user-item interactions).
First, the teacher and the student produce the recommendation list for each user.
Second, BD decides what knowledge to be transferred for each distillation direction based on the rank discrepancy-aware sampling scheme.
Lastly, the teacher and the student are trained with the distillation loss along with the original collaborative filtering loss.
To summarize, each of them is  trained as follows:
\begin{equation}
\begin{split}
& \mathcal{L}_T(\theta_T)  = \mathcal{L}_{CF}(\theta_T) + \lambda_{S \rightarrow T} \cdot \mathcal{L}_{BD}(\theta_T; \theta_S), \\
& \mathcal{L}_S(\theta_S)  = \mathcal{L}_{CF}(\theta_S) + \lambda_{T \rightarrow S} \cdot \mathcal{L}_{BD}(\theta_S; \theta_T),
\end{split}
\end{equation}
where $T$ and $S$ denote the teacher and the student respectively, $\theta_*$ is the model parameters.
$\mathcal{L}_{CF}$ is the collaborative filtering loss depending on the base model which can be any existing recommender, and $\mathcal{L}_{BD}$ is the bidirectional distillation loss.
Lastly, $\lambda_{S \rightarrow T}$ and $\lambda_{T \rightarrow S}$ are the hyperparameters that control the effects of the distillation loss in each direction.

Within BD, the teacher and the student are collaboratively improved with each other based on their complementarity.
Also, during the training, the knowledge distilled between the teacher and the student gets gradually evolved along with the recommenders;
the improvement of the teacher leads to the acceleration of the student's learning, and the accelerated student again improves the teacher's learning.
Trained in the bidirectional way, both the teacher and the student are significantly improved compared to when being trained separately.
As a result, the student trained with BD outperforms the student model trained with the conventional distillation that relies on a pre-trained and fixed teacher.

\subsection{Distillation Loss}
We formalize the distillation loss that transfers the knowledge between the recommenders.
By following the original distillation loss that matches the class distributions of two classifiers for a given image \cite{kd15},
we design the distillation loss for the user $u$ as follows: 
\begin{equation}
\begin{aligned}
\mathcal{L}_{BD}(\theta_T; \theta_S) &=  \sum_{j \in \mathcal{RDS}_{S \rightarrow T}(\mathcal{I}_{u}^{-})} \mathcal{L}_{BCE}(\hat{r}^T_{uj}, \hat{r}^S_{uj})\\
\mathcal{L}_{BD}(\theta_S; \theta_T) &=  \sum_{j \in \mathcal{RDS}_{T \rightarrow S}(\mathcal{I}_{u}^{-})} \mathcal{L}_{BCE}(\hat{r}^S_{uj}, \hat{r}^T_{uj}),
\end{aligned}
\end{equation}
where $\mathcal{L}_{BCE}(p, q) = q \log p + (1-q) \log (1-p)$ is the binary cross-entropy loss,
$\hat{r}_{uj} = P(r_{uj}=1 | u,j)$ is the prediction of a recommender and $\mathcal{RDS}_{*}(\mathcal{I}_{u}^{-})$ is a set of the unobserved items sampled by the rank discrepancy-aware sampling.
$\hat{r}_{uj}$ is computed by $\sigma(z_{uj}/T)$
where $\sigma(\cdot)$ is the sigmoid function, $z_{uj}$ is the logit, and $T$ is the temperature that controls the smoothness.

Our distillation loss is a binary version of the original KD loss function.
Similar to the original KD loss transferring the knowledge of class probabilities, our loss transfers a user's potential positive and negative preferences on the unobserved items.
Specifically, in the binary training set, the unobserved interaction $r_{uj}$ is only labeled as “0”.
However, as mentioned earlier, it is ambiguous whether the user actually dislikes the item or potentially likes the item.
Through the distillation, each recommender can get the other's opinion of how likely (and unlikely) the user would be interested in the item, and such information helps the recommenders to better cope with the ambiguous nature of RS.

\subsection{Rank Discrepancy-aware Sampling}
We propose the rank discrepancy-aware sampling scheme that decides what knowledge to be transferred for each
distillation direction.
As we observed in Section 2, most of the items ranked highly by the teacher are already ranked highly by the student and vice versa.
Thus, the existing methods \cite{cd19,rd18} that simply choose the high-ranked items cannot give enough information to the other recommender.
Moreover, for effective bidirectional distillation, the performance gap between the teacher and the student should be carefully considered in deciding what knowledge to be transferred, as it is obvious that not all knowledge of the student is helpful to improve the teacher.

In this regard, we develop a sampling scheme based on the rank discrepancy of the teacher and the student, and tailor it differently for each distillation direction with consideration of their different capacities.
The underlying idea of the scheme is that each recommender can get informative knowledge by focusing on the items ranked highly by the other recommender, but ranked lowly by itself.
The sampling strategy for each distillation direction is defined as follows:

\textbf{Distillation from the teacher to the student.}
As the teacher has a much better overall performance than the student, the opinion of the teacher should be considered more reliable than that of the student in most cases.
Thus, for this direction of the distillation, we make the student follow the teacher's predictions on many rank-discrepant items.
Formally, for each user $u$, the probability of an item $i$ to be sampled is computed as follows:
\begin{equation}
p_{T \rightarrow S}(i) \propto tanh(\text{max}((rank_S(i) - rank_T(i)) \cdot \epsilon_t, 0)),
\end{equation}
where $rank_T(i)$ and $rank_S(i)$ denote the ranks assigned by the teacher and the student on the item $i$ for the user $u$ respectively, and $rank_{*}(i)=1$ is the highest ranking\footnote{we omit the superscript $u$ from $rank_{*}^{u}(i)$ for the simplicity.}.
We use a hyper-parameter $\epsilon_t$ $(> 0)$ to control the smoothness of the probability.
With this probability function, we sample the items ranked highly by the teacher but ranked lowly by the student.
Since $tanh(\cdot)$ is a saturated function, items with rank discrepancy above a particular threshold would be sampled almost uniformly.
As a result, the student learns the teacher's broad knowledge of most of the rank-discrepant items.

\textbf{Distillation from the student to the teacher.}
As shown in Section 3.2, the teacher is not always superior to the student, especially for RS.
That is, the teacher can be also further improved by learning from the student.
However, at the same time, the large performance gap between the teacher and the student also needs to be considered for effective distillation.
For this direction of the distillation, we make the teacher follow the student's predictions on only a few selectively chosen rank-discrepant items.
The distinct probability function is defined as follows:
\begin{equation}
p_{S \rightarrow T}(i) \propto exp((rank_T(i) - rank_S(i)) \cdot \epsilon_e),
\end{equation}
where $\epsilon_e$ $(> 0)$ is a hyper-parameter to control the smoothness of the probability.
We use the exponential function to put particular emphasis on the items that have large rank discrepancies.
Therefore, this probability function enables the teacher to follow the student's predictions only on the rank-discrepant items that the student has very high confidence in.

\subsection{Model Training}

\begin{algorithm}[t]
\SetKwInOut{Input}{Input}
\SetKwInOut{Output}{Output}
\Input{Training data $\mathcal{D}$, the number of total epochs $e$, rank updating period $p$}
\Output{Teacher model $(\theta_T)$, Student model $(\theta_S)$}
Warm up $\theta_T$ and $\theta_S$ with only $\mathcal{L}_{CF}$\\
\For{$t=0,1,...,(e-1)$}{
\If{$t \text{ } \% \text{ } p == 0$}{
Teacher and Student update their recommendation lists
}
\For{$(u,i) \in \mathcal{D}$}{
\BlankLine
\tcc{Train Teacher}
Draw $\mathcal{RDS}_{S \rightarrow T}(\mathcal{I}_{u}^{-})$ with probability $p_{S \rightarrow T}(\cdot)$ \\
Compute $\mathcal{L}_{CF}(\theta_T)$ and $\mathcal{L}_{BD}(\theta_T; \theta_S)$ \\
Update $\theta_T$
\BlankLine
\tcc{Train Student}
Draw $\mathcal{RDS}_{T \rightarrow S}(\mathcal{I}_{u}^{-})$ with probability $p_{T \rightarrow S}(\cdot)$ \\
Compute $\mathcal{L}_{CF}(\theta_S)$ and $\mathcal{L}_{BD}(\theta_S; \theta_T)$ \\
Update $\theta_S$
}
}
\caption{Bidirectional Distillation Framework.}
\label{BDF}
\end{algorithm}

Algorithm \ref{BDF} describes a pseudo code for the end-to-end training process within BD framework.
The training data $\mathcal{D}$ consists of observed interactions $(u,i)$.
First, the model parameters $\theta_T$ and $\theta_S$ are warmed up only with the collaborative filtering loss (line 1), as the predictions during the first few epochs are very unstable.
Second, we make the recommenders produce the recommendation lists for the subsequent sampling.
Since it is time-consuming to produce the recommendation lists every epoch, we conduct this step every $p$ epochs (line 3-4).
Next, we decide what knowledge to be transferred in each distillation direction via the rank discrepancy-aware sampling.
It is worth noting that the unobserved items sampled by the rank discrepancy-aware sampling can be used also for the collaborative filtering loss.
Finally, we compute the losses with the sampled items for the teacher and the student, respectively, and update the model parameters.

\section{Experiments}
In this section, we validate our proposed framework on 9 experiment settings (3 real-world datasets $\times$ 3 base models).
We first introduce our experimental setup (Section 3.5.1).
Then, we provide a performance comparison supporting the superiority of the BD (Section 3.5.2).
We also provide two analyses for the in-depth understanding of BD (Section 3.5.3, 3.5.4).
Lastly, We provide an ablation study to verify the effectiveness of rank discrepancy-aware sampling (Section 3.5.5) and analyses for important hyperparameters of BD (Section 3.5.6).

\subsection{Experimental Setup}
\textbf{Datasets.} We use three real-world datasets: CiteULike\footnote{https://github.com/changun/CollMetric/tree/master/citeulike-t} \cite{CUL13}, Foursquare\footnote{https://sites.google.com/site/yangdingqi/home/foursquare-dataset} (Tokyo Check-in) \cite{FS14} and Yelp\footnote{https://github.com/hexiangnan/sigir16-eals/blob/master/data/yelp.rating} \cite{yelp16}.
We only keep users who have at least five ratings for CiteULike and Foursquare, ten ratings for Yelp as done in \cite{ncf17, bpr09}.
Data statistics after the preprocessing are presented in Table \ref{BDD}.
We also report the experimental results on ML100K and AMusic, which are used for CD \cite{cd19} for the direct comparison.

\begin{table}[t]
    \centering
  \caption{Data Statistics}
  \begin{tabular}{ccccc}
    \toprule
    Dataset & \#Users & \#Items & \#Ratings & Sparsity \\
    \midrule
    CiteULike & 5,219 & 25,187 & 130,788 & 99.90\% \\
    Foursquare & 2,293 & 61,858 & 537,167 & 99.62\% \\
    Yelp & 25,677 & 25,815 & 730,623 & 99.89\% \\
    \bottomrule
  \end{tabular}
  \label{BDD}
\end{table}

\textbf{Evaluation Protocol and Metrics.}
We adopt the widely used \textit{leave-one-out} evaluation protocol.
For each user, we hold out the last interacted item for testing and the second last interacted item for validation as done in \cite{cd19, ncf17}.
If there is no timestamp in the dataset, we randomly take two observed items for each user.
Then, we evaluate how well each method can rank the test item higher than all the unobserved items for each user (i.e., $\mathcal{I}_u^{-}$).
Note that instead of randomly choosing a predefined number of candidates (e.g., 99), we adopt the full-ranking evaluation that uses all the unobserved items as candidates.
Although it is time-consuming, it enables a more thorough evaluation compared to using random candidates \cite{fullrank20, cd19}.

As we focus on top-$K$ recommendation for implicit feedback, we employ two widely used metrics for evaluating the ranking performance of recommenders: Hit Ratio (H@$K$) \cite{hr16} and Normalized Discounted Cumulative Gain (N@$K$) \cite{NDCG02}.
H@$K$ measures whether the test item is present in the top-$K$ list and N@$K$ assigns a higher score
to the hits at higher rankings in the top-$K$ list.
We compute those two metrics for each user, then compute the average score.
Lastly, we report the average value of five independent runs for all methods.

\textbf{Base Models.}
BD is a model-agnostic framework applicable for any top-$K$ RS.
We validate BD with three base models that have different model architectures and learning strategies.
Specifically, we choose two widely used deep learning models and one latent factor model as follows:
\begin{itemize}
    \item \textbf{NeuMF \cite{ncf17}}: A deep recommender that adopts Matrix Factorization (MF) and Multi-Layer Perceptron (MLP) to capture complex and non-linear user-item relationships. 
    NeuMF uses the point-wise loss function for the optimization.
    \item \textbf{CDAE \cite{cdae16}}: 
    A deep recommender that adopts Denoising Autoencoders (DAE) \cite{dae08} for the collaborative filtering. 
    CDAE uses the point-wise loss function for the optimization.
    \item \textbf{BPR \cite{bpr09}}: A learing-to-rank recommender that adopts MF \cite{mf09} to model the user-item interaction.
    BPR uses the pair-wise loss function for the optimization under the assumption that observed items are more preferred than unobserved items.
\end{itemize}

\textbf{Methods Compared.}
The proposed framework is compared with the state-of-the-art KD methods for top-$K$ RS.
\begin{itemize}
    \item \textbf{Ranking Distillation (RD) \cite{rd18}}: A pioneering KD method for top-$K$ RS.
    RD makes the student give high scores on top-ranked items by the teacher.
    \item \textbf{Collaborative Distillation (CD) \cite{cd19}}: A state-of-the-art KD method for top-$K$ RS. 
    CD makes the student imitate the teacher's scores on the items ranked highly by the teacher.
\end{itemize}

\textbf{Implementation Details.}
For all the base models and baselines, we use PyTorch \cite{pytorch19} for the implementation.
For each dataset, hyperparameters are tuned by using grid searches on the validation set. 
We use Adam optimizer \cite{adam14} with L2 regularization and the learning rate is chosen from $\{$0.00001, 0.0001, 0.001, 0.002$\}$, and we set the batch size as 128.
For NeuMF, we use 2-layer MLP for the network.
For CDAE, we use 2-layer MLP for the encoder and the decoder, and the dropout ratio is set to 0.5.
The number of negative samples is set to 1 for NeuMF and BPR, $5*|\mathcal{I}_{u}^{+}|$ for CDAE as suggested in the original paper \cite{cdae16}.

For the distillation, we adopt as many learning parameters as possible for the teacher model until the ranking performance is no longer increased on each dataset.
Then, we build the student model by employing only one-tenth of the learning parameters used by the teacher.
The number of model parameters of each base model is reported in Table 3.
For KD competitors (i.e., RD, CD), $\lambda_{KD}$ is chosen from $\{$0.01, 0.1, 0.5, 1$\}$, the number of items sampled for the distillation is chosen from $\{$10, 15, 20, 30, 40, 50$\}$, and the temperature $T$ for logits is chosen from $\{$1, 1.5, 2$\}$.
For other hyperparameters, we use the values recommended from the public implementation and the original papers \cite{rd18, cd19}.
For BD, $\lambda_{T \rightarrow S}$ and $\lambda_{S \rightarrow T}$ are set to 0.5, the number of items sampled by rank discrepancy-aware sampling is chosen from $\{$1, 5, 10$\}$, $\epsilon_t$ is chosen from $\{10^{-2}, 10^{-3}, 10^{-4}, 10^{-5}\}$, $\epsilon_e$ is chosen from $\{10^{-3}, 10^{-4}, 10^{-5}\}$, the temperature $T$ for logits is chosen from $\{$2, 5, 10$\}$ and the rank updating period $p$ (in Algorithm 1) is set to 10.

\input{table/5exzmaintable}
\subsection{Performance Comparison}
Table \ref{bdmain} shows the recommendation performance of each KD method on three real-world datasets and three different base models.
In Table \ref{bdmain}, "Teacher" and "Student" indicate the base models trained separately without any distillation technique, "BD-Teacher" and "BD-Student" are the teacher and the student trained simultaneously with BD.
"CD" and "RD" denote the student trained with CD and RD, respectively.
We analyze the experimental results from various perspectives.

We first observe that the teacher recommender is consistently improved by the knowledge transferred from the student with BD by up to 20.79\% (for H@50 on CiteULike).
This result verifies our claim that the knowledge of the student also could be useful for improving the teacher.
Also, this result strongly indicates that the distillation process of BD effectively resolves the challenge of a large performance gap, and successfully transfer the informative knowledge from the student.
In specific, the ranking discrepancy-aware sampling enables the teacher to focus on the items that the student has very high confidence in.
This strategy can minimize the adverse effects from the performance gap, making the teacher be further improved based on the complementary knowledge from the student.

Second, the student recommender is significantly improved by the knowledge transferred from the teacher with BD by up to 39.88\% (for N@50 on CiteULike).
Especially, the student trained with BD considerably outperforms the student trained with the existing KD methods (i.e., RD, CD) by up to 13.94\% (for H@100 on Yelp).
The superiority of BD comes from the two contributions;
BD makes the student follow the teacher’s predictions on the rank-discrepant items which are more informative than the merely high-ranked items.
Also, within BD, the improvement of the teacher leads to the acceleration of the student’s learning, and the accelerated student again improves the teacher’s learning.
With better guidance from the improved teacher, the student with BD can achieve superior performance than the student with the existing KD methods.
To verify the effectiveness of each contribution, we conduct in-depth analyses in the next sections.

\input{table/5exzmodelparams}
Lastly, Table \ref{bdinfer} shows the model size and online inference efficiency for each base model.
For making the inferences, we use PyTorch with CUDA on GTX Titan X GPU and Intel i7-7820X CPU.
As the student has only one-tenth of the learning parameters used by the teacher, the student requires less computational costs, thus achieves lower inference latency.
Moreover, the student trained with BD shows comparable or even better recommendation performance to the teacher (e.g., NeuMF and CDAE on Foursquare).
On CiteULike, we observe that the student achieves comparable performance by employing 20$\sim$30\% of learning parameters used by the teacher.
These results show that BD can be effectively adopted to train a small but powerful recommender.
Moreover, as mentioned earlier, BD is also applicable in setting where there is no constraint on the model size or inference time.
Specifically, BD can be adopted to maximize the performance of a large recommender with numerous learning parameters (i.e., the teacher).

\subsection{Model Size Analysis}
We control the size of the teacher and the student to analyze the effects of the capacity gap between two recommenders on BD.
We report the results of a deep model and a latent factor model (i.e., NeuMF and BPR) on CiteULike.
Figure \ref{bdmsize1} shows the performance of the teacher trained with the student of varying sizes and Figure \ref{bdmsize2} shows the performance of the student trained with the teacher of varying sizes.
“S” indicates the size of the student (10\% of the teacher), “M” indicates the medium size (50\% of the teacher), and “T” indicates the size of the teacher.
Also, “fixed S” refers to the pre-trained student recommender with the size of "S", “fixed T” refers to the pre-trained teacher recommender with the size of "T".
Note that "fixed S/T" are not updated during the training (i.e., unidirectional knowledge distillation).

\begin{figure*}[t]
    \includegraphics[width=0.5\linewidth]{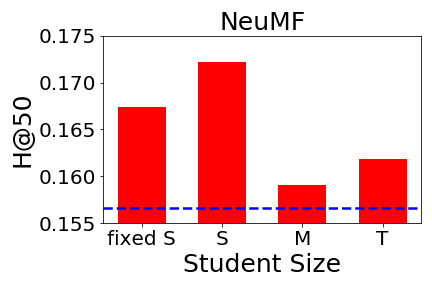}
    \includegraphics[width=0.5\linewidth]{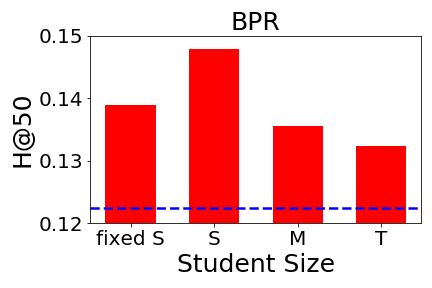}
    \caption{Performance of teacher with student of varying size. The blue dashed line indicates the performance of the model trained separately.}
\label{bdmsize1}
\end{figure*}

\begin{figure*}[t]
    \includegraphics[width=0.5\linewidth]{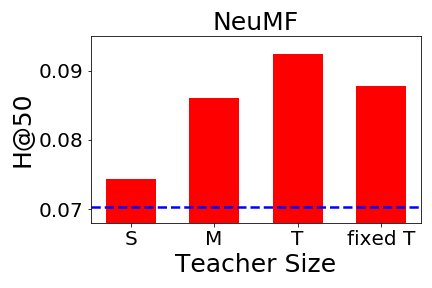}
    \includegraphics[width=0.5\linewidth]{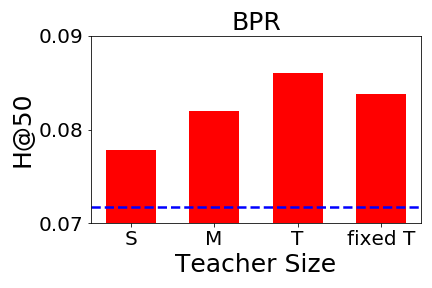}
    \caption{Performance of student with teacher of varying size. The blue dashed line indicates the performance of the model trained separately.}
\label{bdmsize2}
\end{figure*}

First, we observe that both the teacher and the student achieves the greatest performance gain when the capacity gap between two recommenders is largest;
the teacher shows the best performance with the smallest student (i.e., S in Fig. \ref{bdmsize1}) and the student performs best with the largest teacher (i.e., T in Fig. \ref{bdmsize2}).
As mentioned in Section 3.2, the teacher and the student have complementarity as some user-item relationships can be better captured without expensive computations.
In this regard, such complementarity can be maximized when the capacity gap between the two recommenders is large enough.
This result supports our claim that the performance gain comes from the different but reciprocal knowledge of two recommenders.
Also, it is worth noting that there are still performance improvements when two recommenders have identical sizes (i.e., T in Fig. \ref{bdmsize1}, S in Fig. \ref{bdmsize2}).
This can be understood as a kind of self-distillation effect \cite{born18, self19} when the teacher has the same size as the student.
Although they have very similar kinds of knowledge, they can still regularize each other, preventing its counterpart from being overfitted to a few observed interactions.

Second, the teacher is more improved when it is trained along with the learning student than when it is trained with the fixed student (i.e., S vs. fixed S in Fig. \ref{bdmsize1}).
Similarly, the student is more improved when it is trained together with the learning teacher than when it is trained with the fixed teacher (i.e., T vs. fixed T in Fig. \ref{bdmsize2}).
In the unidirectional distillation, the pre-trained recommender (i.e., fixed S/T) is no longer improved, thus always conveys the same knowledge during the training.
On the contrary, within BD, both recommenders are trained together, thus the knowledge distilled between the recommenders gets gradually evolved as they are improved during the training.
As a result, each recommender can be further improved based on the evolved knowledge of its counterparts.
This result again shows the superiority of our bidirectional distillation over the unidirectional distillation of the existing methods.

\subsection{Synchronization Analysis}
We perform in-depth analysis to investigate how well the teacher and the student learn each other's complementary knowledge within the proposed framework.
Specifically, we evaluate how much synchronized the two recommenders are within BD, which shows that they are improved based on each other's complementarity.
To quantify the degree of the synchronization, we define (normalized) Average Rank Difference as follows:
\begin{equation}
\text{Average Rank Diff.} = \frac{1}{n \cdot m} \sum_{\mathcal{D}_{test}} \abs{\text{Rank Diff.}^{u}(i)},
\end{equation}
where $n$ and $m$ are the numbers of users and items, respectively, and $\mathcal{D}_{test}$ is the test set that contains the held-out observed interaction for each user.
The rank difference ($\text{Rank Diff.}^{u}(i)$) is defined in Equation 1.

\begin{table}[t]
\centering
  \caption{Average Rank Difference before and after the training with BD.}
  \begin{tabular}{ccccc}
    \toprule
    \multicolumn{2}{c}{Base Model} & CiteULike & Foursquare & Yelp\\
    \midrule
    \multirow{2}{*}{NeuMF} & before BD & 0.1944 & 0.1020 & 0.0918\\
    & after BD & 0.1323 & 0.0957 & 0.0824\\
    \midrule
    \multirow{2}{*}{CDAE} & before BD & 0.1190 & 0.0963 & 0.0544\\
    & after BD & 0.0871 & 0.0883 & 0.0472\\
    \midrule
    \multirow{2}{*}{BPR} & before BD & 0.1380 & 0.1180 & 0.0560\\
    & after BD & 0.1064 & 0.1082 & 0.0511\\
    \bottomrule
  \end{tabular}
  \label{bdard}
\end{table}

\begin{figure}[t]
    \includegraphics[width=0.5\linewidth]{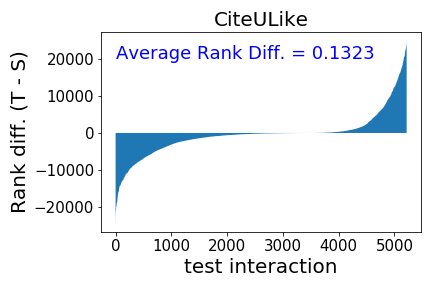}
    \includegraphics[width=0.5\linewidth]{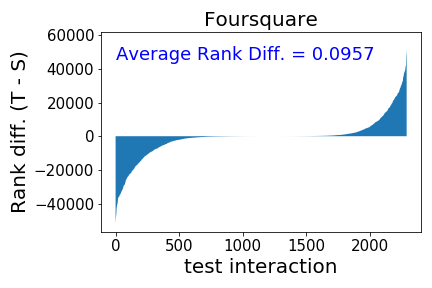}
\caption{Rank difference after the training with BD.}
\label{bdba}
\end{figure}

Table \ref{bdard} shows the change in the average rank difference after the training with BD.
We observe that the average rank difference gets consistently decreased on all the datasets.
This result indicates that the teacher and the student get synchronized by transferring their knowledge to each other.
Figure 4 shows the rank difference between the teacher and the student after the training with BD.
We adopt NeuMF as the base model as done in Section 3.2.
Note that the (normalized) average rank difference is proportional to the extent of the blue area.
We observe that the blue area in Figure \ref{bdba} shrinks compared to that in Figure \ref{TvsS1} after the training with BD.
Interestingly, we observe that the teacher is significantly improved on CiteULike dataset (by up to 20.79\%) which has the largest average rank difference change before and after the training with BD.
The large change indicates that the two recommenders get synchronized well during the training, which leads to significant improvements based on each other's knowledge.

\subsection{Sampling Scheme Analysis}
We examine the effects of diverse sampling schemes on the performance of BD to verify the superiority of the proposed sampling scheme.
Note that the schemes decide what knowledge to be distilled within BD.
We compare five different sampling schemes as follows: 
1) Rank discrepancy-aware sampling, 
2) Rank-aware sampling \cite{cd19},
3) Top-$N$ selection \cite{rd18},
4) Uniform sampling, 
5) \textit{Swapped} rank discrepancy-aware sampling.

The rank discrepancy-aware sampling, which is our proposed scheme, focuses on the rank-discrepant items between the teacher and the student.
On the other hand, the rank-aware sampling (adopted in CD) and top-$N$ selection (adopted in RD) focus on the items ranked highly by one recommender.
The uniform sampling randomly chooses items from the entire recommendation list.
Finally, the swapped rank discrepancy-aware sampling, which is the ablation of the proposed scheme, swaps the sampling probability function of each distillation direction;
we make the teacher follow the student on most of the rank-discrepant items (with \textit{tanh}), and make the student imitate only a few predictions of the teacher (with \textit{exp}).

\begin{table}[t]
\centering
  \renewcommand{\arraystretch}{1}
  \caption{Recommendation performance (H@50) of BD with different sampling schemes on CiteULike. Numbers in bold face are the best results.}
  \begin{tabular}{cc cc}
    \toprule
    Base Model & Sampling Scheme & Teacher & Student\\
    \midrule
    \multirow{5}{*}{NeuMF} & Rank discrepancy-aware & \textbf{0.1722} & \textbf{0.0924} \\
    & Rank-aware \cite{cd19} & 0.1617 & 0.0812 \\
    & Top-$N$ selection \cite{rd18} & 0.1480 & 0.0766 \\
    & Uniform & 0.1590 & 0.0726 \\
    & Swapped rank discrepancy-aware & 0.1512 & 0.0747 \\
    \midrule
    \multirow{5}{*}{CDAE} & Rank discrepancy-aware & \textbf{0.1983} & \textbf{0.0943} \\
    & Rank-aware \cite{cd19} & 0.1818 & 0.0891 \\
    & Top-$N$ selection \cite{rd18} & 0.1757 & 0.0870 \\
    & Uniform & 0.1788 & 0.0819 \\
    & Swapped rank discrepancy-aware & 0.1733 & 0.0851 \\
    \midrule
    \multirow{5}{*}{BPR} & Rank discrepancy-aware & \textbf{0.1479} & \textbf{0.0853} \\
    & Rank-aware \cite{cd19} & 0.1367 & 0.0805 \\
    & Top-$N$ selection \cite{rd18} & 0.1264 & 0.0772 \\
    & Uniform & 0.1306 & 0.0749 \\
    & Swapped rank discrepancy-aware & 0.1281 & 0.0803 \\
    \bottomrule
  \end{tabular}
  \label{bdabl}
\end{table}

Table \ref{bdabl} shows the performance of BD with different sampling schemes.
We first observe that the proposed scheme achieves the best result both for the teacher and the student.
This result verifies the effectiveness of two components of the proposed scheme: 1) sampling based on the rank discrepancy, 2) probability functions differently designed for each distillation direction.
Specifically, we can see the effectiveness of rank discrepant-based sampling by comparing our sampling scheme with rank-aware sampling.
As observed in Section 3.2, the items merely ranked highly by the teacher may be not informative enough to fully enhance the student.
With the proposed scheme, BD focuses on the rank-discrepant items, thus can further improve the recommenders.
Also, we observe that swapping the probability function of each distillation direction (i.e., the swapped rank discrepancy-aware) significantly degrades the performance.
This result strongly indicates that each probability function is well designed to effectively cope with the large performance gap of two recommenders.

\subsection{Hyperparameter Analysis}
In this section, we provide thorough analyses that examine the effects of important hyperparameters on BD.
For the sake of the space, we report the results of NeuMF on CiteULike dataset.
We observe similar tendencies with other base models and datasets.

\begin{figure}[t]
\includegraphics[width=\linewidth, scale=1.7]{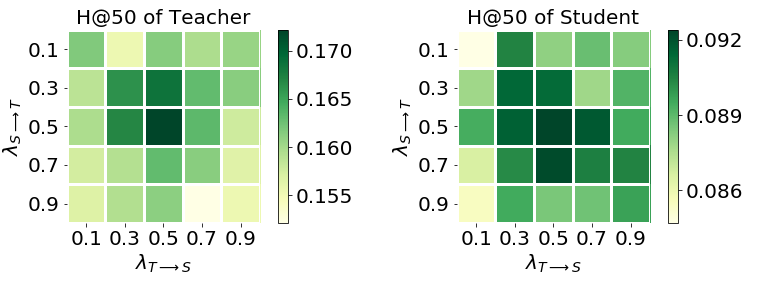}
\caption{Recommendation performance (H@50) of teacher and student with varying lambda.}
\label{bdlambda}
\end{figure}
\begin{figure}[t]
    \includegraphics[width=0.5\linewidth]{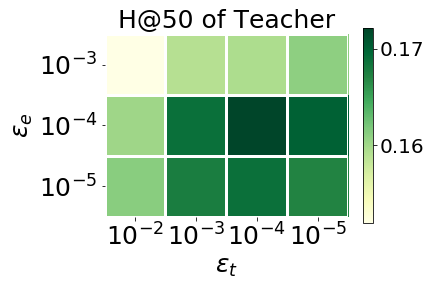}
    \includegraphics[width=0.5\linewidth]{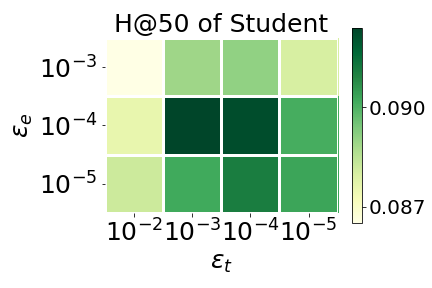}
\caption{Recommendation performance (H@50) of teacher and student with varying smoothing factor.}
\label{bdweight}
\end{figure}
\begin{figure}[t]
    \includegraphics[width=0.5\linewidth]{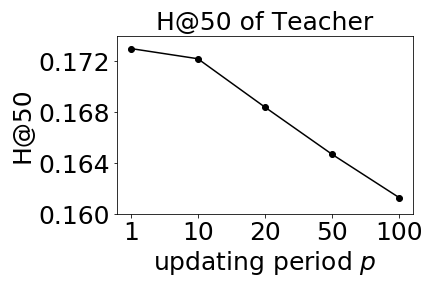}
    \includegraphics[width=0.5\linewidth]{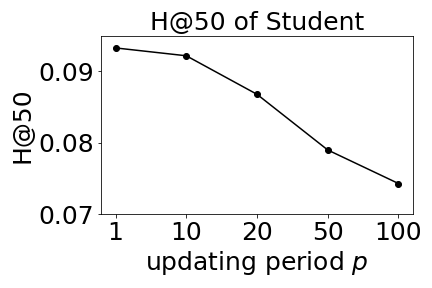}
\caption{Recommendation performance (H@50) of teacher and student with varying updating period.}
\label{bdup}
\end{figure}

First, Figure \ref{bdlambda} shows the performance of the teacher and the student with varying $\lambda_{T \rightarrow S}$ and $\lambda_{S \rightarrow T}$ which control the effects of the distillation losses.
The best performance is achieved when both $\lambda_{T \rightarrow S}$ and $\lambda_{S \rightarrow T}$ are around 0.5.
Also, we observe that both the recommenders are considerably improved when $\lambda_{T \rightarrow S}$ and $\lambda_{S \rightarrow T}$ have similar values (i.e., diagonal entries).
Moreover, we observe that the performance of the student is more robust with respect to $\lambda$ than that of the teacher.
We believe that this is because the student has a lower overall performance than the teacher, thus can easily take advantage of the teacher in broad settings.

Second, Figure \ref{bdweight} shows the performance of the teacher and the student with varying $\epsilon_t$ and $\epsilon_e$ which control the smoothness of the probability functions in rank discrepancy-aware sampling.
When $\epsilon_t$ and $\epsilon_e$ are small, the probability functions become smooth.
The best performance is achieved when both parameters are around $10^{-4}$.
When $\epsilon_e$ is bigger than $10^{-3}$, the probability $p_{S \rightarrow T}(\cdot)$ gets too sharp.
Thus, the teacher cannot fully learn the student's complementary knowledge, which leads to degraded performance.

Lastly, Figure \ref{bdup} shows the performance of the teacher and the student with varying rank updating period $p$ (in Algorithm 1).
Since it is time-consuming to generate the recommendation list every epoch, we update the recommendation list of the two models every $p$ epoch.
In this paper, we use $p=10$ which shows the comparable performance to the upper-bound ($p=1$).

\subsection{Additional Result}
In this section, we report the experimental results on ML100K and AMusic, which are used for CD \cite{cd19}, for the direct comparison with CD.
We do not include this result in our main table, because we consider those datasets are relatively small to simulate the real-world evaluation (ML100K has only 943 users and 1682 items).
We adopt CDAE as the base model and we employ 2-layer MLP for the encoder and the decoder of CDAE.
Experimental results on their experimental settings are as follows:
\begin{table}[h!]
\centering
  \caption{Performance comparison with CD.}
  \begin{tabular}{c cc cc}
    \toprule
    \multirow{2}{*}{Model} & \multicolumn{2}{c}{ML100K} & \multicolumn{2}{c}{AMusic} \\
     & H@50 & N@50 & H@50 & N@50 \\
    \midrule
    Teacher & 0.3966 & 0.1259 & 0.1748 & 0.0533 \\
    BD-Teacher & 0.4317 & 0.1442 & 0.1936 & 0.0607 \\
    BD-Student & 0.4111 & 0.1321 & 0.1723 & 0.0545 \\
    CD & 0.3786 & 0.1232 & 0.1650 & 0.0506 \\
    Student & 0.3503 & 0.1078 & 0.1265 & 0.0416 \\
    \midrule
    \textit{Improv.T} & 8.85\% & 14.54\% & 10.76\% & 13.88\% \\
    \textit{Improv.B} & 8.59\% & 7.22\% & 4.42\% & 7.71\% \\
    \bottomrule
  \end{tabular}
\end{table}

All results are the average of five iterations and statistically significant with p=0.01.
The proposed approach (BD) still outperforms the best competitor (CD) both on ML100K and AMusic.
Moreover, the ranking performance of the teacher also increases with BD.

\section{Related Work}
\textbf{Reducing inference latency of RS.}
Several methods have been proposed for reducing the model size and inference time of recommender systems.
First, a few work adopt discretization techniques to reduce the size of recommenders \cite{discreterec1, discreterec2, discreteAAAI, binary12, discreterec3}.
They learn discrete representations of users and items to make portable recommenders and successfully reduce the model size.
However, their recommendation performance is highly limited due to the restricted capability and thus the loss of recommendation performance is unavoidable \cite{cd19}.
Second, several methods try to accelerate the inference phase by adopting model-dependent techniques \cite{treeRS, inference15, inference17}.
For example, the order-preserving transformations \cite{treeRS} and the pruning techniques \cite{inference15, inference17} have been adopted to reduce the computational costs of the inner product.
Although they can reduce the inference latency, they are applicable only to specific models (e.g., inner product-based models).

To tackle this challenge, a few recent methods \cite{rd18, cd19, DERRD} have adopted knowledge distillation to RS.
KD is a model-agnostic strategy and we can employ any recommender system as the base model. 
RD \cite{rd18} firstly proposes a KD method that makes the student give high scores on the top-ranked items of the teacher’s recommendation list.
Similarly, CD \cite{cd19} makes the student imitate the teacher’s prediction scores with particular emphasis on the items ranked highly by the teacher.
The most recent work, RRD \cite{DERRD}, formulates the distillation process as a relaxed ranking matching problem between the ranking list of the teacher and that of the student.
Since it is daunting for the small student to learn all the prediction results from the large teacher, they focus on the high-ranked items, which can affect the top-$K$ recommendation performance \cite{DERRD}. 
By using such supplementary supervisions from the teacher, they have successfully improved the performance of the student. 
However, they have some critical limitations in that 
1) they rely on the unidirectional distillation, 
2) the high-ranked items are not informative enough.
These limitations are thoroughly analyzed and resolved in this work. 

\textbf{Training multiple models together.}
There have been successful attempts to train multiple models simultaneously for better generalization performance in computer vision and natural language processing \cite{dml18, collaborativelearning18, dual16}.
In computer vision, Deep Mutual Learning (DML) \cite{dml18} trains a cohort of multiple classifiers simultaneously in a peer-teaching scenario where each classifier is trained to follow the predictions of the other classifiers along with its original loss.
Since each classifier starts from a different initial condition, each classifier's predicted probabilities of the next most likely class vary.
DML claims that those secondary quantities provide extra information to the other classifiers that can help them to converge to a more robust minima.
Also, Collaborative Learning \cite{collaborativelearning18} trains several classifier heads of the same network simultaneously by sharing intermediate-level representations.
Collaborative Learning argues that the consensus of multiple views from different heads on the same data provides both supplementary information and regularization to each classifier head.
In natural language processing, Dual Learning \cite{dual16} proposes a learning mechanism that two machine translators teach each other to reduce the costs of human labeling.
Pointing out that the machine translation can be considered as a dual-task forming a closed loop (e.g., English-to-French/ French-to-English), 
Dual Learning generates informative feedback by transferring their outputs to each other and achieves considerable performance improvements without the involvement of a human labeler.
Although the aforementioned methods have effectively increased performance by training multiple models simultaneously, they mostly employ models with the same sizes focusing only on achieving better generalization based on the consensus of various views.

In this paper, we propose a novel Bidirectional Distillation framework whereby two recommenders of different sizes are trained together by transferring their knowledge to each other.
Pointing out that the large recommender and the small recommender have different but complementary knowledge, we design an effective distillation mechanism that considers their huge capacity gap.

\section{Conclusion \& Future Work}
We propose a novel Bidirectional Distillation framework for top-$K$ recommender systems whereby the teacher and the student are collaboratively improved with each other during the training.
Within our proposed framework, both the teacher and the student transfer their knowledge to each other with the distillation loss.
Also, our framework considers the capacity gap between the teacher and the student with differently tailored \textit{rank discrepancy-aware sampling}.
Our extensive experiments on real-world datasets show that BD significantly outperforms the state-of-the-art competitors in terms of the student recommender.
Furthermore, the teacher model also gets the benefit of the student and performs better than when being trained separately.
We also provide analyses for the in-depth understanding of BD and verifying the effectiveness of each proposed component.

Moreover, we can adopt our confidence calibration approach for knowledge distillation.
Since the prediction confidence of the student and the teacher serves as an additional soft target label, the confidence calibration technique can be readily utilized in bidirectional distillation.
As the calibrated confidence can provide more reliable uncertainty on the prediction of the student and the teacher, we expect the confidence calibration may enhance the performance of the bidirectional distillation.

%% file: table/5exzmaintable.tex
\begin{sidewaystable}[ph!] 
 \caption{Performance comparison. \textit{Improv.T} denotes the improvement of the teacher over the original teacher model ("Teacher" in this table), \textit{Improv.B} denotes the improvement of the student over the best competitive method and \textit{Improv.S} denotes the improvement of the student over the original student model ("Student" in this table). $*$ and $**$ indicate $p \leq $ 0.05 and $p \leq $ 0.01 for the paired t-test of BD vs. the best competitor.}
  \resizebox{\textwidth}{!}{%
  \begin{tabular}{cl cccc | cccc | cccc}
    \toprule 
    \multirow{2}{*}{Base Model} & \multirow{2}{*}{Method} & \multicolumn{4}{c}{CiteULike} & \multicolumn{4}{c}{Foursquare} & \multicolumn{4}{c}{Yelp}\\
     & & H@50 & H@100 & N@50 & N@100 & H@50 & H@100 & N@50 & N@100 & H@50 & H@100 & N@50 & N@100 \\
    \bottomrule
     & Teacher & 0.1566 & 0.2272 & 0.0385 & 0.0487 & 0.1612 & 0.2066 & 0.0615 & 0.0685 & 0.0782 & 0.1341 & 0.0200 & 0.0290 \\
     & BD-Teacher & 0.1722 & 0.2443 & 0.0431 & 0.0548 & 0.1704 & 0.2232 & 0.0653 & 0.0737 & 0.0836 & 0.1473 & 0.0217 & 0.0319 \\
     & BD-Student & 0.0924 & 0.1425 & 0.0242 & 0.0325 & 0.1657 & 0.2181 & 0.0637 & 0.0721 & 0.0779 & 0.1315 & 0.0203 & 0.0289  \\
     & CD & 0.0820 & 0.1318 & 0.0227 & 0.0307 & 0.1548 & 0.2063 & 0.0592 & 0.0677 & 0.0712 & 0.1225 & 0.0189 & 0.0263 \\
    NeuMF & RD & 0.0755 & 0.1242 & 0.0197 & 0.0268 & 0.1442 & 0.1819 & 0.0547 & 0.0611 & 0.0662 & 0.1134 & 0.0161 & 0.0237 \\
     & Student & 0.0703 & 0.1167 & 0.0173 & 0.0248 & 0.1264 & 0.1674 & 0.0512 & 0.0578 & 0.0615 & 0.1018 & 0.0153 & 0.0221 \\
    \cmidrule{2-14}
     & \textit{Improv.T} & 9.96\% & 7.53\% & 11.95\% & 12.53\% & 5.73\% & 8.03\% & 6.21\% & 7.59\% & 6.91\% & 9.84\% & 8.50\% & 10.00\% \\
     & \textit{Improv.B} & 12.68\%{**} & 8.16\%{**} & 6.84\%{**} & 5.87\%{**} & 7.05\%{**} & 5.74\%{**} & 7.53\%{**} & 6.50\%{**} & 9.41\%{**} & 7.35\%{**} & 7.41\%{**} & 9.89\%{**} \\
     & \textit{Improv.S} & 31.44\% & 22.11\% & 39.88\% & 30.85\% & 31.11\% & 30.31\% & 24.34\% & 24.74\% & 26.67\% & 29.17\% & 32.68\% & 30.77\% \\
    \midrule
     & Teacher & 0.1710 & 0.2445 & 0.0492 & 0.0611 & 0.1653 & 0.2281 & 0.0650 & 0.0743 & 0.0894 & 0.1523 & 0.0229 & 0.0331 \\
     & BD-Teacher & 0.1983 & 0.2702 & 0.0588 & 0.0674 & 0.1740 & 0.2368 & 0.0680 & 0.0766 & 0.0993 & 0.1681 & 0.0255 & 0.0366 \\
     & BD-Student & 0.0943 & 0.1470 & 0.0269 & 0.0355 & 0.1721 & 0.2242 & 0.0614 & 0.0709 & 0.0745 & 0.1323 & 0.0191 & 0.0268 \\
     & CD & 0.0861 & 0.1332 & 0.0250 & 0.0324 & 0.1629 & 0.2104 & 0.0553 & 0.0651 & 0.0667 & 0.1183 & 0.0176 & 0.0236 \\
    CDAE & RD & 0.0848 & 0.1266 & 0.0241 & 0.0318 & 0.1670 & 0.2146 & 0.0572 & 0.0675 & 0.0657 & 0.1118 & 0.0157 & 0.0231 \\
     & Student & 0.0724 & 0.1090 & 0.0198 & 0.0257 & 0.1491 & 0.2041 & 0.0552 & 0.0641 & 0.0613 & 0.1061 & 0.0151 & 0.0218 \\
    \cmidrule{2-14}
     & \textit{Improv.T} & 15.96\% & 10.51\% & 13.41\% & 10.31\% & 5.26\% & 3.81\% & 4.62\% & 3.10\% & 11.07\% & 10.37\% & 11.35\% & 10.57\% \\
     & \textit{Improv.B} & 9.52\%{**} & 10.36\%{**} & 7.60\%{**} & 9.57\%{**} & 3.05\%{*} & 4.47\%{*} & 7.34\%{**} & 5.04\%{**} & 11.86\%{**} & 11.83\%{**} & 8.52\%{**} & 13.56\%{**} \\
     & \textit{Improv.S} & 30.25\% & 34.86\% & 35.86\% & 38.13\% & 15.43\% & 9.85\% & 11.23\% & 10.61\% & 21.53\% & 24.69\% & 26.49\% & 22.94\% \\
    \midrule
     & Teacher & 0.1224 & 0.1892 & 0.0328 & 0.0415 & 0.1746 & 0.2307 & 0.0662 & 0.0764 & 0.0896 & 0.1383 & 0.0231 & 0.0291 \\
     & BD-Teacher & 0.1479 & 0.2144 & 0.0390 & 0.0498 & 0.1834 & 0.2449 & 0.0713 & 0.0812 & 0.0946 & 0.1488 & 0.0245 & 0.0302 \\
     & BD-Student & 0.0853 & 0.1378 & 0.0213 & 0.0298 & 0.1611 & 0.1991 & 0.0610 & 0.0687 & 0.0733 & 0.1283 & 0.0171 & 0.0243 \\
     & CD & 0.0803 & 0.1271 & 0.0201 & 0.0277 & 0.1515 & 0.1851 & 0.0567 & 0.0622 & 0.0660 & 0.1137 & 0.0144 & 0.0221 \\
    BPR & RD & 0.0782 & 0.1234 & 0.0185 & 0.0259 & 0.1404 & 0.1696 & 0.0519 & 0.0574 & 0.0661 & 0.1126 & 0.0155 & 0.0215 \\
     & Student & 0.0717 & 0.1217 & 0.0171 & 0.0251 & 0.1252 & 0.1618 & 0.0469 & 0.0528 & 0.0559 & 0.1039 & 0.0127 & 0.0199 \\
    \cmidrule{2-14}
     & \textit{Improv.T} & 20.79\% & 13.32\% & 18.90\% & 19.88\% & 5.04\% & 6.13\% & 7.63\% & 6.22\% & 5.58\% & 7.59\% & 6.06\% & 3.78\% \\
     & \textit{Improv.B} & 6.16\%{**} & 8.38\%{**} & 5.97\%{**} & 7.40\%{**} & 6.34\%{**} & 7.56\%{**} & 7.58\%{**} & 10.45\%{**} & 10.89\%{**} & 13.94\%{**} & 10.32\%{**} & 13.02\%{**} \\
     & \textit{Improv.S} & 18.90\% & 13.19\% & 24.56\% & 18.53\% & 28.67\% & 23.05\% & 30.06\% & 30.11\% & 31.13\% & 23.48\% & 34.65\% & 22.11\% \\
    \bottomrule
  \end{tabular}}
  \label{bdmain}
\end{sidewaystable} 

%% file: table/5exzmodelparams.tex
\begin{table}[t]
\centering
\setlength{\tabcolsep}{4pt}
  \renewcommand{\arraystretch}{0.8}
  \caption{Model size and online inference efficiency. Time denotes the wall time used for generating recommendation list for every user.}
  \begin{tabular}{c cc cc cc}
    \toprule
    \multirow{2}{*}{Base Model} & \multicolumn{2}{c}{CiteULike} & \multicolumn{2}{c}{Foursquare} & \multicolumn{2}{c}{Yelp}\\
     & \#Params. & Time & \#Params. & Time & \#Params. & Time \\
    \midrule
    NeuMF-T & 3.05M & 17.39s & 6.42M & 14.77s & 5.15M & 85.12s \\
    NeuMF-S & 0.30M & 14.12s & 0.64M & 7.99s & 0.51M & 69.53s \\
    \midrule
    CDAE-T & 2.80M & 14.73s & 6.36M & 7.95s & 3.89M & 64.45s \\
    CDAE-S & 0.30M & 10.35s & 0.69M & 5.78s & 0.41M & 53.39s \\
    \midrule
    BPR-T & 1.52M & 9.11s & 3.21M & 7.71s & 2.57M & 46.83s \\
    BPR-S & 0.15M & 7.94s & 0.32M & 5.39s & 0.26M & 40.49s \\
    \bottomrule
  \end{tabular}
  \label{bdinfer}
\end{table}

%% file: 4.PerK.tex
The conventional top-\textit{K} recommendation, which presents the top-\textit{K} items with the highest ranking scores, is a common practice for generating personalized ranking lists.
However, is this fixed-size top-$K$ recommendation the optimal approach for every user’s satisfaction?
Not necessarily.
We point out that providing fixed-size recommendations without taking into account user utility can be suboptimal, as it may unavoidably include irrelevant items or limit the exposure to relevant ones.
To address this issue, we introduce Top-Personalized-$K$ Recommendation, a new recommendation task aimed at generating a personalized-sized ranking list to maximize individual user satisfaction.
As a solution to the proposed task, we develop a model-agnostic framework named PerK.
PerK estimates the expected user utility by leveraging calibrated interaction probabilities, subsequently selecting the recommendation size that maximizes this expected utility.
Through extensive experiments on real-world datasets, we demonstrate the superiority of PerK in Top-Personalized-$K$ recommendation task.
We expect that Top-Personalized-$K$ recommendation has the potential to offer enhanced solutions for various real-world recommendation scenarios, based on its great compatibility with existing models.

\section{Introduction}
\begin{figure}[t]
    \includegraphics[width=\linewidth]{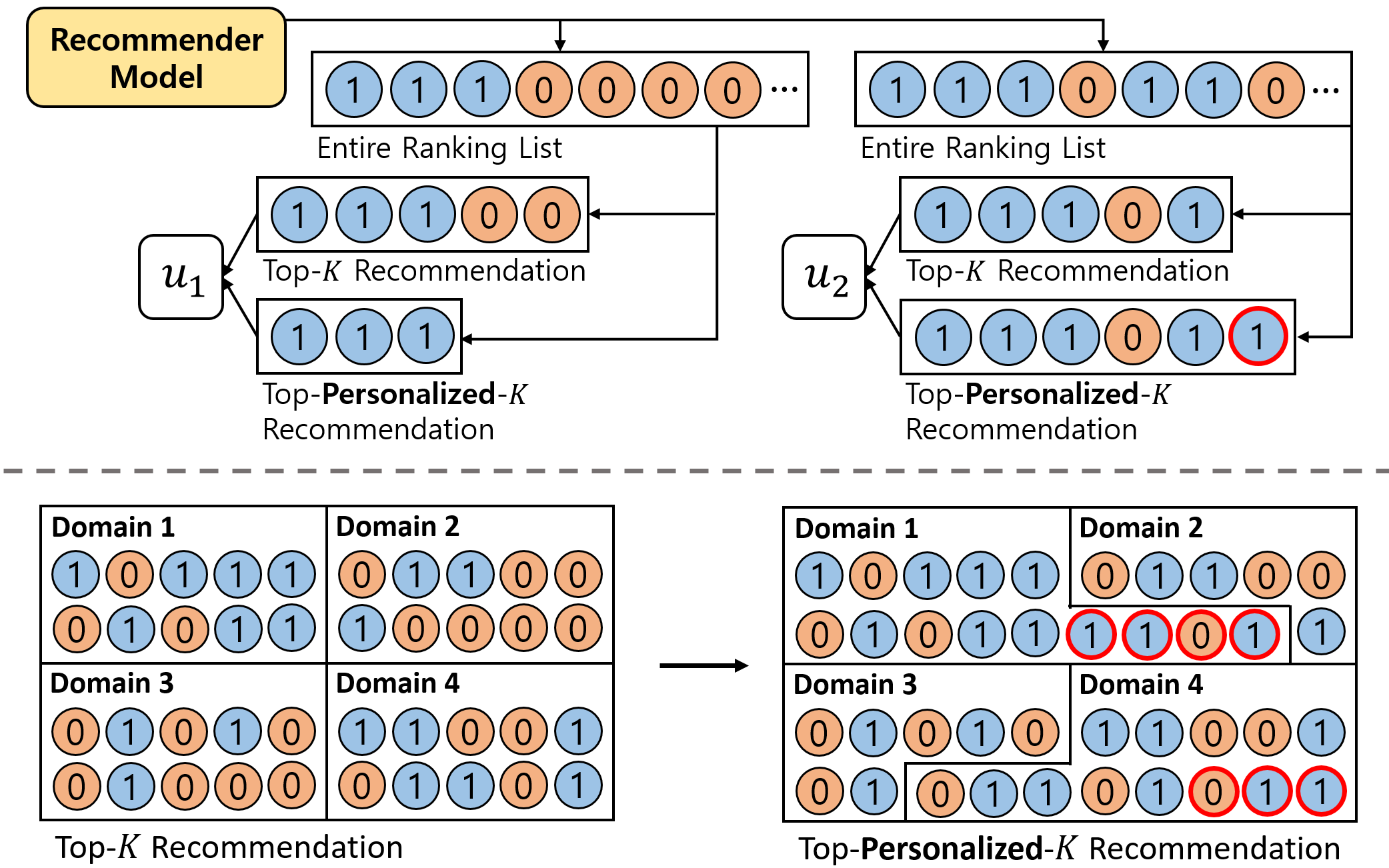}
\caption{An example of top-personalized-$K$ recommendation. 1 and 0 represent the relevant and the irrelevant items, respectively. Note that these labels are not available at the inference phase. Newly added items have red outlines.}
\label{perkintro}
\end{figure}

Personalized recommendations have a significant impact on various daily activities such as shopping, advertising, watching videos, and listening to music.
To generate personalized ranking lists of items, recommender systems utilize the top-$K$ recommendation approach \cite{topk10}, which presents the $K$ items with the highest ranking scores, sorted in descending order.
This approach has become a common practice in recent recommender systems \cite{sgl21, topk22, topkk22} due to its optimality with the globally fixed recommendation size \cite{prp77}.
However, while tremendous efforts have been made on recommender models, an important question has been overlooked in the previous literature: \textit{is the fixed-size top-$K$ recommendation the optimal approach for ensuring every user’s satisfaction?}

To elucidate the potential drawbacks of the top-$K$ recommendation approach, we start with a motivating example with $K=5$ in Figure \ref{perkintro}.
For user 1, the top-$K$ recommendation \textit{unavoidably includes} two irrelevant items in the tail, as the ranking list could not be filled with enough relevant items.
Moreover, for user 2, the recommendation could be further improved by including one additional relevant item in the tail, although user 2 receives a more accurate top-$K$ recommendation with four relevant items.
As such, globally fixing the recommendation sizes results in (1) exposing users to irrelevant items that could lead to ad blindness \cite{swing08} or user attrition \cite{leave06}, and (2) limiting chances to provide relevant items, which can curtail user engagement and revenue \cite{revenue19}.
In this context, we argue the top-$K$ recommendation with a globally fixed recommendation size is not the optimal approach for both user satisfaction and system efficacy.

Instead of globally fixed recommendation sizes, a recommendation size that is \textit{personalized} for individual user satisfaction can create enhanced solutions for various recommendation scenarios.
Back in Figure \ref{perkintro}, by adopting personalized recommendation sizes, the system can increase both users' satisfaction by reducing the effort spent inspecting irrelevant items (user 1) and providing more relevant items (user 2).
Furthermore, the personalized recommendation size paves a way to further increase user satisfaction in various applications, especially for systems with limited resources for making recommendations:
(1) \textit{Multi-domain recommender systems} \cite{cdr12} that display items from various domains on a single constrained screen can strike a balance in recommendation sizes for the maximum overall user satisfaction, by employing adapted-sized ranking lists from each domain (Figure 1 bottom).
(2) In the case of \textit{sponsored advertisements} \cite{swing08}, advertisers can achieve higher user engagement with the same promotion expenses by adjusting the number of promoted items based on each user's expected utility.
(3) In the context of the \textit{prefetching mechanism} \cite{prefetch16}, which caches the initial few seconds of videos expected to be clicked in order to reduce startup delay, the system can minimize cache size and prevent cache pollution by adjusting the number of items to be cached.
In this sense, we claim that embracing personalized recommendation sizes not only enhances user satisfaction but also unlocks diverse optimization possibilities in real-world systems.

In this paper, we propose \textbf{Top-Personalized-$K$ Recommendation}, a new recommendation task resolving the limitation of the top-$K$ recommendation.
Formally, the top-personalized-$K$ recommendation refers to providing a ranking list of the variable size that maximizes individual user satisfaction, which can be quantitatively measured by user utility \cite{fair18, saito22fair}.
As a solution to the proposed task, we develop \textbf{PerK}, a framework to determine the personalized recommendation size with any existing recommender model.
PerK first formulates a bi-level optimization problem where the objective is to determine the recommendation size that maximizes each user's utility (Sec 4.4).
To solve this optimization, we introduce the concept of \textit{expected user utility} (Sec 4.5.2), as it is not feasible to compute the true user utility during the inference phase.
We treat the interaction labels of unobserved items as Bernoulli random variables and derive the expectation of the user utility for various widely-used utility measures.
Derived expected user utilities can be computed with the interaction probability of user-item pairs.

The remaining challenge is obtaining accurate interaction probability with an arbitrary recommender model.
The recommender models do not necessarily output the accurate interaction probability \cite{cal17, kweon22}.
They often output unbounded ranking scores that cannot be treated as probabilities \cite{bpr09, lightgcn20} or miscalibrated interaction probabilities that do not accurately reflect the true likelihood of user-item interactions \cite{ncf17,vae18}.
To address this problem, we propose the use of \textit{calibrated interaction probability} obtained through the user-wise calibration function (Sec 4.5.3).
The calibration function maps the ranking scores of the recommender model to well-calibrated interaction probabilities \cite{kweon22}.
We adopt Platt scaling \cite{platt99} and instantiate it for each user to consider the different distributions of the ranking score across users. 
We train the calibration function to predict the interactions between pairs in a held-out calibration set. 
As a result, the output of the function accurately indicates the true likelihood of interaction, leading to accurate expected user utility.

In summary, we aim to find the optimal recommendation size for each user by (1) obtaining calibrated interaction probability with user-wise calibration, (2) estimating the expected user utility, and (3) determining the recommendation size that results in the maximum expected user utility.
The main contributions of our work can be described as follows:
\begin{itemize}[leftmargin=*]
    \item We highlight the necessity of personalized recommendation size based on its practical advantages in real-world scenarios, which has not been studied well in the previous literature.
    \item We propose Top-Personalized-$K$ Recommendation, a new recommendation task where the recommendation size can be adjusted for each user to enhance individual user satisfaction.    
    \item We develop PerK, a framework to determine the personalized recommendation size by estimating the user's expected utility with the calibrated interaction probability.
    \item We conduct comprehensive experiments with three base recommenders on four real-world datasets, demonstrating the superiority of PerK in the top-personalized-$K$ recommendation task.
\end{itemize}

\section{Related Work}
To the best of our knowledge, personalized recommendation size has not been studied well in the previous literature.
The nearest research line is the document list truncation \cite{arampatzis2009stop, bahri2020choppy} that aims to determine the optimal cutoff position for retrieved documents.
Document list truncation has been applied in various domains, including legal document retrieval \cite{legal07, legal22} and searching \cite{attncut21, mmoecut22}.
The recent methods \cite{bicut19, bahri2020choppy, attncut21, mmoecut22} formulate the truncation problem as a classification task that predicts the optimal position among the candidate cut-off positions.
The target label for this classification is given as a $K$-dimensional vector and each element indicates the probability of being the optimal position.
Then, they deploy large and deep models (e.g., Bi-LSTM \cite{lstm97}, Transformer \cite{transformer17}) for the classification.

\textbf{Limitations.} While this task is related to ours, directly applying truncation methods to the top-personalized-$K$ recommendation leads to poor performance for several reasons.
First, document retrieval datasets \cite{qin2010letor, msmarco16} have sufficient relevant documents, which provide rich target labels for classification.
In contrast, recommendation datasets suffer from severe sparsity and only a small fraction of items is relevant for each user.
Additionally, recommendation datasets often include hidden-relevant items among unobserved ones, providing inaccurate and noisy target labels.
Second, the document retrieval task leverages high-dimensional features extracted from text contents (e.g., tf-idf \cite{tfidf03} or doc2vec \cite{doc2vec12}) to train large transformer models.
However, in our scenarios, these additional features (e.g., review texts) are not always available for unobserved items \cite{unavailable18}, and scalar ranking scores are insufficient for training large transformer models.
As a result, state-of-the-art truncation models (e.g., AttnCut \cite{attncut21} and MtCut \cite{mmoecut22}) show limited performance when applied to the top-personalized-$K$ task, and a solution tailored to recommender systems is required.

\section{Preliminaries}
\subsection{Recommendation with Implicit Feedback}
\textbf{Implicit Feedback.} In this paper, we focus on the recommendation with implicit feedback \cite{implicit08}, a widely adopted scenario for top-$K$ recommendation \cite{topk10}.
Let $\mathcal{U}$ and $\mathcal{I}$ denote a set of users and a set of items, respectively.
For a pair of $u \in \mathcal{U}$ and $i \in \mathcal{I}$, an interaction label $Y_{u,i}$ is given as 1 if their interaction is observed and 0 otherwise.
It is noted that when $Y_{u,i}=0$, it indicates either that the item is irrelevant to the user or that it can be a \textit{hidden-relevant} item of the user \cite{saito20}.
A dataset $\mathcal{D} = \{(u,i) | Y_{u,i}=1\}$ consists of positive pairs, and is split into a training set $\mathcal{D}_{\textup{tr}}$ and a validation set $\mathcal{D}_{\textup{val}}$.
$\mathcal{I}^{-}_u = \{i|(u,i) \notin \mathcal{D}_{\textup{tr}} \}$ denotes the unobserved itemset of $u$.

\textbf{Recommender Model.} A recommender model $f_{\theta}: \mathcal{U} \times \mathcal{I} \rightarrow \mathbb{R}$ learns to assign a ranking score to each user-item pair. 
In the literature, a variety of model architectures for $f_{\theta}$ have been deployed, including matrix factorization \cite{mf09}, neural networks \cite{ncf17,vae18}, and graph neural networks \cite{lightgcn20, sgl21}.
To train recommender models, point-wise loss (e.g., binary cross-entropy, mean squared error), pair-wise loss (e.g., BPR \cite{bpr09}, Margin Ranking loss \cite{margin07}), and list-wise loss (e.g., InfoNCE \cite{infonce18}, Sampled Softmax \cite{nips22soft}) have been adopted.
After the recommender model is trained, we produce a ranking list $\pi$ with unobserved items in $\mathcal{I}_u^-$ by sorting the ranking scores $f_{\theta}(u, i)$ in descending order.

\subsection{User Utility}
In this paper, we adopt user utility for the quantitative measurement of user satisfaction yielded by recommendation \cite{fair18, saito22fair}.
User utility evaluates the ranking list $\pi$ based on how much the recommended items are \textit{exposed} and \textit{relevant} to the user:
\begin{equation}
    U(\pi | u) = \sum_{i \in \pi} \omega(\textup{rank}(i|\pi)) \rho(Y_{u,i}).
    \label{perkutil}
\end{equation}
The function $\omega(\cdot)$ maps an item's rank to the item's exposure based on the position-bias model \cite{positionbias08}, and $\rho(\cdot)$ casts the relevance of an item to the user utility.\footnote{In the evaluation phase, we have the true relevance (i.e., irrelevant or hidden-relevant) for unobserved items.}
$\textup{rank}(i|\pi)$ is the rank of item $i$ in the ranking list $\pi$.
Depending on the formulation of $\omega(\cdot)$ and $\rho(\cdot)$, we can represent various utility measures.
For example, Discounted Cumulative Gain (DCG) \cite{NDCG02} can be represented with $\omega(r)=\frac{1}{\textup{log}_2(1+r)}$ and $\rho(Y)=\mathbbm{1}_{\{Y=1\}}(Y)$.
Additionally, other utility measures, such as Normalized Discounted Cumulative Gain \cite{NDCG02}, Penalized Discounted Cumulative Gain \cite{NDCG02, attncut21}, F1 score \cite{f111}, and Truncated Precision \cite{f111, vae18} have been adopted in the literature and will be investigated in this work.

\subsection{Top-K Recommendation}
\begin{definition}[Top-$K$ Recommendation]
The top-$K$ recommendation refers to providing a ranking list of $K$ items with the highest ranking scores.
\end{definition}
\noindent Typically, the recommendation size ($K$) is globally fixed and predefined by systems, taking into account platform constraints such as screen size, thumbnail dimensions, and promotion expense.
The probability ranking principle \cite{prp77} guarantees that this approach maximizes user utility under a fixed value of $K$ and for any decreasing function of $\omega(\cdot)$ in Eq.\ref{perkutil}.
That is, we get:
\begin{equation}
    \pi_{K} = \argsort_{i \in \mathcal{I}_u^{-}} f_{\theta}(u,i) [:K] = \argmax_{\pi \in \Pi_K} U(\pi|u), 
    \label{topk}
\end{equation}
where $[:K]$ denotes to take the first $K$ elements of the list and $\Pi_K$ is a set of all possible ranking lists, each with a size of $K$.
In the rest of this paper, $\pi_K$ represents the sorted top-$K$ ranking list obtained by Eq.\ref{topk} (e.g., $\pi_{m}$ denotes the top-$m$ ranking list).

\textbf{Limitations.}
Despite the prevalence and advancements in the top-$K$ recommendation, as discussed in Section 4.1, the top-$K$ recommendation has limitations in that it provides a fixed-size recommendation without consideration of user utility.
In this case, users must inspect irrelevant items to filter them out, which can be time-consuming, especially in domains with lengthy inspection times (e.g., movies).
This process can negatively impact user satisfaction \cite{sigir22}, resulting in users ignoring future recommendations \cite{swing08} or even leaving the system \cite{leave06}.
Moreover, the fixed size scheme may further degrade the efficiency of real-world applications, such as presenting an equal number of items from each domain without taking into account user preferences or promoting an equal number of items to each user without considering the users' expected utility.
Nevertheless, the methodology for determining the appropriate number of items to present remains unexplored.

\section{Proposed Task}
We here firstly propose a new recommendation task, named \textbf{Top-Personalized-$K$ Recommendation}, as a means to overcome the limitation of the top-$K$ recommendation.
\begin{definition}[Top-Personalized-$K$ Recommendation]
The top-personalized-$K$ recommendation refers to providing a ranking list where the recommendation size is \textit{optimized} in $[K]$ for each user, to maximize individual user utility.\footnote{$[K] = \{1,2,3,...,K\}$}
\end{definition}
\noindent This approach ensures that each user receives a tailored-sized recommendation, helps avoid presenting irrelevant items and providing more number of relevant items.
The problem of finding the optimal recommendation size $k_{\textup{max}}$ can be formulated as the following bi-level optimization problem:
\begin{equation}
\begin{aligned}
    & k_{\textup{max}} = \argmax_{k \in [K]} U(\pi_k | u) \\
    & \textup{s.t.} \quad \pi_k  = \argmax_{\pi \in \Pi_k} U(\pi|u) \\
    & \quad \quad \quad \ = \argsort_{i \in \mathcal{I}_u^{-}} f_{\theta}(u,i) [:k],  \ \ \ \forall k \in [K].
\end{aligned}
\label{def2}
\end{equation}
Fortunately, the inner optimization can be done readily since the top-$k$ ranking list $\pi_{k}$ is the optimal ranking list for each $k$.
Then, in the outer optimization, we would like to select the $k$ where $\pi_{k}$ yields the highest user utility.
However, it is noted that \textit{directly computing the user utility is infeasible}, since we do not have access to the true relevance of unobserved items in the inference phase.

\textbf{Applications.} While our work primarily focuses on the technical aspects, the proposed task has several implications for real-world applications.
Various scenarios can adopt personalized recommendation sizes by modifying the constraint in Eq.\ref{def2}.
For instance, in \textit{multi-domain recommender systems} \cite{cdr12, cdr17}, the total number of recommended items from various domains is constrained due to the single limited screen.
Instead of displaying the equivalent number of items from each domain, we can adjust the recommendation size of each domain to maximize the overall utility under the constraint.
Similarly, in \textit{sponsored advertisement} \cite{swing08, sponsored11}, the advertiser has a budget constraint on the promotion expenses.
In this case, instead of promoting the same number of items to all users, the system can present personalized numbers of promotions depending on each user’s utility while still satisfying the global budget constraint.

\textbf{PerK on Multi-Domain Scenario.}
In the multi-domain scenario, PerK generates top-personalized-$K$ recommendations for each domain by slightly modifying Eq.\ref{def2}:
\begin{equation}
\begin{aligned}
    & \{k^x_{\textup{max}}\}_{x \in \mathcal{S}} = \argmax_{\{ k^x \}_{x \in \mathcal{S}}} \sum_{x \in \mathcal{S}} U(\pi_{k^x} | u, x) \\
    & \textup{s.t.} \quad \pi_{k^x}  = \argsort_{i \in \mathcal{I}_{u,x}^{-}} f^x_{\theta}(u,i) [:k],  \ \ \ \forall k^x \in [K], x \in \mathcal{S} \\
    & \quad \quad  \sum_{x \in \mathcal{S}} k^x \leq N
\end{aligned}
\end{equation}
$k^x$ is the recommendation size for domain $x$ and $\mathcal{S}$ is set of all domains.
$U(\pi_{k^x} | u, x)$ is user utility of user $u$ and domain $x$.
$\mathcal{I}_{u,x}^{-}$ is the unobserved itemset of user $u$ and domain $x$, $f^x_{\theta}$ is the recommender model for domain $x$ (it can be any cross/multi-domain recommender model).
$N$ is the total recommendation size and $\mathbb{N}$ is the natural numbers set.
To briefly explain, PerK finds the recommendation size for each domain to maximize the average user utility across all domains under the constraint that the total recommendation size should not exceed $N$.
This optimization problem is a variant of the Knapsack problem \cite{salkin1975knapsack} and can be solved by dynamic programming.
We present a case study of the multi-domain scenario on the four largest domains of Amazon datasets in Section 4.6.6.

\section{Proposed Framework}
\subsection{Overview}
We propose \textbf{PerK}, a novel framework to find the optimal recommendation size for the top-personalized-$K$ recommendation.
PerK is a model-agnostic framework, allowing it to be adapted for any item recommendation scenario with existing recommenders.
To solve the bi-level optimization problem in Eq.\ref{def2}, PerK utilizes:
\begin{itemize}[leftmargin=*]
    \item \textbf{(Sec 4.5.2) Expected User Utility}:
    Expected user utility can be estimated by treating the interaction labels for unobserved items as Bernoulli random variables.
    PerK derives the computational form of the expected user utility for widely-used utility functions, which can be computed with the interaction probabilities.
    \item \textbf{(Sec 4.5.3) Calibrated Interaction Probability}:
    To obtain accurate interaction probabilities, PerK utilizes user-wise calibration functions instantiated and trained for each user.
    The calibration function maps the ranking scores of the recommender model to the calibrated interaction probabilities.   
\end{itemize}
To sum up, given a pre-trained recommender model, (1) PerK trains the user-wise calibration functions and gets the calibrated interaction probabilities for unobserved items (Sec 4.5.3), 
(2) PerK estimates the expected user utility for each candidate size $k$ with the calibrated  probability (Sec 4.5.2),
(3) PerK selects the size with the maximum expected user utility, and provides the recommendation list having the selected size (Sec 4.5.4).

\subsection{Expected User Utility}
We cannot compute the true user utility in Eq.\ref{def2}, since we do not have access to the true relevance of unobserved items in the inference phase.
To overcome this issue, PerK estimates the \textbf{expected user utility} instead of the true value by treating the interaction label $Y$ for unobserved items as Bernoulli random variables.
We defined the expected user utility as follows:
\begin{equation}
    \mathbb{E}_{Y}[ U(\pi | u) ] = \mathbb{E}_{Y} \Big[ \sum_{i \in \pi} \omega(\textup{rank}(i|\pi)) \rho(Y_{u,i}) \Big].
\end{equation}
For simplicity, we transform the above formalization for the top-$k$ ranking list $\pi_k$ as follows:
\begin{equation}
    \mathbb{E}_{Y}[ U(\pi_k | u) ] = \mathbb{E}_{Y} \Big[ \sum_{r=1}^{k} \omega(r) \rho(Y_{u,r}) \Big].
\end{equation}
With slight abuse of terminology, let $Y_{u,r}$ denote the interaction variable for user $u$ and the $r^{\textup{th}}$ item in $\pi_{k}$.
In the rest of this section, we derive the computational form of expected user utility for four widely-adopted utility measures.

\textbf{Normalized Discounted Cumulative Gain (NDCG).}
NDCG \cite{NDCG02}, one of the most established utility measures, is formulated as:
\begin{equation}
    U_{\textup{NDCG}}(\pi_k | u) = \frac{U_{\textup{DCG}}(\pi_k | u)}{U_{\textup{IDCG}}(\pi_k | u)} = \frac{\sum_{r=1}^{k} \frac{\mathbbm{1}_{\{Y=1\}}(Y_{u,r})}{\textup{log}_2(1+r)}}{\sum_{r = 1}^{\textup{min}(S_Y^u ,k)} \frac{1}{\textup{log}_2(1+r)}}.
    \label{perkndcg}
\end{equation}
$S_Y^u = \sum_{i \in \mathcal{I}_u^{-}} Y_{u,i}$ is the sum of all interaction variables for unobserved items of user $u$.
The expected NDCG with respect to the random variable $Y$ is:
\begin{equation}
\begin{aligned}
    \mathbb{E}_{Y}[ U_{\textup{NDCG}}(\pi_k | u) ] & = \mathbb{E}_{Y}\Big[ \frac{U_{\textup{DCG}}(\pi_k | u)}{U_{\textup{IDCG}}(\pi_k | u)} \Big] \\
    & = \sum_{Y_{u,1}, ..., Y_{u,|\mathcal{I}_u^{-}|}} P(Y_{u,1}, ..., Y_{u,|\mathcal{I}_u^{-}|}) \frac{U_{\textup{DCG}}(\pi_k | u)|_{Y_{u,*}}}{U_{\textup{IDCG}}(\pi_k | u)|_{Y_{u,*}}}.
\end{aligned}
\end{equation}
$U_{\textup{DCG}}(\pi_k | u)$ and $U_{\textup{IDCG}}(\pi_k | u)$ are conditionally independent given $Y_{u,*}$.
Since investigating all possible combinations of $Y_{u,r}$ is intractable, we re-formulate the above equation with the summation over possible $S_Y^u$ by adopting the total expectation theorem \cite{totalexpec08}. 
\begin{equation}
\begin{aligned}
    \mathbb{E}_{Y}[ U_{\textup{NDCG}}(\pi_k | u) ] & = \sum_{m = 1}^{|\mathcal{I}_u^{-}|} P(S_Y^u=m) \ \mathbb{E}_{Y|S_Y^u = m} \Big[ \frac{U_{\textup{DCG}}(\pi_k | u)}{U_{\textup{IDCG}}(\pi_k | u)} \Big] \\    
    & = \sum_{m = 1}^{|\mathcal{I}_u^{-}|} P(S_Y^u=m) \frac{\mathbb{E}_{Y|S_Y^u=m}[U_{\textup{DCG}}(\pi_k | u)]}{U_{\textup{IDCG}}(\pi_k | u)|_{S_Y^u=m}}.
\end{aligned}
\label{perkeq18}
\end{equation}
$U_{\textup{DCG}}(\pi_k | u)$ and $U_{\textup{IDCG}}(\pi_k | u)$ are conditionally independent given $S_Y^u$, and $U_{\textup{IDCG}}(\pi_k | u)|_{S_Y^u=m} = \sum_{r=1}^{\textup{min}(m,k)} \frac{1}{\textup{log}(1+r)}$.
The expected DCG with respect to $Y$ conditioned on $S_Y^u=m$ can be computed as follows:
\begin{equation}
\begin{aligned}
     \mathbb{E}_{Y|{S_Y^u=m}} [U_{\textup{DCG}}(\pi_k | u)] & = \mathbb{E}_{Y|{S_Y^u=m}} \Big[\sum_{r=1}^{k} \frac{\mathbbm{1}_{\{Y=1\}}(Y_{u,r})}{\textup{log}_2(1+r)} \Big] \\
     & = \sum_{r=1}^{k} \frac{P(Y_{u,r}=1 | S_Y^u=m)}{\textup{log}_2(1+r)} \\ 
     & = \sum_{r=1}^{k} \frac{P(Y_{u,r}=1)P(S_{Y_{-r}}^u=m-1)}{P(S_Y^u=m) \textup{log}_2(1+r)},
\end{aligned}
\label{perkeq19}
\end{equation}
where $S_{Y_{-r}}^u = S_{Y}^u - Y_{u,r}$.
After aggregating Eq.\ref{perkeq18} and Eq.\ref{perkeq19}, we get
\begin{equation}
    \mathbb{E}_{Y}[ U_{\textup{NDCG}}(\pi_k | u) ] = \sum_{m = 1}^{|\mathcal{I}_u^{-}|} \frac{\sum_{r=1}^{k} \frac{P(Y_{u,r}=1)P(S_{Y_{-r}}^u=m-1)}{\textup{log}_2(1+r)}}{\sum_{r=1}^{\textup{min}(m,k)} \frac{1}{\textup{log}_2(1+r)}}.
\end{equation}
For scalability, we adopt \textbf{two simple approximations} and get
\begin{equation}
    \mathbb{E}_{Y}[ U_{\textup{NDCG}}(\pi_k | u) ] \approx \sum_{m = 1}^{M} \frac{\sum_{r=1}^{k} \frac{P(Y_{u,r}=1)P(S_{Y}^u=m-1)}{\textup{log}_2(1+r)}}{\sum_{r=1}^{\textup{min}(m,k)} \frac{1}{\textup{log}_2(1+r)}}.
    \label{perkapp}
\end{equation}
(1) We aggregate the summation only to $M \leq 2000$ rather than $|\mathcal{I}_u^{-}|$ (here, $M$ is a hyperparameter), since the users are likely to interact with only a few items among the unobserved items.
(2) We replace $P(S_{Y_{-r}}^u=m-1)$ with $P(S_{Y}^u=m-1)$ as we confirmed that the effect of one interaction for $S_{Y}^u$ is negligible.
These simple techniques make the expected user utility can be estimated in real-time.
$P(S_{Y}^u=m-1) = P( \sum_{i \in \mathcal{I}_u^{-}} Y_{u,i} =m-1)$ follows the Poisson-Binomial distribution \cite{poisonbinomial60}, and can be computed only with the interaction probabilities $P(Y_{u,i}=1)$ for $i \in \mathcal{I}_u^{-}$.

\textbf{Penalized Discounted Cumulative Gain (PDCG).}
PDCG \cite{NDCG02}, a utility measure based on DCG, has a penalizing term for the irrelevant items in the ranking list.\footnote{It is called DCG in the document retrieval field \cite{attncut21,mmoecut22}, however, we call it PDCG to distinguish it from DCG in the item recommendation field in Eq.\ref{perkndcg}.}
\begin{equation}
    U_{\textup{PDCG}}(\pi_k | u) = \sum_{r=1}^{k} \frac{\mathbbm{1}_{\{Y=1\}}(Y_{u,r})- \mathbbm{1}_{\{Y=0\}}(Y_{u,r})}{\textup{log}_2(1+r)}.
\end{equation}
The expected PDCG with respect to interaction variable $Y$ is computed as follows:
\begin{equation}
\begin{aligned}
     \mathbb{E}_{Y}[ U_{\textup{PDCG}}(\pi_k | u) ] & = \mathbb{E}_{Y} \Big[ \sum_{r=1}^{k} \frac{\mathbbm{1}_{\{Y=1\}}(Y_{u,r})-\mathbbm{1}_{\{Y=0\}}(Y_{u,r})}{\textup{log}_2(1+r)} \Big] \\
    & = \sum_{r=1}^{k} \frac{2 \cdot P(Y_{u,r}=1) - 1}{\textup{log}_2(1+r)}.   
\end{aligned}
\label{perkpdcg}
\end{equation}

\textbf{F1 Score (F1).}
F1 \cite{f111} is a utility measure computed as the harmonic mean of Precision and Recall.
\begin{equation}
\begin{aligned}
    U_{\textup{F1}}(\pi_k | u) & = \frac{2 \cdot \frac{\sum_{r=1}^{k} \mathbbm{1}_{\{Y=1\}}(Y_{u,r})}{k} \cdot \frac{\sum_{r=1}^{k} \mathbbm{1}_{\{Y=1\}}(Y_{u,r})}{S_Y^u}}{\frac{\sum_{r=1}^{k} \mathbbm{1}_{\{Y=1\}}(Y_{u,r})}{k} + \frac{\sum_{r=1}^{k} \mathbbm{1}_{\{Y=1\}}(Y_{u,r})}{S_Y^u}} \\
    & = \frac{2 \cdot \sum_{r=1}^{k} \mathbbm{1}_{\{Y=1\}}(Y_{u,r})}{S_Y^u + k}.
\end{aligned}
\end{equation}
The expected F1 with respect to the interaction variable $Y$ is computed as follows:
\begin{equation}
\begin{aligned}
     \mathbb{E}_{Y}[ U_{\textup{F1}}(\pi_k | u) ] &= \sum_{m = 1}^{|\mathcal{I}_u^{-}|} P(S_Y^u=m) \ \mathbb{E}_{Y|S_Y^u = m} [ U_{\textup{F1}}(\pi_k | u)] \\     
    & = \sum_{m=1}^{|\mathcal{I}_u^{-}|} P(S^u_Y=m) \frac{2 \cdot \sum_{r=1}^{k} P(Y_{u,r}=1|S^u_Y=m)}{m+k} \\  
    & = \sum_{m=1}^{|\mathcal{I}_u^{-}|} \frac{2 \cdot \sum_{r=1}^{k} P(Y_{u,r}=1) P(S_{Y_{-r}}^u=m-1)}{m+k} \\  
    & \approx \sum_{m=1}^{M} \frac{2 \cdot \sum_{r=1}^{k} P(Y_{u,r}=1)P(S^u_Y=m-1) }{m+k}.  
\end{aligned}
\label{perkf1}
\end{equation}
Here, we use the total expectation theorem and apply the same approximations as done in Eq.\ref{perkapp}.

\textbf{Truncated Precision (TP).}
TP \cite{f111, vae18} is a utility measure that addresses the limitations of Recall and Precision.\footnote{For example, Precision cannot have a value of 1 if $k$ is larger than $S_Y^u$. TP is referred to as Recall in \cite{vae18}. However, in this work, we use the term Truncated Precision to distinguish it from the standard definition of Recall.}
\begin{equation}
    U_{\textup{TP}}(\pi_k | u) = \frac{\sum_{r=1}^{k} \mathbbm{1}_{\{Y=1\}}(Y_{u,r})}{\textup{min}(k,S_Y^u)}.
\end{equation}
The expected TP with respect to the interaction variable $Y$ is computed as follows:
\begin{equation}
\begin{aligned}
     \mathbb{E}_{Y}[ U_{\textup{TP}}(\pi_k | u) ] &= \sum_{m = 1}^{|\mathcal{I}_u^{-}|} P(S_Y^u=m) \ \mathbb{E}_{Y|S_Y^u = m} [ U_{\textup{TP}}(\pi_k | u)] \\     
    & = \sum_{m=1}^{|\mathcal{I}_u^{-}|} P(S^u_Y=m) \frac{\sum_{r=1}^{k} P(Y_{u,r}=1|S^u_Y=m)}{\textup{min}(k,m)} \\  
    & = \sum_{m=1}^{|\mathcal{I}_u^{-}|} \frac{\sum_{r=1}^{k} P(Y_{u,r}=1) P(S_{Y_{-r}}^u=m-1)}{\textup{min}(k,m)} \\  
    & \approx \sum_{m=1}^{M} \frac{\sum_{r=1}^{k} P(Y_{u,r}=1)P(S^u_Y=m-1) }{\textup{min}(k,m)}.  
\end{aligned}
\label{perktp}
\end{equation}
Here, we use the total expectation theorem and apply the same approximations as done in Eq.\ref{perkapp}.

\textbf{Interaction Probability.}
Up to this point, we have formulated the computational forms of expected user utility for four widely-adopted utility measures: NDCG (Eq.\ref{perkapp}), PDCG (Eq.\ref{perkpdcg}), F1 (Eq.\ref{perkf1}), and TP (Eq.\ref{perktp}).
In common, these measures all require the interaction probability $P(Y_{u,r}=1)$ for the estimation.
In the following subsection, we present our solution to obtain accurate interaction probabilities with an arbitrary recommender model.

\subsection{Calibrated Interaction Probability}
Recommender models do not necessarily output accurate interaction probability.
They often output the ranking score that can have any value of an unbounded real number \cite{bpr09, lightgcn20, sgl21}, making it difficult to treat it as a probability.
Furthermore, even when a model is trained to output probabilities \cite{ncf17, vae18}, it has been demonstrated that these probabilities may not accurately reflect the true likelihood (i.e., model miscalibration) \cite{cal17, kweon22}.

To address this, we introduce a post-processing calibration function $g_\phi: \mathbb{R} \rightarrow [0,1]$ that maps the ranking score $s_{u,i}=f_{\theta}(u,i)$ of a pre-trained recommender to the calibrated interaction probability $g_\phi(s_{u,i}) =\hat{P}(Y_{u,i}=1|s_{u,i})$.
A probability $p$ is regarded \textit{calibrated} if it indicates the ground-truth correctness likelihood \cite{beta17}:
$\mathbb{E} [ Y | g_{\phi}(s) = p ] = p$.
For example, if we have 100 user-item pairs with $p = 0.2$, we expect 20 of them to have interactions $(Y = 1)$.

We adopt Platt scaling $g_{\phi}(s) = \sigma(as + b)$ \cite{platt99}, a generalized form of temperature scaling \cite{cal17}.
This calibration function has been deployed effectively for model calibration in computer vision \cite{tsrevisit21, temperaturescaling21}, natural language processing \cite{tstransforment20}, and recommender system \cite{kweon22}.
The key difference is that PerK instantiates the calibration function for each user, while previous calibration work \cite{kweon22} deploys one global calibration function covering all users.
That is, we have:
\begin{equation}
    g_{\phi_u}(s) = \sigma(a_u s + b_u), \quad \forall u \in \mathcal{U}.
\end{equation}
The user-specific parameters $\phi_u=\{ a_u, b_u \}$ are related to the distribution of the ranking score of each user \cite{beta17,kweon22}.
Therefore, the user-wise calibration function can consider the different distributions of the ranking score across users.

We train the calibration function to predict the interactions of pairs in a calibration set constructed for each user.
It is a common practice to adopt the validation set as the calibration set \cite{cal17, kweon22}.
In this work, the calibration set for user $u$ is constructed as follows:
\begin{equation}
    \mathcal{D}_u^{\textup{cal}} = \{(u,i,Y_{u,i}=1) | i \in \mathcal{I}_u^{\textup{val}}\} \cup \{(u,i,Y_{u,i}=0) | i \in \mathcal{I}_u^{-} \setminus \mathcal{I}_u^{\textup{val}}\},
\end{equation}
where $\mathcal{I}_u^{\textup{val}}\ = \{i|(u,i) \in \mathcal{D}_{\textup{val}} \}$.
We also use the binary cross-entropy loss, a widely-adopted loss function for the calibration of binary classifiers \cite{cal17, kweon22, beta17}, as our calibration loss.
\begin{equation}
    \mathcal{L}_u^{\textup{cal}} = \sum_{(u,i,Y_{u,i}) \in \mathcal{D}^{\textup{cal}}_u} - Y_{u,i}\text{log}( g_{\phi_u}(s_{u,i})) - (1-Y_{u,i}) \text{log}(1 - g_{\phi_u}(s_{u,i})).
\end{equation}
During the fitting of the calibration function, the base recommender model $f_{\theta}$ is fixed and only $g_{\phi_u}$ is updated.
It is worth mentioning that Platt scaling with binary cross-entropy is mathematically equivalent to logistic regression and can be efficiently solved. 
Additionally, as we only have two learnable parameters for each user, our calibration functions have a negligible impact on space complexity.

\subsection{Optimization Procedure}
\begin{algorithm}[t]
\DontPrintSemicolon
\SetKwInOut{Input}{Input}
\SetKwInOut{Output}{Output}
\Input{Training set $\mathcal{D}_{\textup{tr}}$, validation set $\mathcal{D}_{\textup{val}}$, pre-trained recommender model $f_\theta$, maximum recommendation size $K$, expected user utility $\mathbb{E}_{Y}[ U(\pi_{k} | u) ]$}
\Output{Top-personalized-$K$ recommendation $\pi^{\textup{perK}}$}
\BlankLine
Fit user-wise calibration $g_{\phi_u}$ on $f_\theta$ \Comment{Section 4.5.3}
\BlankLine
$\pi_{K} = \argsort_{i \in \mathcal{I}_u^{-}} f_{\theta}(u,i) [:K] $ \Comment{Top-$K$ Recommendation} \\
\For{$k \in [K]$}{
$\pi_{k} = \pi_{K}[:k]$ \\
Compute $\mathbb{E}_{Y}[ U(\pi_{k} | u) ]$ with $g_{\phi_u}$ \Comment{Section 4.5.2}
}
\BlankLine
$k_{\textup{max}} = \argmax_{k \in [K]} \mathbb{E}_{Y} [U(\pi_{k} | u)]$. \\
$\pi^{\textup{perK}} = \pi_{k_{\textup{max}}}$ \\
\caption{Top-Personalized-K Recommendation with PerK}
\label{perkalgo}
\end{algorithm}

Algorithm \ref{perkalgo} presents the entire procedure of PerK solving the bi-level optimization in Eq.\ref{def2} for the top-personalized-$K$ recommendation.
(line 1): PerK first fits the user-wise calibration function $g_{\phi_{u}}$ with $\mathcal{D}_u^{\textup{cal}}$ and $\mathcal{L}_u^{\textup{cal}}$ on top of pre-trained and fixed recommender $f_\theta$.
(line 2): PerK generates top-$K$ recommendation by sorting the ranking score.
(line 3-5): For each candidate size $k$, PerK estimates the expected user utility $\mathbb{E}_{Y}[ U(\pi_{k} | u) ]$ by using the calibrated interaction probability $g_{\phi_{u}}(s_{u,i})=\hat{P}(Y_{u,i}=1|s_{u,i})$.
(line 6-7): Lastly, PerK select $k_{\textup{max}}$, the recommendation size with the highest expected user utility, and provide the top-personalized-$K$ recommendation $\pi^{\textup{perK}} = \pi_{k_{\textup{max}}}$ to user $u$.

\section{Experiments}
\subsection{Experiment Setup}
We provide a summary of the experiment setup due to limited space.
We will publicly provide the GitHub repository of this work in the final version.

\begin{table}[ht]
\centering
  \caption{Data statistics after the preprocessing.}
  \begin{tabular}{ccccc}
    \toprule
    Dataset & \#Users & \#Items & \#Interactions & Sparsity \\
    \midrule
    MovieLens 10M & 69,838 & 8,939 & 9,985,038 & 98.40\% \\
    CiteULike & 3,277 & 16,807 & 178,187 & 99.68\% \\
    MovieLens 25M & 162,414 & 18,424 & 24,811,113 & 99.17\% \\
    Amazon Books & 157,809 & 112,048 & 8,460,428 & 99.95\% \\
    \bottomrule
  \end{tabular}
\end{table}
\textbf{Datasets.}
We use four real-world datasets including MovieLens 10M,\footnote{https://grouplens.org/datasets/movielens/} CiteULike \cite{CUL13}, MovieLens 25M, and Amazon Books \cite{amazon19}.
These datasets are publicly available and have been widely used in the literature \cite{ncf17, lightgcn20, cdl15, cvae17}.
We adopt the 20-core setting for all datasets as done in MovieLens datasets.
We randomly split each user’s interactions into a training set (60\%), a validation set (20\%), and a test set (20\%).

\textbf{Methods Compared.}
We compare PerK with various traditional and recent methods.
Specifically, we adopt
\begin{itemize}
    \item \textbf{Oracle}: It uses the ground-truth labels of the test set to determine the optimal recommendation size for each user.
\end{itemize}
and three traditional methods:
\begin{itemize}
    \item \textbf{Top-$k$}: It denotes the top-$k$ recommendation with globally fixed recommendation size. We adopt $k \in \{ 1, 5, 10, 20, 50\}$.
    \item \textbf{Rand}: It randomly selects the recommendation size for each user. It represents the lower bound of the performance.
    \item \textbf{Val-$k$}: It selects the recommendation size that maximizes validation utility for each user.
\end{itemize}
and two state-of-the-art methods for truncating document retrieval results:
\begin{itemize}
    \item \textbf{AttnCut} \cite{attncut21}: It deploys a classification model with Bi-LSTM and Transformer encoder to predict the best cut-off position.
    \item \textbf{MtCut} (MMoECut) \cite{mmoecut22}: It deploys MMoE \cite{mmoe18} on top of AttnCut architecture and adopt multi-task learning.
\end{itemize}
and the proposed framework:
\begin{itemize}
    \item \textbf{PerK (ours)}: We determine the personalized recommendation size by estimating the expected user utility with calibrated interaction probability.
\end{itemize}
For AttnCut and MtCut, we use the source code of the authors.\footnote{\url{https://github.com/Woody5962/Ranked-List-Truncation}}
We use a $K$-dimentional ranking score vector for the input of the models.
We have tried to use the item embeddings as additional features for the models, but it did not increase the performance.
For each dataset, hyperparameters are tuned by using grid searches on the validation set. 
We use Adam optimizer with learning rate in $\{10^{-2}, 10^{-3}, 10^{-4}\}$ and weight decay in $\{0, 10^{-1}, 10^{-3}, 10^{-5}, 10^{-7}, 10^{-9} \}$.
We set the batch size to 64 and the embedding size to 64.
The number of layers and the number of transformer heads are chosen from $\{1, 2 \}$ and the dropout ratio is set to 0.2.
Each model is trained until the convergence of validation performa
For a quantitative comparison of user utility, we set the maximum recommendation size to 50 for all methods.
We also tried 100 and observed a similar performance improvement for PerK.

\textbf{Base Recommender Models ($f_{\theta}$).}
We adopt three widely-used recommender models with various model architectures and loss functions.
Bayesian Personalized Ranking (BPR) \cite{bpr09} captures the user-item relevance by the inner product of the user and the item embeddings, and is trained with the loss function that makes the model put the higher ranking score on the observed pair than the unobserved pair.
Neural Collaborative Filtering (NCF) \cite{ncf17} adopts the feed-forward neural networks to output the ranking score of a user-item pair and is trained with the binary cross-entropy loss.
LightGCN (LGCN) \cite{lightgcn20} adopts simplified Graph Convolutional Networks (GCN) \cite{GCN} to capture the high-order interaction between the user and the item, and is trained with the loss function of BPR.
Since PerK is a model-agnostic framework, other models can be also adopted for PerK in future work.

For all the base recommender models, we basically follow the source code of the authors and use PyTorch \cite{pytorch19} for the implementation.
For each dataset, hyperparameters are tuned by using grid searches on the validation set. 
We use Adam optimizer \cite{adam14} with learning rate in $\{10^{-2}, 10^{-3}, 10^{-4}\}$ and weight decay in $\{0, 10^{-1}, 10^{-3}, 10^{-5}, 10^{-7}, 10^{-9} \}$.
We set the batch size to 8192 and the embedding size is chosen from $\{ 64, 128 \}$ for all base models.
For NCF and LGCN, the number of layers is chosen from $\{ 1, 2\}$.
The negative sample rate is set to 1 for all models.
Each model is trained until the convergence of validation performance.
The ranking score from the base recommender model serves as input for PerK, AttnCut, and MtCut.

\subsection{Comparison of User Utility}
\input{table/6EXztable2.tex}
Table \ref{perkmain1}, \ref{perkmain2}, and \ref{perkmain3} present the average user utility yielded by the recommendation of each compared method.
\textit{For a fair comparison, we adopt the same base model and the same target utility throughout the training of AttnCut, MtCut, and PerK in each experimental configuration.}
We report the average result of three independent runs. 

\textbf{Benefits of Top-Personalized-$K$ Recommendation.}
We first observe that personalizing the recommendation size results in higher user utility than the globally fixed size.
Oracle shows significantly higher utility compared to the best value of Top-$k$.
This upper bound on user utility highlights the importance of the proposed task in improving user satisfaction.
Accordingly, methods determining the personalized recommendation sizes (i.e., AttnCut, MtCut, and PerK) generally outperform Top-$k$.
Furthermore, we notice that Top-$k$ shows a large performance deviation depending on the recommended size.
Therefore, the system should avoid naively setting a globally fixed recommendation size and instead determine the personalized size for improved user satisfaction.

\textbf{Effectiveness of PerK.}
We observe that PerK outperforms the competitors in the top-personalized-$K$ recommendation task.
Val-$k$ does not perform well since overfitting to the validation utility is not effective due to the sparse and noisy interactions.
Furthermore, the effectiveness of the methods for document list truncation (i.e., AttnCut and MtCut) remains limited, since the target label for the classification does not provide enough supervision due to the hidden-positive and noisy interactions.
Indeed, we discovered that they can be easily overfitted to select the size that works well globally (Figure \ref{perkdist1} and \ref{perkdist2}), resulting in comparable performance to the best of Top-$k$.
In contrast, PerK does not rely on a deep model and considers the hidden-relevant items, to estimate the expected user utility.
PerK directly computes the expected user utility in mathematical form with the calibrated interaction probability and significantly outperforms AttnCut and MtCut in the top-personalized-$K$ task.

\begin{figure}[ht]
\centering 
    \includegraphics[width=0.35\linewidth]{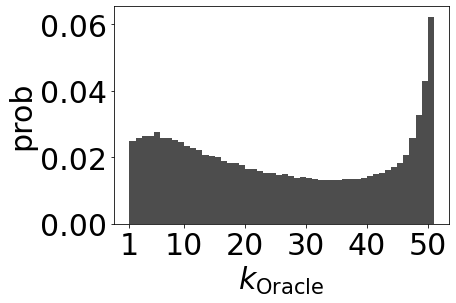}
    \includegraphics[width=0.27\linewidth]{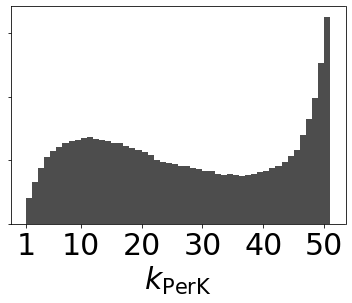}
    \includegraphics[width=0.27\linewidth]{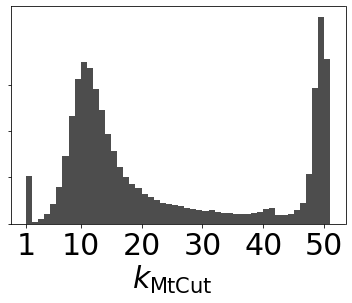}
\caption{Distribution of recommendation sizes from each method on MovieLens 25M. Three figures in the same row share the y-axis.}
\label{perkdist1}
\end{figure}

\begin{figure}[ht]
    \includegraphics[width=0.35\linewidth]{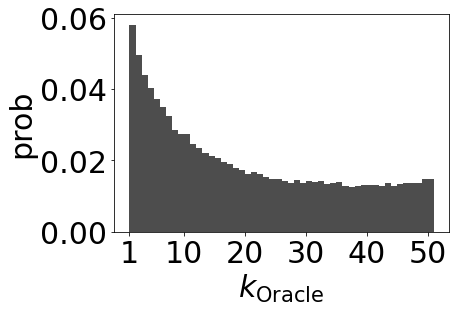}
    \includegraphics[width=0.27\linewidth]{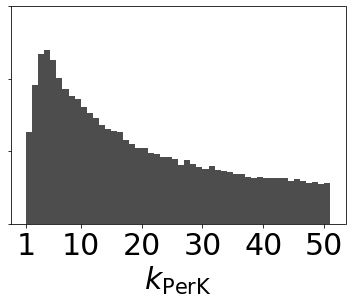}
    \includegraphics[width=0.27\linewidth]{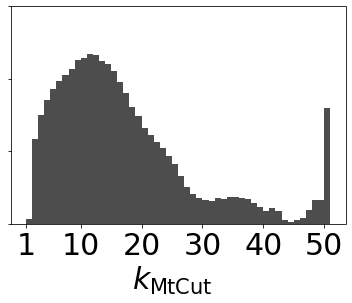}
\caption{Distribution of recommendation sizes from each method on Amazon Books. Three figures in the same row share the y-axis.}
\label{perkdist2}
\end{figure}

\subsection{Personalized Recommendation Size}
Figure \ref{perkdist1} and \ref{perkdist2} show the distributions of recommendation sizes determined by Oracle, PerK, and MtCut.
The base model is BPR and the target user utility is F1.
We have the following findings:
(1) The distributions of Oracle show that the optimal recommendation size for maximum user utility differs for each user.
(2) The distributions of MtCut are severely skewed towards $k \in [10, 20]$ and have a high peak in that range.
The reason is that MtCut is overfitted to select the globally well-performing $k$ which falls into the range of [10, 20] (refer to F1 on MovieLens 25M and Amazon Books with BPR in Table \ref{perkmain1}).
(3) The distributions of PerK are smooth and fairly close to those of Oracle.
It is noted that we can set the constraints for the minimum recommendation size in Eq.\ref{def2}, if the system requires it.

\subsection{User Degree vs Personalized Size}
\begin{figure}[t]
\centering 
\includegraphics[width=0.45\linewidth]{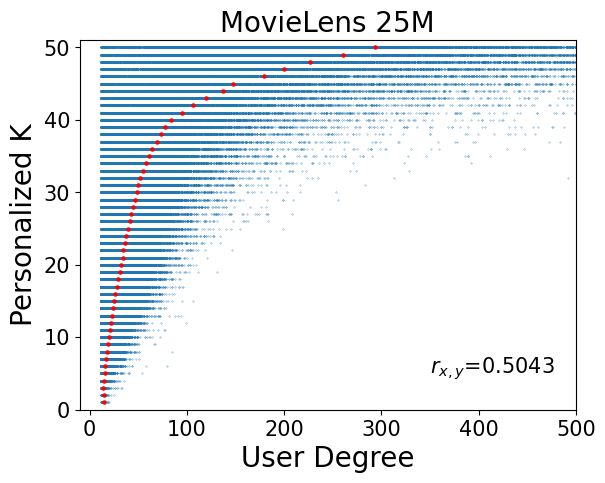}
\includegraphics[width=0.45\linewidth]{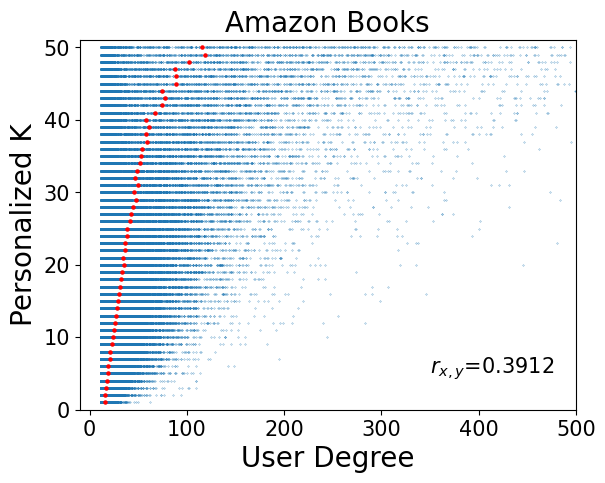}
\caption{User degree versus Personalized Recommendation Size. Red dots denote the average x (user degree) for each y (personalized K).}
\label{udprs}
\end{figure}
Figure \ref{udprs} shows the relation between user degree (number of interactions) and selected personalized recommendation size.
The base model is BPR and the target utility is F1.
$r_{x,y}$ denotes the Pearson correlation coefficient.
Since both correlation coefficients are positive, there is a positive correlation between user degree and selected personalized recommendation size.
If a user has sufficient interactions, the recommender model can capture stable and certain preferences for that user.
Therefore, generally speaking, the mean confidence in top items increases and more number of recommendations is provided.
On the other hand, users with fewer interactions also may get larger recommendations (Upper left of Figure \ref{udprs}).
In this case, if a user has a distinct preference for small interactions, the recommender model also may provide many recommendations for that distinct preference.

\subsection{Ablation Study for Calibration}
\begin{table}[ht]
\caption{Calibration error and user utility for each calibration method. User utilities are computed when each calibration method is applied to PerK. Lower is better for ECE.}
\resizebox{\linewidth}{!}{%
\begin{tabular}{ccccccc}
\toprule
Dataset & Calibration & ECE $\downarrow$  & NDCG   &PDCG & F1  &TP   \\
\toprule
\multirow{3}{*}{ML10M}  & uncalibrated & 0.1284 & 0.4083 & -0.1071 & 0.1966 & 0.4481\\
    & global            &  0.0046 & 0.4255 & 0.0838 & 0.2175 & 0.4925\\
    & user-wise         &  \textbf{0.0011} & \textbf{0.4482}& \textbf{0.0928} &\textbf{0.2293}  &\textbf{0.5335} \\
\midrule
\multirow{3}{*}{CiteULike}  & uncalibrated & 0.0480 & 0.1811 & -0.5667 & 0.0975 & 0.2512\\
                            & global  & 0.0017 & 0.1940 &-0.5306 & 0.1054 & 0.2639 \\
                            & user-wise & \textbf{0.0003} & \textbf{0.2061}&  \textbf{-0.5177}& \textbf{0.1080} & \textbf{0.2764}\\
\midrule
\multirow{3}{*}{ML25M}      & uncalibrated  & 0.1572 & 0.3612 & -0.2311 & 0.1669 & 0.3812\\
                            & global  & 0.0422 & 0.3761 & -0.1015 & 0.1757 & 0.4007 \\
                            & user-wise & \textbf{0.0098} & \textbf{0.4056}&  \textbf{0.0047}& \textbf{0.1960} & \textbf{0.4742}\\
\midrule
\multirow{3}{*}{ABooks}     & uncalibrated & 0.3371  & 0.0812 & -0.8864 & 0.0427 & 0.1391 \\
                            & global & 0.0581 & 0.0946 & -0.8655 & 0.0449 & 0.1477 \\
                            & user-wise & \textbf{0.0171} & \textbf{0.0982}& \textbf{-0.8646} & \textbf{0.0468} & \textbf{0.1553}\\   
\bottomrule
\end{tabular}}
\label{perkabl2}
\end{table}

Table \ref{perkabl2} present the ablation study on the calibration method with NCF \cite{ncf17} as a base model.
We compare (1) the calibration performance, and (2) user utility when it is applied to PerK.
The calibration performance is measured by Expected Calibration Error (ECE) \cite{bbq15}, a widely used metric for measuring the gap between the output probability and true likelihood of interaction \cite{kweon22, cal17}.
We observe that the proposed user-wise calibration function shows lower ECE than the global calibration function.
Accordingly, PerK yields higher user utilities when it adopts user-wise calibration, demonstrating the superiority of the proposed user-wise calibration over the global calibration.

\subsection{Space and Time analysis}
\begin{table}[ht]
\caption{Space and Time analysis. \#Params. denotes the number of learnable parameters and Time denotes the wall time (in ms) used for generating a ranking list for a user.}
\resizebox{\linewidth}{!}{%
\begin{tabular}{c|c|cc|cc|cc|cc}
\toprule
Base & \multirow{2}{*}{Method} & \multicolumn{2}{c|}{\textbf{ML10M}} & \multicolumn{2}{c|}{\textbf{CiteULike}} & \multicolumn{2}{c|}{\textbf{ML25M}} & \multicolumn{2}{c}{\textbf{ABooks}} \\
model &                         & \#Params.           & Time            & \#Params.           & Time             & \#Params.           & Time            & \#Params.           & Time           \\
\toprule
\multirow{2}{*}{BPR}        & Top-k                   & 5.041M             & 2.871           & 1.291M              & 2.086            & 11.574M            & 2.374           & 34.542M            & 2.636          \\
                            & PerK                    & 5.181M             & 3.022           & 1.292M              & 2.241            & 11.898M            & 2.652           & 34.857M            & 2.973          \\
\midrule
\multirow{2}{*}{NCF}      & Top-k                   & 10.092M            & 4.092           & 2.581M              & 2.566            & 23.156M            & 3.067           & 69.117M            & 7.767          \\
                            & PerK                    & 10.232M            & 4.406           & 2.588M              & 2.817            & 23.480M            & 3.378           & 69.432M            & 8.871          \\
\midrule
\multirow{2}{*}{LightGCN}       & Top-k                   & 5.042M             & 3.195           & 2.571M              & 2.493            & 23.147M            & 2.628           & 34.542M            & 3.781          \\
                            & PerK                    & 5.181M             & 3.449           & 2.577M              & 2.759            & 23.472M            & 2.969           & 34,857M            & 4.106       \\  
\bottomrule
\end{tabular}}
\label{perkspace}
\end{table}

Table \ref{perkspace} shows the number of learnable parameters and inference time of Top-$k$ and PerK.
The target user utility is NDCG for PerK.
We use PyTorch \cite{pytorch19} with CUDA on GTX Titan Xp GPU and Intel Xeon(R) E5-2640 v4 CPU.

First, PerK does not significantly increase the number of learnable parameters from Top-$k$.
PerK only has two additional parameters for each user's calibration function, and it has a negligible impact considering the size of the user embedding is typically selected in the range of 64-128.

\begin{figure}[h!]
\centering
    \begin{minipage}[r]{0.49\linewidth}
        \centering
          \includegraphics[width=0.49\textwidth]{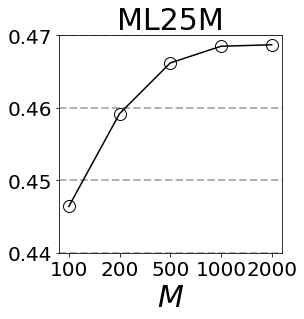}
          \includegraphics[width=0.49\textwidth]{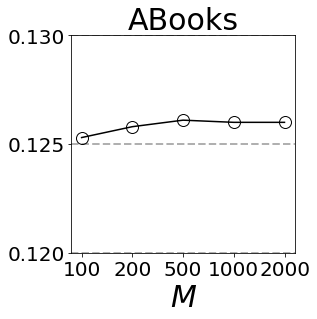}
    \caption*{(a) NDCG}
    \end{minipage}
    \hfill%
    \begin{minipage}[r]{0.49\linewidth}
        \centering
          \includegraphics[width=0.49\textwidth]{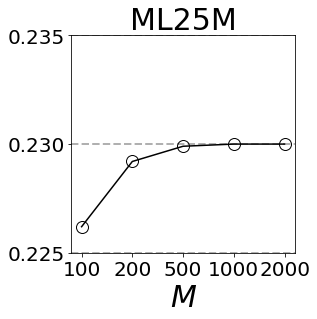}
          \includegraphics[width=0.49\textwidth]{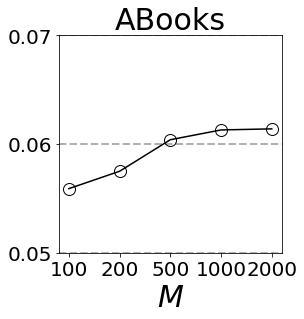}
    \caption*{(b) F1}
    \end{minipage}
    \caption{Hyperparameter study of $M$}
\label{perkhyper}
\end{figure}
Second, the inference time is increased by about 10\%.
To speed up the estimation of expected user utility,
(1) We perform the user-wise calibration for all users together with a few matrix operations and estimate the expected user utility with various $k$ in a parallel way,
and (2) We adopt two approximation techniques for fast computation of the expected user utility in Eq.\ref{perkndcg}.
The hyperparameter study for $M$, which is used for this approximation, with BPR is presented in Figure \ref{perkhyper}.
We can see that it is enough to aggregate the summation in Eq.7d just to $M=2000$ rather than to the number of all unobserved items $|\mathcal{I}_u^{-}|$.

\subsection{Case Study}
\begin{figure}[h!]
\centering 
    \includegraphics[width=1\linewidth]{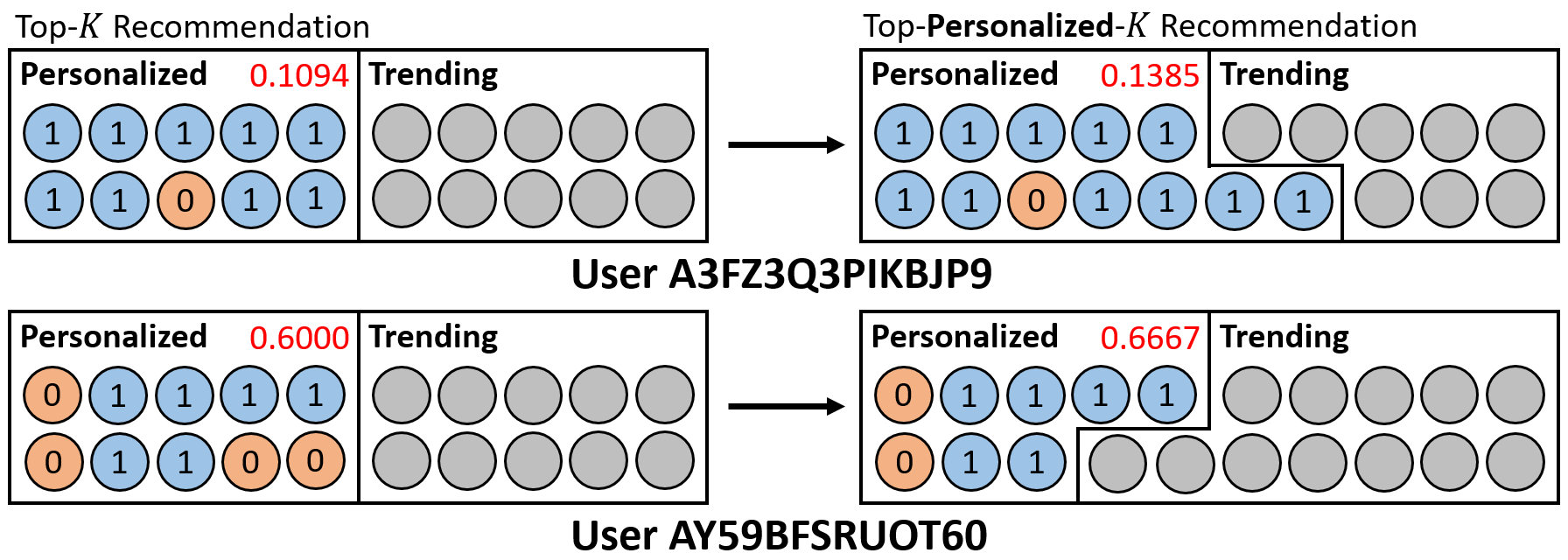}
\caption*{(a) Single-domain scenario on Amazon Books}

    \includegraphics[width=1\linewidth]{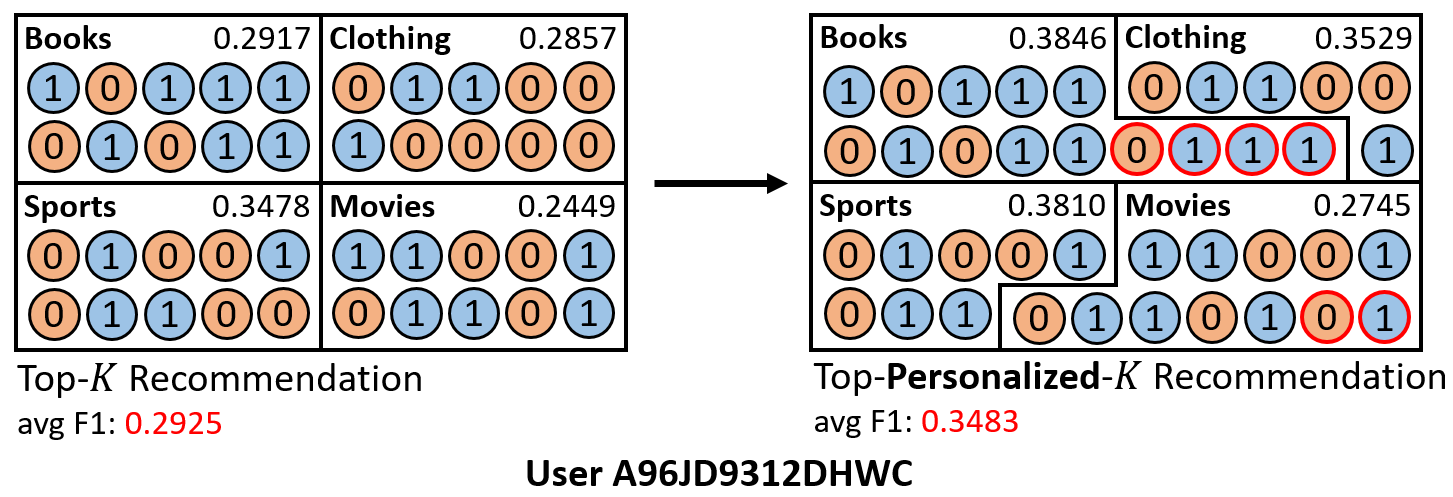}
\caption*{(b) Multi-domain scenario on Amazon}
\caption{Case study. 1 and 0 represent the relevant and the irrelevant labels, which are unavailable at the inference phase. Decimal numbers indicate the F1 score.}
\label{perkcase}
\end{figure}

We present a case study on Amazon dataset to provide concrete examples of how PerK can be applied in real-world applications.

\textbf{Single-domain scenario.}
Figure \ref{perkcase} (a) shows a case study of the single-domain scenario on Amazon Books with BPR as a base model and F1 as a target utility.
Real-world recommender systems often display trending items alongside personalized ones, to create a balanced experience that encompasses both popular choices and individual preferences \cite{jain2015trends}.
We determine the number of personalized items with PerK, and the remaining slots are then populated with trending items.
As a result, F1 for the personalized items increases by a large margin for both users, and more trending items can be presented to the second user.
In this context, the top-personalized-$K$ recommendation allows the system to strike a balance between the exploration-exploitation trade-off by effectively adjusting the number of personalized items.

\textbf{Multi-domain scenario.}
Figure  \ref{perkcase} (b) shows a case study of the multi-domain scenario on the four largest domains of Amazon datasets with BPR as a base model and F1 as a target utility.
Here, we have a constraint on the total recommendation size considering the system's limited resources, such as screen size and thumbnail dimensions.
We generate top-personalized-$K$ recommendations for each domain by slightly modifying Eq.\ref{def2}, to maximize the average F1 across all domains under the constraint that the total recommendation size should not exceed 40.
As a result, we can see that the average F1 score increases as we adopt the top-personalized-$K$ recommendation scheme.
This example demonstrates that the top-personalized-$K$ recommendation can be adopted to multi-domain systems for displaying the optimized number of items from each domain on a single constrained page.

\section{Conclusion}
We first highlight the necessity of personalized recommendation size based on its practical advantages in real-world scenarios, which has not been studied well in the previous literature.
Then, we propose Top-Personalized-$K$ Recommendation, a new recommendation task that aims to find the optimal recommendation size for each user to maximize individual user satisfaction.
As a solution to the top-personalized-$K$ recommendation, we propose PerK, a framework determining the recommendation size that maximizes the expected user utility estimated by using calibrated interaction probabilities.
In our thorough experiments on real-world datasets, PerK outperforms recent competitors in the top-personalized-$K$ recommendation task.
We believe that the top-personalized-$K$ recommendation can provide enhanced solutions for various item recommendation scenarios and anticipate future work on applications including multi-domain recommender systems, sponsored advertisements, and prefetching mechanisms.

%% file: table/6EXztable2.tex
\begin{sidewaystable}[ph!] 
 \caption{Average user utility of recommendations produced by compared methods when they are optimized for each target utility measure. \textit{Improv.} denotes the improvement of PerK over MtCut. The numbers in boldface denote the best result in each setting. *We conduct the paired t-test with a 0.05 level and every \textit{Improv.} is statistically significant.}
  \resizebox{\textwidth}{!}{%
  \begin{tabular}{cl | cccc | cccc | cccc | cccc}
    \toprule 
     Base & \multirow{2}{*}{Method} & \multicolumn{4}{c|}{MovieLens 10M} & \multicolumn{4}{c|}{CiteULike} & \multicolumn{4}{c|}{MovieLens 25M} & \multicolumn{4}{c}{Amazon Books}\\
     Model & & NDCG & PDCG & F1 & TP &NDCG & PDCG & F1 & TP & NDCG & PDCG & F1 & TP & NDCG & PDCG & F1 & TP \\
    \toprule
     & Oracle &0.6190 & 0.6702  & 0.3152 & 0.7078 & 0.3020 & -0.3677  & 0.1806 & 0.3940 & 0.5961 & 0.6089  & 0.2883 & 0.6812 & 0.1895 & -0.7438  & 0.1133 & 0.2643 \\
     \cmidrule{2-18}
     & Top-1 & 0.4549 & -0.0902 & 0.0546 & 0.4549 & 0.1837 & -0.6326  & 0.0338 & 0.1837 & 0.4392 & -0.1216 & 0.0492 & 0.4392 & 0.0961 & -0.8078  & 0.0216 & 0.0961 \\
     & Top-5  &0.3925 & -0.6428 & 0.1541 & 0.3751 & 0.1489 & -2.0787  & 0.0842 & 0.1401 & 0.3729 & -0.7561 & 0.1357 & 0.3547 & 0.0786 & -2.4938  & 0.0507 & 0.0741 \\
     & Top-10  &0.3725 & -1.3715 & 0.1984 & 0.3584 & 0.1495 & -3.3650  & 0.1024 & 0.1483 & 0.3496 & -1.5372 & 0.1757 & 0.3328 & 0.0828 & -3.9401  & 0.0582 & 0.0852 \\
     & Top-20 & 0.3734 & -2.8434 & 0.2228 & 0.3849 & 0.1668 & -5.4853  & 0.1108 & 0.1907 & 0.3460 & -3.0494 & 0.2013 & 0.3510 & 0.0955 & -6.2660  & 0.0584 & 0.1155 \\
     BPR & Top-50 & 0.4057 & -7.1636 & 0.2117 & 0.4841 & 0.2021 & -10.8747 & 0.0930 & 0.2834 & 0.3740 & -7.3813 & 0.1983 & 0.4421 & 0.1203 & -11.8590 & 0.0496 & 0.1826 \\
    \cmidrule{2-18}
     & Rand & 0.3872 & -3.6387 & 0.2046 & 0.4144 & 0.1738 & -6.2622  & 0.1003 & 0.2058 & 0.3611 & -3.8266 & 0.1865 & 0.3801 & 0.0998 & -7.1022  & 0.0535 & 0.1284 \\
     & Val-$k$ & 0.4562 & -0.0640 & 0.2270 & 0.4906 & 0.1922 & -0.6162  & 0.1061 & 0.2266 & 0.4368 & -0.0532 & 0.2052 & 0.4347 & 0.1093 & -0.8378  & 0.0563 & 0.1362 \\
    \cmidrule{2-18}
     & AttnCut & 0.4604 & -0.0702 & 0.2371 & 0.4969 & 0.2024 & -0.6142  & 0.1063 & 0.2835 & 0.4392 & -0.1214 & 0.2087 & 0.4392 & 0.1203 & -0.8078  & 0.0566 & 0.1826 \\
     & MtCut & 0.4621 & -0.0167 & 0.2383 & 0.5037 & 0.2034 & -0.6022  & 0.1071 & 0.2838 & 0.4413 & -0.0616 & 0.2140 & 0.4639 & 0.1212 & -0.8078  & 0.0568 & 0.1826 \\
     \cmidrule{2-18}
     & PerK & \textbf{0.4901} & \textbf{0.2087}  & \textbf{0.2538} & \textbf{0.5711} & \textbf{0.2159} & \textbf{-0.4971}  & \textbf{0.1117} & \textbf{0.2993} & \textbf{0.4687} & \textbf{0.1876}  & \textbf{0.2300} & \textbf{0.5401} & \textbf{0.1261} & \textbf{-0.7952}  & \textbf{0.0619} & \textbf{0.1894} \\
     & \textit{Improv.} & 6.06\%* &   -      & 6.50\%* & 13.38\%* & 6.15\%* &   -       & 4.30\%* & 5.46\%* & 6.21\%* &    -     & 7.48\%* & 16.43\%* & 4.04\%* &   -       & 8.98\%* & 3.72\%* \\
    \bottomrule
  \end{tabular}}
  \label{perkmain1}
\end{sidewaystable} 

\begin{sidewaystable}[ph!] 
 \caption{Average user utility of recommendations produced by compared methods when they are optimized for each target utility measure. \textit{Improv.} denotes the improvement of PerK over MtCut. The numbers in boldface denote the best result in each setting. *We conduct the paired t-test with a 0.05 level and every \textit{Improv.} is statistically significant.}
  \resizebox{\textwidth}{!}{%
  \begin{tabular}{cl | cccc | cccc | cccc | cccc}
    \toprule 
     Base & \multirow{2}{*}{Method} & \multicolumn{4}{c|}{MovieLens 10M} & \multicolumn{4}{c|}{CiteULike} & \multicolumn{4}{c|}{MovieLens 25M} & \multicolumn{4}{c}{Amazon Books}\\
     Model & & NDCG & PDCG & F1 & TP &NDCG & PDCG & F1 & TP & NDCG & PDCG & F1 & TP & NDCG & PDCG & F1 & TP \\
    \toprule
     & Oracle & 0.5760 & 0.4792  & 0.2896 & 0.6688  & 0.2813 & -0.4076  & 0.1685 & 0.3690 & 0.5310 & 0.3544  & 0.2513 & 0.6157  & 0.1444 & -0.8246  & 0.0885 & 0.2076 \\
     \cmidrule{2-18}
     & Top-1  & 0.4081 & -0.1839 & 0.0474 & 0.4081  & 0.1700 & -0.6601  & 0.0304 & 0.1700 & 0.3761 & -0.2479 & 0.0393 & 0.3761  & 0.0666 & -0.8668  & 0.0150 & 0.0666 \\
     & Top-5  & 0.3521 & -0.8801 & 0.1350 & 0.3368  & 0.1403 & -2.1273  & 0.0786 & 0.1330 & 0.3206 & -1.0628 & 0.1117 & 0.3055  & 0.0563 & -2.6231  & 0.0370 & 0.0537 \\
     & Top-10  & 0.3360 & -1.6822 & 0.1764 & 0.3251  & 0.1384 & -3.4438  & 0.0942 & 0.1368 & 0.3011 & -1.9434 & 0.1474 & 0.2873  & 0.0606 & -4.1048  & 0.0434 & 0.0638 \\
     & Top-20  & 0.3397 & -3.2150 & 0.2023 & 0.3545  & 0.1529 & -5.5941  & 0.1015 & 0.1733 & 0.2980 & -3.5640 & 0.1721 & 0.3033  & 0.0712 & -6.4683  & 0.0443 & 0.0888 \\
     NCF & Top-50  & 0.3721 & -7.6021 & 0.1966 & 0.4525  & 0.1866 & -11.0026 & 0.0870 & 0.2618 & 0.3237 & -8.0277 & 0.1744 & 0.3870  & 0.0916 & -12.1148 & 0.0384 & 0.1437 \\
    \cmidrule{2-18}
     & Rand  & 0.3527 & -3.9939 & 0.1866 & 0.3821  & 0.1609 & -6.4157  & 0.0910 & 0.1884 & 0.3102 & -4.3422 & 0.1603 & 0.3298  & 0.0749 & -7.3212  & 0.0407 & 0.0987 \\
     & Val-$k$  & 0.4123 & -0.1030  & 0.2071 & 0.4554  & 0.1794 & -0.5801  & 0.0989 & 0.2101 & 0.3769 & -0.5071 & 0.1729 & 0.3842  & 0.0791 & -0.8935  & 0.0422 & 0.1018 \\
    \cmidrule{2-18}
     & AttnCut  & 0.4154 & -0.0839 & 0.2068 & 0.4612  & 0.1812 & -0.5527  & 0.0991 & 0.2619 & 0.3762 & -0.1166 & 0.1755 & 0.3994  & 0.0918 & -0.8668  & 0.0437 & 0.1438 \\
     & MtCut  & 0.4195 & 0.0805 & 0.2153 & 0.4798  & 0.1891 & -0.5499  & 0.1006 & 0.2627 & 0.3798 & -0.0142 & 0.1837 & 0.4079  & 0.0923 & -0.8667  & 0.0441 & 0.1441 \\
     \cmidrule{2-18}
     & PerK  & \textbf{0.4482} & \textbf{0.0928}  & \textbf{0.2293} & \textbf{0.5335}  & \textbf{0.2061} & \textbf{-0.5177}  & \textbf{0.1080} & \textbf{0.2764} & \textbf{0.4056} & \textbf{0.0047}  & \textbf{0.1960} & \textbf{0.4742}  & \textbf{0.0982} & \textbf{-0.8646}  & \textbf{0.0468} & \textbf{0.1553} \\
     & \textit{Improv.} & 6.84\%* &   -      & 6.50\%* & 11.19\%* & 8.99\%* &   -       & 7.36\%* & 5.22\%* & 6.79\%* &   -      & 6.70\%* & 16.25\%* & 6.39\%* &     -     & 6.12\%* & 7.77\%*\\
    \bottomrule
  \end{tabular}}
  \label{perkmain2}
\end{sidewaystable} 

\begin{sidewaystable}[ph!] 
 \caption{Average user utility of recommendations produced by compared methods when they are optimized for each target utility measure. \textit{Improv.} denotes the improvement of PerK over MtCut. The numbers in boldface denote the best result in each setting. *We conduct the paired t-test with a 0.05 level and every \textit{Improv.} is statistically significant.}
  \resizebox{\textwidth}{!}{%
  \begin{tabular}{cl | cccc | cccc | cccc | cccc}
    \toprule 
     Base & \multirow{2}{*}{Method} & \multicolumn{4}{c|}{MovieLens 10M} & \multicolumn{4}{c|}{CiteULike} & \multicolumn{4}{c|}{MovieLens 25M} & \multicolumn{4}{c}{Amazon Books}\\
     Model & & NDCG & PDCG & F1 & TP &NDCG & PDCG & F1 & TP & NDCG & PDCG & F1 & TP & NDCG & PDCG & F1 & TP \\
    \toprule
     & Oracle & 0.6249 & 0.7459  & 0.3094 & 0.7051  & 0.3471 & -0.2352  & 0.2104 & 0.4400 & 0.5893 & 0.4840  & 0.2864 & 0.6772 & 0.1782 & -0.7571  & 0.1061 & 0.2454 \\
     \cmidrule{2-18}
     & Top-1  & 0.4749 & -0.0502 & 0.0565 & 0.4749  & 0.2182 & -0.5636  & 0.0420 & 0.2182 & 0.4236 & -0.1527 & 0.0490 & 0.4236 & 0.0919 & -0.8162  & 0.0213 & 0.0919 \\
     & Top-5  & 0.4004 & -0.5959 & 0.1531 & 0.3799  & 0.1829 & -1.8801  & 0.1068 & 0.1735 & 0.3603 & -0.8313 & 0.1349 & 0.3427 & 0.0744 & -2.5187  & 0.0487 & 0.0699 \\
     & Top-10  & 0.3763 & -1.3345 & 0.1952 & 0.3584  & 0.1829 & -3.1084  & 0.1270 & 0.1816 & 0.3391 & -1.6383 & 0.1741 & 0.3237 & 0.0783 & -3.9772  & 0.0552 & 0.0803 \\
     & Top-20  & 0.3734 & -2.8186 & 0.2180 & 0.3793  & 0.1999 & -5.1979  & 0.1303 & 0.2244 & 0.3379 & -3.1805 & 0.1988 & 0.3449 & 0.0898 & -6.3212  & 0.0543 & 0.1077 \\
     LightGCN & Top-50  & 0.3997 & -7.1674 & 0.2068 & 0.4673  & 0.2354 & -10.5738 & 0.1037 & 0.3174 & 0.3687 & -7.5533 & 0.1951 & 0.4396 & 0.1117 & -11.9466 & 0.0449 & 0.1663 \\
    \cmidrule{2-18}
     & Rand  & 0.3875 & -3.5835 & 0.2001 & 0.4062  & 0.2075 & -6.0557  & 0.1182 & 0.2450 & 0.3527 & -3.9406 & 0.1839 & 0.3735 & 0.0938 & -7.1826  & 0.0497 & 0.1187 \\
     & Val-$k$  & 0.4652 & 0.0417  & 0.2225 & 0.4894  & 0.2296 & -0.4288  & 0.1286 & 0.2620 & 0.4253 & -0.0864  & 0.1953 & 0.4439 & 0.1017 & -0.8469  & 0.0532 & 0.1247 \\
    \cmidrule{2-18}
     & AttnCut  & 0.4749 & -0.0412 & 0.2237 & 0.4899  & 0.2260 & -0.5616  & 0.1291 & 0.3186 & 0.4242 & -0.1527 & 0.2041 & 0.4359 & 0.1121 & -0.8162  & 0.0516 & 0.1707 \\
     & MtCut  & 0.4761 & 0.1171  & 0.2318 & 0.4949  & 0.2369 & -0.5487  & 0.1319 & 0.3223 & 0.4317 & 0.0649 & 0.2113 & 0.4440 & 0.1185 & -0.8162  & 0.0536 & 0.1761 \\
     \cmidrule{2-18}
     & PerK  & \textbf{0.4993} & \textbf{0.3309}  & \textbf{0.2489} & \textbf{0.5702}  & \textbf{0.2551} & \textbf{-0.3989}  & \textbf{0.1383} & \textbf{0.3438} & \textbf{0.4543} & \textbf{0.0876}  & \textbf{0.2277} & \textbf{0.4742} & \textbf{0.1261} & \textbf{-0.7912}  & \textbf{0.0571} & \textbf{0.1878} \\
     & \textit{Improv.} & 4.87\%* &    -     & 7.38\%* & 15.22\%* & 7.68\%* &   -       & 4.85\%* & 6.67\%* & 5.24\%* &    -     & 7.76\%* & 6.80\%* & 6.41\%* &    -      & 6.53\%* & 6.64\%* \\
    \bottomrule
  \end{tabular}}
  \label{perkmain3}
\end{sidewaystable} 

%% file: 5.Con.tex
In this dissertation, I investigated confidence calibration in recommendations and introduced two real-world applications of confidence on recommendations.
Specifically, (1) I proposed Gaussian calibration and Gamma calibration that map the ranking scores of pre-trained models to well-calibrated preference probabilities, without affecting the recommendation performance (Section 2).
(2) I propose a Bidirectional Distillation framework whereby both the teacher and the student collaboratively improve with each other by treating the confidence of the counterpart as additional learning guidance (Section 3).
(3) I introduce Top-Personalized-$K$ Recommendation, a new recommendation task adjusting the number of presented items based on the expected user utility estimated with calibrated probability (Section 4).
We validate the superiority of each proposed method with extensive experiments on real-world datasets.
We also provide in-depth analyses to verify the effectiveness of each proposed component.
In the future, I plan to continue my work on model calibration for recommender systems, particularly for fairness and large language models of recommendations. 

\begin{figure}[h!]
\centering
    \includegraphics[scale=0.5]{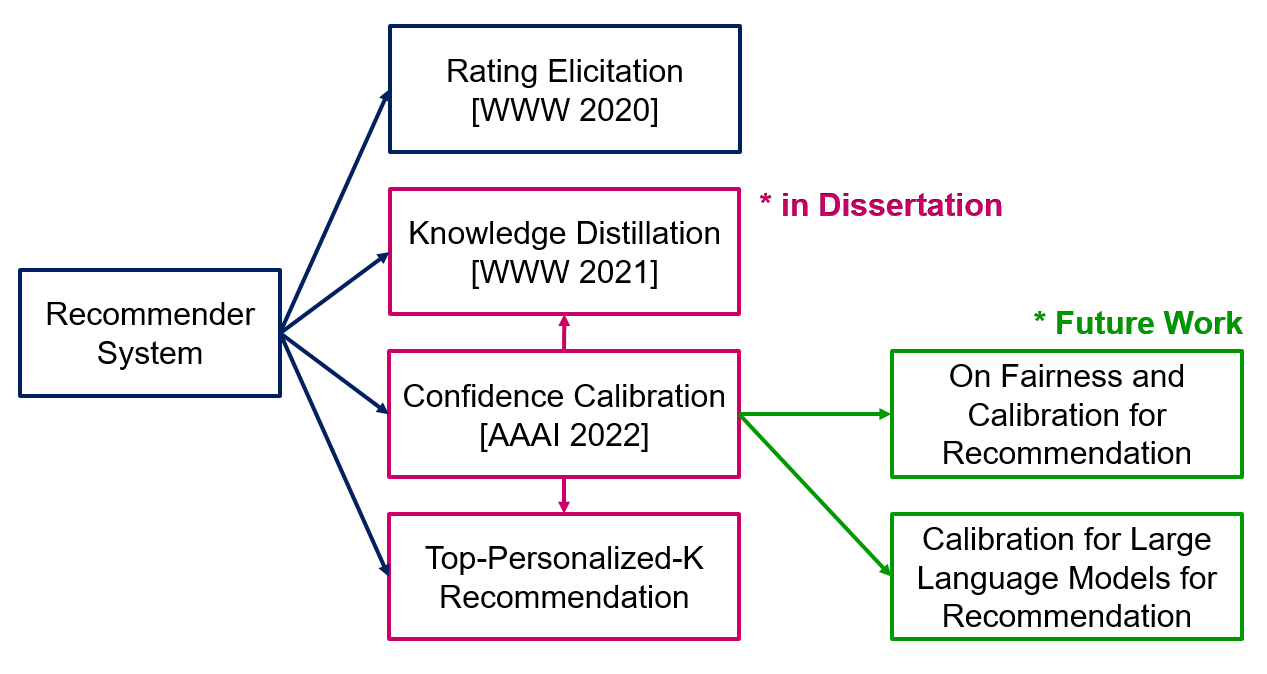} 
\caption{An overview of this dissertation.}
\label{conclusion}
\end{figure}

%% file: 6.abst_kor.tex
개인화된 추천은 쇼핑, 뉴스, 검색, 동영상, 음악 등 다양한 일상 활동에 큰 영향을 미친다.
대부분의 추천 시스템은 추천 결과에 대한 신뢰도를 전혀 제공하지 않고 사용자에게 상위 항목만 표시하고 있다.
하지만, 동일한 순위 위치의 의미는 사용자마다 다르다.
한 사용자는 30\%의 확률로 첫 번째 항목을 좋아할 수 있는 반면, 다른 사용자는 90\%의 확률로 첫 번째 항목을 좋아할 수 있다.
이 논문에서는 추천 시스템의 추천 결과에 대한 신뢰도를 추정할 수 있는 방법을 제시하고, 추정된 신뢰도를 활용하는 두 가지 실제 응용법을 소개한다.
구체적으로, (1) 사전 훈련된 추천 모델의 순위 점수를 잘 보정된 신뢰도로 매핑하는 가우스 보정과 감마 보정을 제안한다 (Section 2).
(2) 큰 교사 모델과 작은 학생 모델이 서로의 신뢰도를 추가적인 학습 레이블로 사용하여 학습하는 양방향 지식 증류 프레임워크를 제안한다 (Section 3).
(3) 보정된 신뢰도를 이용하여 추정된 예상 사용자 효용을 기반으로 추천 항목 수를 조정하는 새로운 추천 패러다임인 Top-Personalized-K Recommendation를 제안한다 (Section 4).
본 논문에서는 제안된 모델들의 성능을 실제 데이터셋에 대한 실험으로 평가하였으며, 여러 분석을 통해 그 타당성을 입증하였다.
본 논문에서 소개되는 추천시스템의 보정된 신뢰도는 최근 주목 받고 있는 추천시스템의 형평성과 추천시스템을 위한 초거대 언어모델에도 활용가능하며, 앞으로 이에대한 연구를 진행 할 예정이다.

%% file: 7.ack.tex
2015년에 포항에 와서, 9년 동안 공부하고 연구하다가, 박사 학위를 받게 되었습니다.
지난 5년 동안 지도해주신 유환조 교수님, 자유로운 연구 환경을 조성해 주시면서도 깊은 통찰을 말씀해 주실 때마다 많이 배웠습니다.
바쁘신 와중에도 박사학위 논문 심사를 맡아주신 이근배 교수님, 김동우 교수님, 이남훈 교수님, 고영명 교수님께도 감사의 말씀을 드립니다.
교수님들이 주신 피드백을 바탕으로 더 탄탄한 연구를 할 수 있었고, 앞으로도 좋은 연구를 하도록 노력하겠습니다.

학부 4년, 대학원 5년을 다시금 돌아보니 많은 분들에게 도움을 받으며, 배우며, 성장했던 것 같습니다.
모난 사람 없이 서로서로 잘 챙겼던 데이터 인텔리전스 연구실 가족들, 학부 때부터 같이 지낸 산경과 식구들, 축구보다 술을 더 같이한 기창/준혁/가종, 동기들보다 더 편하게 잘 지낸 분반 형님들, 신년회 때마다 보는 경남과학고 동기들, 등산메이트 성호와 홍준, 포항에서 지낸다고 자주 보지 못했던 제일중 친구들까지, 모두가 각자의 장점을 가지고 있어 하나씩 배우려고 노력했습니다.
다들 고맙습니다.
같이 만나 웃고 떠들며 대학원 생활을 잘 보낼 수 있었습니다.
마지막으로, 항상 응원해 주는 가족들에게도 감사의 인사를 전합니다.
앞으로 보답하려 노력하겠습니다.
제 곁에 머물렀던 모두가 바라는 일 이루기를 바라고, 시간이 지나 건강한 모습으로 다시 만났으면 좋겠습니다.

%% file: 8.cv.tex
\curriculumvitae[korean]
    \begin{personaldata}
        \name       {Wonbin Kweon}
        \email {wonbin.kweon@gmail.com}
    \end{personaldata}

    \begin{education}
        \item[2015. 3.\ --\ 2019. 2.] Department of Industrial and Management Engineering, \\ 
                                        Pohang University of Science and Technology (B.S.) \\

         \item[2019. 3.\ --\ 2024. 2.] Department of Convergence IT Engineering, \\
                                        Pohang University of Science and Technology (Ph.D.)
    \end{education}

   \begin{experience}
	\item[PC Members] AAAI’23-24, KDD’23, NeurIPS’23, ICLR’24, WWW’24, ICML'24
        \item[Awards] Fourth Prize, The 3rd POSTECH Research Award, 2023 \\
                        AAAI Student Scholarship, 2022 \\
                        Third Prize, Kakao Music Playlist Recommendation Competition, 2020 \\
                        WWW Student Scholarship, 2020
   \end{experience}

    \vspace*{10mm} \noindent           
    \begin{center}\large\textbf{Publications}\end{center}
    \vspace{3mm}
    \begin{enumerate}
        \item \textbf{Wonbin Kweon}, SeongKu Kang, Junyoung Hwang, Hwanjo Yu. Deep Rating Elicitation for New Users in Collaborative Filtering. WWW 2020.
        \item SeongKu Kang, Junyoung Hwang, \textbf{Wonbin Kweon}, Hwanjo Yu. DE-RRD: A Knowledge Distillation Framework for Recommender System. CIKM 2020.
        \item \textbf{Wonbin Kweon}, SeongKu Kang, Hwanjo Yu. Bidirectional Distillation for Top-K Recommender System. WWW 2021.
        \item SeongKu Kang, Junyoung Hwang, \textbf{Wonbin Kweon}, Hwanjo Yu. Topology Distillation for Recommender System. KDD 2021.
        \item SeongKu Kang, Dongha Lee, \textbf{Wonbin Kweon}, Junyoung Hwang, Hwanjo Yu. Consensus Learning from Heterogeneous Objectives for One-Class Collaborative Filtering. WWW 2022.
        \item \textbf{Wonbin Kweon}, SeongKu Kang, Hwanjo Yu. Obtaining Calibrated Probabilities with Personalized Ranking Models. AAAI 2022 (oral presentation).
        \item SeongKu Kang, \textbf{Wonbin Kweon}, Dongha Lee, Jianxun Lian, Xing Xie, Hwanjo Yu. Distillation from Heterogeneous Models for Top-K Recommendation. WWW 2023.
        \item \textbf{Wonbin Kweon}, SeongKu Kang, Sanghwan Jang, Hwanjo Yu. Top-Personalized-K Recommendation. WWW 2024.
        \item \textbf{Wonbin Kweon}, Hwanjo Yu. Doubly Calibrated Estimator for Recommendation on Data Missing Not At Random. WWW 2024.
    \end{enumerate}

    \afterpage{\blankpage}  

%% file: main.bbl
\begin{thebibliography}{100}

\bibitem{kweon22}
Wonbin Kweon, SeongKu Kang, and Hwanjo Yu.
\newblock Obtaining calibrated probabilities with personalized ranking models.
\newblock In {\em AAAI}, 2022.

\bibitem{bpr09}
Steffen Rendle, Christoph Freudenthaler, Zeno Gantner, and Lars Schmidt-Thieme.
\newblock Bpr: Bayesian personalized ranking from implicit feedback.
\newblock In {\em UAI}, 2009.

\bibitem{iso12}
Aditya~Krishna Menon, Xiaoqian~J Jiang, Shankar Vembu, Charles Elkan, and
  Lucila Ohno-Machado.
\newblock Predicting accurate probabilities with a ranking loss.
\newblock In {\em ICML}, 2012.

\bibitem{pruning09}
Avi Arampatzis, Jaap Kamps, and Stephen Robertson.
\newblock Where to stop reading a ranked list? threshold optimization using
  truncated score distributions.
\newblock In {\em SIGIR}, 2009.

\bibitem{cal17}
Chuan Guo, Geoff Pleiss, Yu~Sun, and Kilian~Q Weinberger.
\newblock On calibration of modern neural networks.
\newblock In {\em ICML}, 2017.

\bibitem{kull2019beyond}
Meelis Kull, Miquel Perello~Nieto, Markus K{\"a}ngsepp, Telmo Silva~Filho, Hao
  Song, and Peter Flach.
\newblock Beyond temperature scaling: Obtaining well-calibrated multi-class
  probabilities with dirichlet calibration.
\newblock In {\em NeurIPS}, 2019.

\bibitem{rahimi2020intra}
Amir Rahimi, Amirreza Shaban, Ching-An Cheng, Richard Hartley, and Byron Boots.
\newblock Intra order-preserving functions for calibration of multi-class
  neural networks.
\newblock In {\em NeurIPS}, 2020.

\bibitem{iso72}
Richard~E Barlow and Hugh~D Brunk.
\newblock The isotonic regression problem and its dual.
\newblock {\em Journal of the American Statistical Association},
  67(337):140--147, 1972.

\bibitem{platt99}
John Platt, Alexander Smola, Peter Bartlett, Bernhard Scholkopf, and Dale
  Schuurmans.
\newblock Probabilistic outputs for support vector machines and comparisons to
  regularized likelihood methods.
\newblock {\em Advances in large margin classifiers}, 10(3):61--74, 1999.

\bibitem{ips16}
Tobias Schnabel, Adith Swaminathan, Ashudeep Singh, Navin Chandak, and Thorsten
  Joachims.
\newblock Recommendations as treatments: Debiasing learning and evaluation.
\newblock In {\em ICML}, 2016.

\bibitem{saito19}
Yuta Saito.
\newblock Unbiased pairwise learning from implicit feedback.
\newblock In {\em NeurIPS 2019 Workshop on Causal Machine Learning}, 2019.

\bibitem{ips94}
James~M Robins, Andrea Rotnitzky, and Lue~Ping Zhao.
\newblock Estimation of regression coefficients when some regressors are not
  always observed.
\newblock {\em Journal of the American statistical Association},
  89(427):846--866, 1994.

\bibitem{margin07}
Markus Weimer, Alexandros Karatzoglou, Quoc Le, and Alex Smola.
\newblock Cofirank-maximum margin matrix factorization for collaborative
  ranking.
\newblock In {\em NeurIPS}, 2007.

\bibitem{beta17}
Meelis Kull, Telmo Silva~Filho, and Peter Flach.
\newblock Beta calibration: a well-founded and easily implemented improvement
  on logistic calibration for binary classifiers.
\newblock In {\em AISTATS}, 2017.

\bibitem{bbq15}
Mahdi~Pakdaman Naeini, Gregory Cooper, and Milos Hauskrecht.
\newblock Obtaining well calibrated probabilities using bayesian binning.
\newblock In {\em AAAI}, 2015.

\bibitem{hist01}
Bianca Zadrozny and Charles Elkan.
\newblock Obtaining calibrated probability estimates from decision trees and
  naive bayesian classifiers.
\newblock In {\em ICML}, 2001.

\bibitem{mukhoti2020calibrating}
Jishnu Mukhoti, Viveka Kulharia, Amartya Sanyal, Stuart Golodetz, Philip Torr,
  and Puneet Dokania.
\newblock Calibrating deep neural networks using focal loss.
\newblock In {\em NeurIPS}, 2020.

\bibitem{gamma99}
Christoph Baumgarten.
\newblock A probabilistic solution to the selection and fusion problem in
  distributed information retrieval.
\newblock In {\em SIGIR}, 1999.

\bibitem{exp69}
John~A Swets.
\newblock Effectiveness of information retrieval methods.
\newblock {\em American Documentation}, 20(1):72--89, 1969.

\bibitem{normexp01}
Raghavan Manmatha, Toni Rath, and Fangfang Feng.
\newblock Modeling score distributions for combining the outputs of search
  engines.
\newblock In {\em SIGIR}, 2001.

\bibitem{gammagaussian10}
Evangelos Kanoulas, Keshi Dai, Virgil Pavlu, and Javed~A Aslam.
\newblock Score distribution models: assumptions, intuition, and robustness to
  score manipulation.
\newblock In {\em SIGIR}, 2010.

\bibitem{scikit-learn}
F.~Pedregosa, G.~Varoquaux, A.~Gramfort, V.~Michel, B.~Thirion, O.~Grisel,
  M.~Blondel, P.~Prettenhofer, R.~Weiss, V.~Dubourg, J.~Vanderplas, A.~Passos,
  D.~Cournapeau, M.~Brucher, M.~Perrot, and E.~Duchesnay.
\newblock Scikit-learn: Machine learning in {P}ython.
\newblock {\em Journal of Machine Learning Research}, 12:2825--2830, 2011.

\bibitem{ips19}
Xiaojie Wang, Rui Zhang, Yu~Sun, and Jianzhong Qi.
\newblock Doubly robust joint learning for recommendation on data missing not
  at random.
\newblock In {\em ICML}, 2019.

\bibitem{joachims2017unbiased}
Thorsten Joachims, Adith Swaminathan, and Tobias Schnabel.
\newblock Unbiased learning-to-rank with biased feedback.
\newblock In {\em WSDM}, 2017.

\bibitem{rosenbaum2002overt}
Paul~R Rosenbaum.
\newblock Overt bias in observational studies.
\newblock In {\em Observational studies}, pages 71--104. Springer, 2002.

\bibitem{clip19}
Yi~Su, Lequn Wang, Michele Santacatterina, and Thorsten Joachims.
\newblock Cab: Continuous adaptive blending for policy evaluation and learning.
\newblock In {\em ICML}, 2019.

\bibitem{ncf17}
Xiangnan He, Lizi Liao, Hanwang Zhang, Liqiang Nie, Xia Hu, and Tat-Seng Chua.
\newblock Neural collaborative filtering.
\newblock In {\em WWW}, 2017.

\bibitem{cml17}
Cheng-Kang Hsieh, Longqi Yang, Yin Cui, Tsung-Yi Lin, Serge Belongie, and
  Deborah Estrin.
\newblock Collaborative metric learning.
\newblock In {\em WWW}, 2017.

\bibitem{lightgcn20}
Xiangnan He, Kuan Deng, Xiang Wang, Yan Li, Yongdong Zhang, and Meng Wang.
\newblock Lightgcn: Simplifying and powering graph convolution network for
  recommendation.
\newblock In {\em SIGIR}, 2020.

\bibitem{leave06}
Sean~M McNee, John Riedl, and Joseph~A Konstan.
\newblock Being accurate is not enough: how accuracy metrics have hurt
  recommender systems.
\newblock In {\em CHI'06 extended abstracts on Human factors in computing
  systems}, pages 1097--1101, 2006.

\bibitem{dre}
Wonbin Kweon, Seongku Kang, Junyoung Hwang, and Hwanjo Yu.
\newblock Deep rating elicitation for new users in collaborative filtering.
\newblock In {\em WWW}, 2020.

\bibitem{kweon2021bidirectional}
Wonbin Kweon, SeongKu Kang, and Hwanjo Yu.
\newblock Bidirectional distillation for top-k recommender system.
\newblock In {\em WWW}, 2021.

\bibitem{rd18}
Jiaxi Tang and Ke~Wang.
\newblock Ranking distillation: Learning compact ranking models with high
  performance for recommender system.
\newblock In {\em KDD}, 2018.

\bibitem{cd19}
J.~{Lee}, M.~{Choi}, J.~{Lee}, and H.~{Shim}.
\newblock Collaborative distillation for top-n recommendation.
\newblock In {\em ICDM}, 2019.

\bibitem{DERRD}
SeongKu Kang, Junyoung Hwang, Wonbin Kweon, and Hwanjo Yu.
\newblock De-rrd: A knowledge distillation framework for recommender system.
\newblock In {\em CIKM}, 2020.

\bibitem{CUL13}
Hao Wang, Binyi Chen, and Wu-Jun Li.
\newblock Collaborative topic regression with social regularization for tag
  recommendation.
\newblock In {\em IJCAI}, 2013.

\bibitem{FS14}
Dingqi Yang, Daqing Zhang, Vincent~W Zheng, and Zhiyong Yu.
\newblock Modeling user activity preference by leveraging user spatial temporal
  characteristics in lbsns.
\newblock {\em IEEE Transactions on Systems, Man, and Cybernetics: Systems},
  45(1):129--142, 2014.

\bibitem{resnet16}
Kaiming He, Xiangyu Zhang, Shaoqing Ren, and Jian Sun.
\newblock Deep residual learning for image recognition.
\newblock In {\em CVPR}, 2016.

\bibitem{cifar10}
Alex Krizhevsky, Geoffrey Hinton, et~al.
\newblock Learning multiple layers of features from tiny images.
\newblock 2009.

\bibitem{overthinking19}
Yigitcan Kaya, Sanghyun Hong, and Tudor Dumitras.
\newblock Shallow-deep networks: Understanding and mitigating network
  overthinking.
\newblock In {\em ICML}, 2019.

\bibitem{kd15}
Geoffrey Hinton, Oriol Vinyals, and Jeff Dean.
\newblock Distilling the knowledge in a neural network.
\newblock {\em arXiv preprint arXiv:1503.02531}, 2015.

\bibitem{yelp16}
Xiangnan He, Hanwang Zhang, Min-Yen Kan, and Tat-Seng Chua.
\newblock Fast matrix factorization for online recommendation with implicit
  feedback.
\newblock In {\em SIGIR}, 2016.

\bibitem{fullrank20}
Walid Krichene and Steffen Rendle.
\newblock On sampled metrics for item recommendation.
\newblock In {\em KDD}, 2020.

\bibitem{hr16}
Huayu Li, Richang Hong, Defu Lian, Zhiang Wu, Meng Wang, and Yong Ge.
\newblock A relaxed ranking-based factor model for recommender system from
  implicit feedback.
\newblock In {\em IJCAI}, 2016.

\bibitem{NDCG02}
Kalervo J{\"a}rvelin and Jaana Kek{\"a}l{\"a}inen.
\newblock Cumulated gain-based evaluation of ir techniques.
\newblock {\em ACM Transactions on Information Systems (TOIS)}, pages 422--446,
  2002.

\bibitem{cdae16}
Yao Wu, Christopher DuBois, Alice~X Zheng, and Martin Ester.
\newblock Collaborative denoising auto-encoders for top-n recommender systems.
\newblock In {\em WSDM}, 2016.

\bibitem{dae08}
Pascal Vincent, Hugo Larochelle, Yoshua Bengio, and Pierre-Antoine Manzagol.
\newblock Extracting and composing robust features with denoising autoencoders.
\newblock In {\em ICML}, 2008.

\bibitem{mf09}
Yehuda Koren, Robert Bell, and Chris Volinsky.
\newblock Matrix factorization techniques for recommender systems.
\newblock {\em Computer}, 42(8):30--37, 2009.

\bibitem{pytorch19}
Adam Paszke, Sam Gross, Francisco Massa, Adam Lerer, James Bradbury, Gregory
  Chanan, Trevor Killeen, Zeming Lin, Natalia Gimelshein, Luca Antiga, et~al.
\newblock Pytorch: An imperative style, high-performance deep learning library.
\newblock In {\em NeurIPS}, 2019.

\bibitem{adam14}
Diederik~P Kingma and Jimmy Ba.
\newblock Adam: A method for stochastic optimization.
\newblock In {\em ICLR}, 2015.

\bibitem{born18}
Tommaso Furlanello, Zachary~C Lipton, Michael Tschannen, Laurent Itti, and
  Anima Anandkumar.
\newblock Born again neural networks.
\newblock In {\em ICML}, 2018.

\bibitem{self19}
Linfeng Zhang, Jiebo Song, Anni Gao, Jingwei Chen, Chenglong Bao, and Kaisheng
  Ma.
\newblock Be your own teacher: Improve the performance of convolutional neural
  networks via self distillation.
\newblock In {\em CVPR}, 2019.

\bibitem{discreterec1}
Defu Lian, Rui Liu, Yong Ge, Kai Zheng, Xing Xie, and Longbing Cao.
\newblock Discrete content-aware matrix factorization.
\newblock In {\em KDD}, 2017.

\bibitem{discreterec2}
Hanwang Zhang, Fumin Shen, Wei Liu, Xiangnan He, Huanbo Luan, and Tat-Seng
  Chua.
\newblock Discrete collaborative filtering.
\newblock In {\em SIGIR}, 2016.

\bibitem{discreteAAAI}
Yan Zhang, Defu Lian, and Guowu Yang.
\newblock Discrete personalized ranking for fast collaborative filtering from
  implicit feedback.
\newblock In {\em AAAI}, 2017.

\bibitem{binary12}
Ke~Zhou and Hongyuan Zha.
\newblock Learning binary codes for collaborative filtering.
\newblock In {\em KDD}, 2012.

\bibitem{discreterec3}
Wang-Cheng Kang and Julian McAuley.
\newblock Candidate generation with binary codes for large-scale top-n
  recommendation.
\newblock In {\em CIKM}, 2019.

\bibitem{treeRS}
Yoram Bachrach, Yehuda Finkelstein, Ran Gilad-Bachrach, Liran Katzir, Noam
  Koenigstein, Nir Nice, and Ulrich Paquet.
\newblock Speeding up the xbox recommender system using a euclidean
  transformation for inner-product spaces.
\newblock In {\em RecSys}, 2014.

\bibitem{inference15}
Christina Teflioudi, Rainer Gemulla, and Olga Mykytiuk.
\newblock Lemp: Fast retrieval of large entries in a matrix product.
\newblock In {\em SIGMOD}, 2015.

\bibitem{inference17}
Hui Li, Tsz~Nam Chan, Man~Lung Yiu, and Nikos Mamoulis.
\newblock Fexipro: fast and exact inner product retrieval in recommender
  systems.
\newblock In {\em SIGMOD}, 2017.

\bibitem{dml18}
Ying Zhang, Tao Xiang, Timothy~M Hospedales, and Huchuan Lu.
\newblock Deep mutual learning.
\newblock In {\em CVPR}, 2018.

\bibitem{collaborativelearning18}
Guocong Song and Wei Chai.
\newblock Collaborative learning for deep neural networks.
\newblock In {\em NeurIPS}, 2018.

\bibitem{dual16}
Di~He, Yingce Xia, Tao Qin, Liwei Wang, Nenghai Yu, Tie-Yan Liu, and Wei-Ying
  Ma.
\newblock Dual learning for machine translation.
\newblock In {\em NeurIPS}, 2016.

\bibitem{topk10}
Paolo Cremonesi, Yehuda Koren, and Roberto Turrin.
\newblock Performance of recommender algorithms on top-n recommendation tasks.
\newblock In {\em RecSys}, 2010.

\bibitem{sgl21}
Jiancan Wu, Xiang Wang, Fuli Feng, Xiangnan He, Liang Chen, Jianxun Lian, and
  Xing Xie.
\newblock Self-supervised graph learning for recommendation.
\newblock In {\em SIGIR}, 2021.

\bibitem{topk22}
Yankai Chen, Huifeng Guo, Yingxue Zhang, Chen Ma, Ruiming Tang, Jingjie Li, and
  Irwin King.
\newblock Learning binarized graph representations with multi-faceted
  quantization reinforcement for top-k recommendation.
\newblock In {\em KDD}, 2022.

\bibitem{topkk22}
Lianghao Xia, Chao Huang, and Chuxu Zhang.
\newblock Self-supervised hypergraph transformer for recommender systems.
\newblock In {\em KDD}, 2022.

\bibitem{prp77}
Stephen~E Robertson.
\newblock The probability ranking principle in ir.
\newblock {\em Journal of documentation}, 33(4):209--304, 1977.

\bibitem{swing08}
Andrei Broder, Massimiliano Ciaramita, Marcus Fontoura, Evgeniy Gabrilovich,
  Vanja Josifovski, Donald Metzler, Vanessa Murdock, and Vassilis Plachouras.
\newblock To swing or not to swing: learning when (not) to advertise.
\newblock In {\em CIKM}, 2008.

\bibitem{revenue19}
Dietmar Jannach and Michael Jugovac.
\newblock Measuring the business value of recommender systems.
\newblock {\em ACM Transactions on Management Information Systems (TMIS)},
  10(4):1--23, 2019.

\bibitem{cdr12}
Jie Tang, Sen Wu, Jimeng Sun, and Hang Su.
\newblock Cross-domain collaboration recommendation.
\newblock In {\em KDD}, 2012.

\bibitem{prefetch16}
Stefan Wilk, Dominik Schreiber, Denny Stohr, and Wolfgang Effelsberg.
\newblock On the effectiveness of video prefetching relying on recommender
  systems for mobile devices.
\newblock In {\em 2016 13th IEEE Annual Consumer Communications \& Networking
  Conference (CCNC)}, pages 429--434. IEEE, 2016.

\bibitem{fair18}
Ashudeep Singh and Thorsten Joachims.
\newblock Fairness of exposure in rankings.
\newblock In {\em KDD}, 2018.

\bibitem{saito22fair}
Yuta Saito and Thorsten Joachims.
\newblock Fair ranking as fair division: Impact-based individual fairness in
  ranking.
\newblock In {\em KDD}, 2022.

\bibitem{vae18}
Dawen Liang, Rahul~G Krishnan, Matthew~D Hoffman, and Tony Jebara.
\newblock Variational autoencoders for collaborative filtering.
\newblock In {\em WWW}, 2018.

\bibitem{arampatzis2009stop}
Avi Arampatzis, Jaap Kamps, and Stephen Robertson.
\newblock Where to stop reading a ranked list? threshold optimization using
  truncated score distributions.
\newblock In {\em SIGIR}, 2009.

\bibitem{bahri2020choppy}
Dara Bahri, Yi~Tay, Che Zheng, Donald Metzler, and Andrew Tomkins.
\newblock Choppy: Cut transformer for ranked list truncation.
\newblock In {\em SIGIR}, 2020.

\bibitem{legal07}
Stephen Tomlinson, Douglas~W Oard, Jason~R Baron, and Paul Thompson.
\newblock Overview of the trec 2007 legal track.
\newblock In {\em TREC}, 2007.

\bibitem{legal22}
Yixiao Ma, Qingyao Ai, Yueyue Wu, Yunqiu Shao, Yiqun Liu, Min Zhang, and
  Shaoping Ma.
\newblock Incorporating retrieval information into the truncation of ranking
  lists for better legal search.
\newblock In {\em SIGIR}, 2022.

\bibitem{attncut21}
Chen Wu, Ruqing Zhang, Jiafeng Guo, Yixing Fan, Yanyan Lan, and Xueqi Cheng.
\newblock Learning to truncate ranked lists for information retrieval.
\newblock In {\em AAAI}, 2021.

\bibitem{mmoecut22}
Dong Wang, Jianxin Li, Tianchen Zhu, Haoyi Zhou, Qishan Zhu, Yuxin Wen, and
  Hongming Piao.
\newblock Mtcut: A multi-task framework for ranked list truncation.
\newblock In {\em WSDM}, 2022.

\bibitem{bicut19}
Yen-Chieh Lien, Daniel Cohen, and W~Bruce Croft.
\newblock An assumption-free approach to the dynamic truncation of ranked
  lists.
\newblock In {\em SIGIR}, 2019.

\bibitem{lstm97}
Sepp Hochreiter and J{\"u}rgen Schmidhuber.
\newblock Long short-term memory.
\newblock {\em Neural computation}, 9(8):1735--1780, 1997.

\bibitem{transformer17}
Ashish Vaswani, Noam Shazeer, Niki Parmar, Jakob Uszkoreit, Llion Jones,
  Aidan~N Gomez, {\L}ukasz Kaiser, and Illia Polosukhin.
\newblock Attention is all you need.
\newblock In {\em NeurIPS}, 2017.

\bibitem{qin2010letor}
Tao Qin, Tie-Yan Liu, Jun Xu, and Hang Li.
\newblock Letor: A benchmark collection for research on learning to rank for
  information retrieval.
\newblock {\em Information Retrieval}, 13(4):346--374, 2010.

\bibitem{msmarco16}
Tri Nguyen, Mir Rosenberg, Xia Song, Jianfeng Gao, Saurabh Tiwary, Rangan
  Majumder, and Li~Deng.
\newblock Ms marco: A human generated machine reading comprehension dataset.
\newblock In {\em CoCo@NeurIPS}, 2016.

\bibitem{tfidf03}
Juan Ramos et~al.
\newblock Using tf-idf to determine word relevance in document queries.
\newblock In {\em Proceedings of the first instructional conference on machine
  learning}, volume 242, pages 29--48. Citeseer, 2003.

\bibitem{doc2vec12}
Quoc Le and Tomas Mikolov.
\newblock Distributed representations of sentences and documents.
\newblock In {\em ICML}, 2014.

\bibitem{unavailable18}
Zhiyong Cheng, Ying Ding, Lei Zhu, and Mohan Kankanhalli.
\newblock Aspect-aware latent factor model: Rating prediction with ratings and
  reviews.
\newblock In {\em WWW}, 2018.

\bibitem{implicit08}
Yifan Hu, Yehuda Koren, and Chris Volinsky.
\newblock Collaborative filtering for implicit feedback datasets.
\newblock In {\em ICDM}, 2008.

\bibitem{saito20}
Yuta Saito, Suguru Yaginuma, Yuta Nishino, Hayato Sakata, and Kazuhide Nakata.
\newblock Unbiased recommender learning from missing-not-at-random implicit
  feedback.
\newblock In {\em WSDM}, 2020.

\bibitem{infonce18}
Aaron van~den Oord, Yazhe Li, and Oriol Vinyals.
\newblock Representation learning with contrastive predictive coding.
\newblock {\em arXiv preprint arXiv:1807.03748}, 2018.

\bibitem{nips22soft}
An~Zhang, Wenchang Ma, Xiang Wang, and Tat-Seng Chua.
\newblock Incorporating bias-aware margins into contrastive loss for
  collaborative filtering.
\newblock In {\em NeurIPS}, 2022.

\bibitem{positionbias08}
Nick Craswell, Onno Zoeter, Michael Taylor, and Bill Ramsey.
\newblock An experimental comparison of click position-bias models.
\newblock In {\em WWW}, 2008.

\bibitem{f111}
David~MW Powers.
\newblock Evaluation: from precision, recall and f-measure to roc,
  informedness, markedness and correlation.
\newblock {\em Journal of Machine Learning Technologies}, 2:37--63, 2011.

\bibitem{sigir22}
Enrique Amig{\'o}, Stefano Mizzaro, and Damiano Spina.
\newblock Ranking interruptus: When truncated rankings are better and how to
  measure that.
\newblock In {\em SIGIR}, 2022.

\bibitem{cdr17}
Tong Man, Huawei Shen, Xiaolong Jin, and Xueqi Cheng.
\newblock Cross-domain recommendation: An embedding and mapping approach.
\newblock In {\em IJCAI}, 2017.

\bibitem{sponsored11}
Song Yao and Carl~F Mela.
\newblock A dynamic model of sponsored search advertising.
\newblock {\em Marketing Science}, 30(3):447--468, 2011.

\bibitem{salkin1975knapsack}
Harvey~M Salkin and Cornelis~A De~Kluyver.
\newblock The knapsack problem: a survey.
\newblock {\em Naval Research Logistics Quarterly}, 22(1):127--144, 1975.

\bibitem{totalexpec08}
Patrick Billingsley.
\newblock {\em Probability and measure}.
\newblock John Wiley \& Sons, 2008.

\bibitem{poisonbinomial60}
Lucien Le~Cam.
\newblock An approximation theorem for the poisson binomial distribution.
\newblock {\em Pacific Journal of Mathematics}, 10(4):1181--1197, 1960.

\bibitem{tsrevisit21}
Matthias Minderer, Josip Djolonga, Rob Romijnders, Frances Hubis, Xiaohua Zhai,
  Neil Houlsby, Dustin Tran, and Mario Lucic.
\newblock Revisiting the calibration of modern neural networks.
\newblock In {\em NeurIPS}, 2021.

\bibitem{temperaturescaling21}
Zhipeng Ding, Xu~Han, Peirong Liu, and Marc Niethammer.
\newblock Local temperature scaling for probability calibration.
\newblock In {\em ICCV}, 2021.

\bibitem{tstransforment20}
Shrey Desai and Greg Durrett.
\newblock Calibration of pre-trained transformers.
\newblock In {\em EMNLP}, 2020.

\bibitem{amazon19}
Jianmo Ni, Jiacheng Li, and Julian McAuley.
\newblock Justifying recommendations using distantly-labeled reviews and
  fine-grained aspects.
\newblock In {\em EMNLP-IJCNLP}, 2019.

\bibitem{cdl15}
Hao Wang, Naiyan Wang, and Dit-Yan Yeung.
\newblock Collaborative deep learning for recommender systems.
\newblock In {\em KDD}, 2015.

\bibitem{cvae17}
Xiaopeng Li and James She.
\newblock Collaborative variational autoencoder for recommender systems.
\newblock In {\em KDD}, 2017.

\bibitem{mmoe18}
Jiaqi Ma, Zhe Zhao, Xinyang Yi, Jilin Chen, Lichan Hong, and Ed~H Chi.
\newblock Modeling task relationships in multi-task learning with multi-gate
  mixture-of-experts.
\newblock In {\em KDD}, 2018.

\bibitem{GCN}
Thomas~N Kipf and Max Welling.
\newblock Semi-supervised classification with graph convolutional networks.
\newblock In {\em ICLR}, 2017.

\bibitem{jain2015trends}
Sarika Jain, Anjali Grover, Praveen~Singh Thakur, and Sourabh~Kumar Choudhary.
\newblock Trends, problems and solutions of recommender system.
\newblock In {\em International conference on computing, communication \&
  automation}, pages 955--958. IEEE, 2015.

\end{thebibliography}
